\documentclass[a4paper,11pt,twoside]{article}

\usepackage[utf8]{inputenc}
\usepackage[T1]{fontenc}
\usepackage[english]{babel}

\usepackage{charter}

\usepackage{indentfirst}
\usepackage{xcolor}
\usepackage{verbatim}

\usepackage{hyperref}

\usepackage[top=3.cm, bottom=4.0cm, left=2.2cm, right=2.2cm]{geometry}
\usepackage{amsmath}
\usepackage{amsthm}	
\usepackage{amsfonts}	
\usepackage{amssymb	
           ,bbm
            ,units 
           }
\usepackage{enumerate}
\usepackage[nobysame
           ,alphabetic
           ,initials
           ]{amsrefs}

\usepackage{wrapfig}
\usepackage{graphicx}
\usepackage{epsfig,epstopdf}
\usepackage{subfigure}

\numberwithin{equation}{section}
\numberwithin{figure}{section}
 \usepackage[nodayofweek]{datetime}

\renewcommand*{\thefootnote}{\fnsymbol{footnote}}

\title{Equilibrium pricing under relative performance concerns}
\author{
\normalsize Jana Bielagk\footnote{J. Bielagk thankfully acknowledges partial support from \emph{Funding for the equality and promotion of women} of HU Berlin.} \\[8pt]
         \footnotesize  Humboldt-Universit\"at zu Berlin\\
         \footnotesize  Institut f\"ur Mathematik\\ 
         \footnotesize  Unter den Linden 6\\
         \footnotesize  10099 Berlin\\
         \footnotesize  Germany\\
          \\
        \small  bielagk@math.hu-berlin.de
 \and
\normalsize Arnaud Lionnet \\[8pt] 
        \footnotesize  E.N.S. Cachan\\ 
        \footnotesize  CMLA \\ 
        \footnotesize  61 av. du Pr\'esident Wilson \\
        \footnotesize  94235 Cachan cedex \\
		\footnotesize  France\\
          \\
        \small  Arnaud.Lionnet@cmla.ens-cachan.fr 
 \and
 	\normalsize Gon\c calo Dos Reis\footnote{G. dos Reis acknowledges support from the \emph{Funda{\c c}$\tilde{\text{a}}$o para a Ci$\hat{e}$ncia e a Tecnologia} (Portuguese Foundation for Science and Technology) through the project UID/MAT/00297/2013 (Centro de Matem\'atica e Aplica\c c$\tilde{\text{o}}$es CMA/FCT/UNL).} \\[8pt]
         \footnotesize  University of Edinburgh\\ 
         \footnotesize  School of Mathematics \\
         \footnotesize  Edinburgh, EH9 3FD, UK\\  
         \footnotesize  and \\
	\footnotesize  Centro Matem\'atica e Aplica\c c$\tilde{\text{o}}$es \\
	\footnotesize (CMA), FCT, UNL, Portugal \\
        \small  G.dosReis@ed.ac.uk
}

 \date{\today}

\theoremstyle{plain}
\newtheorem{theorem}{Theorem}[section]
\newtheorem{lemma}[theorem]{Lemma}
\newtheorem{proposition}[theorem]{Proposition}
\newtheorem{corollary}[theorem]{Corollary}

\newtheorem{definition}[theorem]{Definition}

\newtheorem{remark}[theorem]{Remark}
\newtheorem{example}[theorem]{Example}
\newtheorem{assumption}[theorem]{Assumption}

\newcommand{\bA}{\mathbb{A}}
\newcommand{\bD}{\mathbb{D}}
\newcommand{\bE}{\mathbb{E}}

\newcommand{\bL}{\mathbb{L}}
\newcommand{\bN}{\mathbb{N}}
\newcommand{\bP}{\mathbb{P}}
\newcommand{\bQ}{\mathbb{Q}}
\newcommand{\bR}{\mathbb{R}}

\newcommand{\cA}{\mathcal{A}}

\newcommand{\cE}{\mathcal{E}}
\newcommand{\cF}{\mathcal{F}}

\newcommand{\cH}{\mathcal{H}}

\newcommand{\cS}{\mathcal{S}}
\newcommand{\cT}{\mathcal{T}}

\newcommand{\cZ}{\mathcal{Z}}

\newcommand{\wh}{\widehat}
\newcommand{\wt}{\widetilde}
\newcommand{\tY}{\wt{Y}}  

\newcommand{\tZ}{\wt{Z}}

\newcommand{\ud}{\mathrm{d}}
\newcommand{\uds}{\mathrm{d}s}
\newcommand{\udt}{\mathrm{d}t}
\newcommand{\udu}{\mathrm{d}u}

\newcommand{\udw}{\mathrm{d}W}
\newcommand{\udws}{\mathrm{d}W_s}
\newcommand{\udwt}{\mathrm{d}W_t}
\newcommand{\udwr}{\mathrm{d}W_r}

\newcommand{\1}{\mathbbm{1}}
\newcommand{\lev}{\left\langle}
\newcommand{\rev}{\right\rangle}

\newcommand{\sgn}{\text{sgn}}

\newcommand{\abs}[1]{\lvert#1\rvert }																		
\newcommand{\scalar}[2]{\left\langle #1 , #2 \right\rangle}										
\renewcommand{\matrix}[1]{\begin{bmatrix}#1\end{bmatrix}}									

\newcommand{\cHBMO}{{\cH_{\text{BMO}}}}

\begin{document}

\selectlanguage{english}

\maketitle
\renewcommand*{\thefootnote}{\arabic{footnote}}

\begin{abstract} 
We investigate the effects of the social interactions of a finite set of agents on an equilibrium pricing mechanism. A derivative written on non-tradable underlyings is introduced to the market and priced in an equilibrium framework by agents who assess risk using convex dynamic risk measures expressed by Backward Stochastic Differential Equations (BSDE). Each agent is not only exposed to financial and non-financial risk factors, but she also faces performance concerns with respect to the other agents. 

Within our proposed model we prove the existence and uniqueness of an equilibrium whose analysis involves systems of fully coupled multi-dimensional quadratic BSDEs. We extend the theory of the representative agent by showing that a non-standard aggregation of risk measures is possible via weighted-dilated infimal convolution. We analyze the impact of the problem's parameters on the pricing mechanism, in particular how the agents' performance concern rates affect prices and risk perceptions. In extreme situations, we find that the concern rates destroy the equilibrium while the risk measures themselves remain stable.
\end{abstract}
\vspace{0.3cm}

\vspace{0.3cm}


%
%
%
\newpage 

\section{Introduction}

The importance of relative concerns in human behavior has been emphasized both in economic and sociological studies; making a $1$ EUR profit when everyone else made $2$ EUR ``feels'' distinctly different had everyone else lost $2$ EUR. A diverse literature handling problems dealing with some form of strategic and/or social interaction in the form of relative performance concerns exists: In both \cite{HinzRudolphAntolinEtAl2010} and \cite{Pedraza2013} the social interaction component appears in the form of peer-based under-performance penalties known as ``Minimum Return Guarantees''; the comparison is usually done via tracking a relevant market index, something quite standard in pension fund management. Another type of performance concerns arises in problems where the agents' consumption is taken into account. The utility functions used there exhibit a ``keeping up with the Joneses'' behavior as introduced in \cite{Duesenberry1949} and developed by \cite{Abel1990}, \cite{Abel1999} (see further \cite{Gali1994}, \cite{ChanKogan2001}, \cite{Gomez2007}  and \cite{XiourosZapatero2009}); in other words, the benchmark for the standard of living is the averaged consumption of the population and one computes the individual's consumption preferences in relation to that benchmark. Another type of concern criterion, an internal one, uses the past consumption of the agent as a benchmark for the current consumption; \cite{HarlHeal1973} introduced this ``habit formation'' approach. A more mathematical finance approach, as well as a literature overview, can be found in \cite{EspinosaTouzi2015} or \cite{FreiDosReis2011}. These last two papers are the inspiration for this one. 

In this paper we study the effects of social interaction between economic agents on a market equilibrium, the efficiency of a securitization mechanism and the global risk.
We consider a finite set $\bA$ of $N$ agents having access to an incomplete market consisting of an exogenously priced liquidly traded financial asset. The incompleteness stems from a non-tradable external risk factor, such as the amount of rain or the temperature, to which those agents are exposed. In an attempt to reduce the individual and overall market risks, a social planner introduces to the market a derivative written on the external risk source, allowing the agents in $\bA$ to reduce their exposures by trading on it. The question of the actual completeness of the resulting market has been addressed in some generality in the literature, and we refer for instance to \cite{Schwarz2015}. Questions about pricing and benefits of such securities written on non-tradable assets have been approached in the literature many times, we refer in particular to \cite{HorstPirvuDosReis2010} where the new derivative is priced within an equilibrium framework according to supply and demand rules and more generally to \cite{AndersonRaimondo2008}. Equilibrium analysis of incomplete markets is commonly confined to certain cases such as single agent models (\cite{HeLeland1993}, \cite{GarleanuPedersenPoteshman2009}), multiple agent models where markets are complete in equilibrium (\cite{DuffieHuang1985}, \cite{HorstPirvuDosReis2010}, \cite{KaratzasLehoczkyShreve1990}), or models with particular classes of goods (\cite{JofreRockafellarWets2010}) or preferences (\cite{CarmonaFehrHinzEtAl2010}). For a good overview of equilibrium issues stemming the from market incompleteness as well as some solutions we point the reader to \cite{KardarasXingZitkovic2015} and references therein. 

Although we follow ideas similar to those in \cite{HorstPirvuDosReis2010}, our goal is to understand how such a pricing mechanism and risk assessments are affected when the agents have relative performance concerns with respect to each other. Each agent $a\in \bA$ has an endowment $H^a$ over the time period $[0,T]$  depending on both risk factors. 
Her investment strategy $\pi^a$ in stock 
and the newly introduced derivative induces a gains process $(V^a_t(\pi^a))_{t\in[0,T]}$. 
For a given \emph{performance concern rate} $\lambda^a\in[0,1]$ the agent seeks to minimize the risk
\begin{align}
\label{intro-perffunc} 
 \rho^a\bigg(
     H^a + \big(1-\lambda^a\big) V_T^{a}(\pi^a) 
+\lambda^a \Big(
	   V_T^{a}(\pi^a)
			- \frac{1}{N -1}\sum_{b \in \bA \setminus \{a\}} V_T^{b}(\pi^b)
          \Big)
			\bigg),
\end{align}
where $\rho^a$ is a risk measure ($\rho^a$ is further described below). The first two terms inside $\rho^a$ correspond to the classical situation of an isolated agent $a$ trading optimally in the market to profit from market movements and to hedge the financial risks inherent to $H^a$. The last term is the \emph{relative performance concern} and corresponds to the difference between her own trading gains and the average trading gains achieved by her peers. Intuitively, as $\lambda^a\in[0,1]$ increases, the agent is less concerned with the risks associated to her endowment $H^a$ and more concerned with how she fares against the average performance of the other agents in $\bA$. For instance, given no endowments, if $\lambda^a=1/2$ and agent $a$ made $1$ EUR from trading while the others all made $2$ EUR then she perceives no gain at all.

Each agent $a\in \bA$ uses a monetary convex risk measure $\rho^a$. The theory of monetary, possibly convex, possibly coherent, risk measures was initiated by \cite{ArtznerDelbaenEberEtAl1999} and later extended by \cite{FoellmerSchied2002} and \cite{FrittelliRosazza2002}. One special class of risk measures, the so-called \emph{$g$-conditional risk measures}, which are closely related to the so-called \emph{$g$-conditional expectations} (see \cite{Gianin2006}), are those defined through Backward Stochastic Differential Equations (BSDEs), see \cite{Peng1997}, \cite{ElKarouiRavanelli2009} and \cite{BarrieuElKaroui2009}. Our use of BSDEs is motivated by two general aspects. The first is that it generically allows to solve stochastic control problems away from the usual Markovian setup where one uses the HJB approach in combination with PDE theory, see e.g. \cite{Touzi2013}. The second is that optimization can be carried out in closed sets of constraints without the assumption of convexity for which one usually uses duality theory, see \cite{HuImkellerMuller2005}.

The form of relative performance concern we use and its study using BSDEs can be traced back to \cite{Espinosa2010} and \cite{EspinosaTouzi2015}. Their setting is quite different from that presented here; the authors show the existence of a Nash equilibrium for a pure-interaction game of optimal investment without idiosyncratic endowments to hedge ($H^a=0$ for all agents), and where the agents optimize, in a Black--Scholes stock market, the expected utility of the gains they make from trading under individual constraints. These two works are followed by \cite{FreiDosReis2011}, where a general discussion on the existence of equilibrium, with endowments, is given including counter examples to such existence.

\textbf{Methodology and content of the paper.} 
All agents optimize their respective  functional given by \eqref{intro-perffunc} and, as the derivative is priced endogenously via an equilibrium framework, the market price of external risk is also part of the problem's solution. Equilibrium in our game is a set of investment strategies and a market price of external risk giving rise to a certain martingale measure. 

In the first part of this work, we show the existence of the Nash equilibrium in our problem and how to compute it for general risk measures induced by BSDEs.
The analysis is carried out in two steps. The first involves solving the individual optimization problem for each agent given the other agents' actions, the so-called best response problem. The second consists in showing that it is possible to find all best responses simultaneously in such a way that supply and demand for the derivative match, which, in turn, yields the market price of external risk. (We generally think of a market with a zero net supply of the derivative; however, the methodology allows to treat cases where some agents who were allowed to trade in the derivative left the market such that the (active) agents in $\bA$ hold together a non-zero position). We then verify that the market price of external risk associated to the best responses satisfies the necessary conditions.

This last step is more complex. Since the agents assess their risks using  dynamic risk measures given by BSDEs, the general equilibrium analysis leads to a system of fully coupled non-linear multi-dimensional BSDEs (possibly quadratic). We proceed by extending the representative agent approach (see \cite{Negishi1960}) where aggregation of the agents into a single economy and optimal Pareto risk sharing are equivalent to simultaneous individual optimization. From the works of \cite{BarrieuElKaroui2005} and \cite{BarrieuElKaroui2009}, we extend their infimal convolution (short inf-convolution) technique to agents with interdependent utility functions. This approach yields a single risk measure that encompasses the risk preferences for each agent and through which it is possible to find a single representative economy. We point out that in order to cope with the cross dependence induced by the performance concern rates, standard inf-convolution techniques do not lead to a single representative economy. We use the technique in a non-standard fashion via weighted-dilations of each agent's risk measure $\rho^a$ (see Section 3.2 in \cite{BarrieuElKaroui2009} and Section \ref{sec4-repagent} below); to the best of our knowledge this type of analysis is new and of independent interest. The closest reference to this is \cite{Rueschendorf2013}, where some form of weighted inf-convolution appears.

The second part of this work focuses on the case of agents using entropic risk measures, which can be treated more explicitly and allows for an in-depth study of the impact of the concern rates. In identifying the Nash equilibrium, we are led to a system of fully coupled  multi-dimensional quadratic BSDEs whose analysis is, in general, quite involved. Based on the works of \cite{Espinosa2010} and \cite{EspinosaTouzi2015}, the authors of \cite{FreiDosReis2011} give several counter examples to the existence of solutions; nonetheless, positive results do exist although none are very general, see e.g. \cite{Tevzadze2008}, \cite{CheriditoNam2015}, \cite{Frei2014}, \cite{KramkovPulido2016}, \cite{HuTang2016}, \cite{XingZitkovic2016}, \cite{JamneshanKupperLuo2014} and \cite{luo2015solvability}. In our case we are able to solve the system.

\textbf{Findings.} 
Within the case of entropic risk measures, we study in detail a model of two agents $a$ and $b$ with opposite exposures to the external risk factor, so that one has incentives to buy the derivative while the other has incentive to sell. In this model we are able to specify the structure of the equilibrium. Using both analytical methods and numerical computations, when the analytics are not tractable, we explore the behavior of the agents as the model parameters vary. We give particular attention to how the relative performance concern rates, and thereby the strength of the coupling between the agents, deviate from the standard case of non-interacting agents (when $\lambda^a=\lambda^b=0$).

We find that as either agent's risk tolerance increases, their risk lowers.  If any concern rate $\lambda$ increases then the agents engage in less trading of the derivative. This is because every unit of derivative bought by one is a unit sold by the other and hence the gains of one are the losses of the other. Consequently, if an agent is more concerned about the relative performance, she will tend to trade less with the others.  Also, as expected, we find that if it is the buyer of derivatives whose concern rate increases, the derivative's price decreases, while it increases in the case of the seller.

Very interestingly, we find that the risk of a single agent increases if the other agents become more concerned with their relative performance but that it \emph{decreases} as this agent becomes more concerned. Consequently, if the agents were to play this game repeatedly and their concern rate were to vary over time, they would both find it more advantageous to become more concerned (or jealous). As they both do so, the trading activity in the derivative decreases, but their activity in the stock increases and explodes - the equilibrium does not exist anymore. Surprisingly, this behavior is \emph{not captured at all by the risk measures}! 
This non-trivial, and perhaps not desirable, behavior of the system after introduction of the derivative is not without similarities with what is found in the models of \cite{CaccioliMarsiliVivo2009} and \cite{CorsiMarmiLillo2013}. It is a reminder that, when evaluating the benefits of financial innovation, one should not focus of the economy of an individual agent (who sees clear benefits in the form of a risk reduction) but really have a systemic view of the impact of the new instrument.

\textbf{Organization of the paper.} 
In Section 2, we define the general market, agents, optimization problem and equilibrium that we consider. Sections 3 and 4 are devoted to solving the general optimization problem for a set of agents having arbitrary risk measures. In the former we solve the optimization problem for each agent, given the strategies of all others.  In the latter we deal with the aggregation of individual risk measures and identification of the representative agent, and we solves the equilibrium for the whole system. Sections 5 and 6 contain the particular case where the agents use entropic risk measures. In this more tractable setting, Section 5 explores theoretically the influence of various parameters on the global risk while Section 6 focuses on a model with 2 agents with opposite risk profiles, and thoroughly explores the influence of the concern rates, in particular, on the individual behaviors, risks, and the consequences for the whole system. We also present numerical results. Section 7 concludes the study.

\textbf{Acknowledgments.} 
The authors would like to thank the two anonymous referees for their to-the-point and informative reports that have allowed us to improve several aspects of this work.

%
%
%
\section{The model}
\label{sec:model}

We consider a finite set $\bA$ of $N$ agents, without loss of generality $\bA = \left\{1,2,\ldots,N\right\}$, with random endowments $H^a$, $a \in \bA$, to be received at a terminal time $T < \infty$. 
They trade continuously in the financial market which comprises a {stock} and a newly introduced {structured security} (called derivative), aiming to minimize their risk. For simplicity we assume that money can be lent or borrowed at the risk-free rate zero. Stock prices follow an \textsl{exogenous} diffusion process and are not affected by the agents' demand. 
By contrast, the derivative is traded only by the agents from $\bA$ and priced \emph{endogenously} such that demand matches supply. 
We point the reader to Appendix \ref{appendix-NotionStochAna} for a full overview of the notation and stochastic setup.

\subsection{The market}

\subsubsection*{Sources of risk and underlyings} 

Throughout this paper $t\in[0,T]$. In our model, there are two independent sources of randomness, represented by a $2$-dimensional standard Brownian motion $W=(W^S,W^R)$ on a standard filtered probability space $(\Omega,(\cF_t)_{t=0}^T, \bP)$, where $(\cF_t)$ is the filtration generated by $W$ and augmented by the $\bP$-null sets. The Brownian motion $W^R$ drives the external and non-tradable risk process $(R_t)$, which is thought of as a temperature process or a precipitation index. For analytical convenience we assume that $(R_t)$ follows a Brownian motion with drift being a stochastic process $\mu^R :\Omega \times [0,T]\to \bR$ and constant volatility $b > 0$, i.e.,
\begin{align} 
\label{diffusion-R}
    \ud R_t = \mu^R_t \udt + b \ud W^R_t, \qquad\text{with } R_0=r_0 \in \bR.
\end{align}
The Brownian motion $W^S$ drives the stock price process $(S_t)$ according to
\begin{align} 
\label{diffusion-S}
    \ud S_t 
		&= \mu^S_t S_t \ \udt + \sigma^S_t S_t \ \udw^S_t, 
    \\ \nonumber
	    &= \mu^S_t S_t \ \udt + \lev \sigma_t, \udwt\rev
    	\quad \text{with } \sigma_t := (\sigma^{S}_t S_t, 0) \in \bR^2 , \text{ and } S_0 =s_0 > 0.
\end{align}
We assume throughout this work that the stochastic processes $\mu^R, \mu^S, \sigma^S: \Omega \times [0,T] \to \bR$ are $(\cF_t)$-adapted, with $\sigma^S>0$.

\subsubsection*{Market price of risk: financial and external}  

We recall (see e.g. \cite{HorstMuller2007}) that any linear pricing scheme on the set $L^2(\bP)$ of square integrable random variables with respect to $\bP$
can be identified with a $2$-dimensional predictable process $\theta$ such that the exponential process $(\cE^\theta_t)$ defined by
\begin{equation} 
\label{density-process}
    \cE^\theta_t 
		:= \cE\left( - \int_0^\cdot \langle \theta_s ,\ud W_s\rangle \right)_t
		 =\exp\left\{ - \int_0^t \langle\theta_s ,\ud W_s\rangle -
    \frac{1}{2} \int_0^t {|\theta_s|^2} \uds \right\},\quad t\in[0,T],
\end{equation}
is a uniformly integrable martingale. 
This ensures that the measure $\bP^\theta$ defined by having density $\cE^\theta_T$ against $\bP$ is indeed a probability measure (the \emph{pricing measure}), and the present price of a random terminal payment $X$ is then given by $\bE^{\theta}[X]$, where $\bE^\theta$ denotes the expectation with respect to $\bP^\theta$.
For any such $\theta$, we introduce the $\bP^\theta$-Brownian motion
\begin{align*}
W^\theta_t = W_t + \int_0^t \theta_s \,\uds, \qquad t\in[0,T].  
\end{align*}

The first component $\theta^S$ of the vector $\theta:=(\theta^S,\theta^R)$ is the \textsl{market price of financial risk}. Under the assumption that there is no arbitrage, $S$ must be a martingale under $\bP^\theta$ and, from the exogenously given dynamics of $S$, $\theta^S$ is necessarily given by $\theta^S_{t}={ \mu^S_t}/{\sigma^S_t}$.
The process $\theta^R$ on the other hand {is unknown}. It is the \textsl{market price of external risk} and will be derived {endogenously} by the market clearing condition (or constant net supply condition, see below).

\subsubsection*{The agents' endowments and the derivative's payoffs}

The agents $a \in \bA$ receive at time $T$ the income $H^a$ which depends on the financial and external risk factors. While the agents are able to trade in the financial market to hedge away some of their financial risk, a basis risk remains originating in the agent's exposure to the non-tradable risk process $R$. A derivative with payoff $H^D$ at maturity time $T$ 
is externally introduced in the market. By trading in the derivative $H^D$, the agents have now a way to reduce their basis risk. 

We give general conditions on the endowments, derivative payoff and coefficients appearing in the dynamics of $S$ and $R$ (see Appendix \ref{appendix-NotionStochAna} for notation).
Throughout the rest of this work {the following assumption stands} for all results. 

\begin{assumption}[Standing assumption on the data of the problem]
\label{assump-StandardData}
The processes $\mu^R$, $\mu^S$, $\sigma^S$ and  $\theta^S := {\mu^S}/{\sigma^S}$ are bounded (belong to $\cS^\infty$). The random variables $H^D$ and $H^a$, $a\in \bA$, are bounded (belong to $L^\infty(\cF_T)$).
\end{assumption}


\subsubsection*{Price of the derivative, trading in the market and the agent's strategies} 

Assuming no arbitrage opportunities, the price process $(B_t^{\theta})$ of $H^D$ is given by its expected payoff under $\bP^{\theta}$; in other words
    $B^\theta_\cdot  = \bE^\theta \left[H^D | \cF_\cdot \right]$.
Since $H^D$ is bounded, writing the $\bP^\theta$-martingale as a stochastic integral against the $\bP^\theta$-Brownian motion $W^\theta$ (with the martingale representation theorem) yields a 2-dimensional square-integrable adapted process  $\kappa^\theta :=(\kappa^{S},\kappa^{R})$ such that for $t \in [0,T]$ 
\begin{align} 
\label{Bond-Price-MRT-decomposition} 
    B^\theta_t & 
    = \bE^\theta[H^D] 
        + \int_0^t \langle\kappa_s^\theta, \udws^\theta\rangle
   = \bE^\theta[H^D] 
		+ \int_0^t \langle \kappa_s^\theta , \udws  \rangle
		+ \int_0^t \lev \kappa^\theta_s, \theta_s \rev \uds.
\end{align}
Note that we have $(B,\kappa)\in\cS^\infty\times \cHBMO(\bP^\theta)$.
We denote by $\pi^{a,1}_t$ and $\pi^{a,2}_t$ the number of units agent $a \in \bA$ holds in the stock and the derivative at time $t \in [0,T]$, respectively. 
Using a self-financing strategy $\pi^a:=(\pi^{a,1}, \pi^{a,2})$ valued in $\bR^2$, her gains from trading up to time $t \in [0,T]$, under the pricing measure $\bP^\theta$ inducing the prices $(B_t^{\theta})$ for the derivative, are given by 
\begin{align*} 
	V^a_t =  V_t(\pi^a) 
   	&=\int_0^t \pi^{a,1}_s \; \ud S_s+\int_0^t \pi^{a,2}_s \; \ud B^\theta_s  
		\\
		&= \int_0^t  \lev \pi^{a,1}_s \sigma_s + \pi^{a,2}_s \kappa_s^{\theta}, \theta_s  \rev \; \uds
				+ \int_0^t \langle\pi_s^{a,1} \sigma_s + \pi_s^{a,2} \kappa^{\theta}_s , \ud W_s \rangle.
\end{align*}
We require that the trading strategies be integrable against the prices, $\pi^a \in L^2\big((S,B^\theta), \bP^\theta \big)$ (i.e. $\bE^\theta  \left[\lev \; V_\cdot(\pi^a) \; \rev_{T}\right]< \infty $), so that the gains process are square-integrable martingales under $\bP^\theta$.

\subsection{Preferences, risk minimization and equilibrium}

\subsubsection*{The agents' measure of risk}

The agents assess their risk using a dynamic convex time-consistent risk measure $\rho^a_{\cdot}$ induced by a Backward Stochastic Differential Equation (BSDE). This means that the risk $\rho^a_t(\xi^a)$ which agent $a \in \bA$ associates at time $t\in[0,T]$ with an $\cF_T$-measurable random position $\xi^a$ is given by $Y^a_t$, where $(Y^a,Z^a)$ is the solution to the BSDE
\begin{equation*}
		-\ud Y^a_{t}=g^{a}(t,Z^a_{t})\udt-\langle Z^a_{t},\ud W_{t}\rangle 
    \quad \mbox{with terminal condition} \quad 
		Y^a_{T}=-\xi^a.
\end{equation*}
The driver $g^a$ encodes the risk preferences of the agent. 
We assume that $g^a$ has the following properties:
\begin{assumption}
\label{ass:driver}
The map $g^a:[0,T]\times\bR^2\to \bR$ is a deterministic continuous function. Its restriction to the space variable, $z \mapsto g^a(\cdot,z)$, is continuously differentiable, strictly convex and attains its minimum.
\end{assumption}
For any fixed $(t,\vartheta) \in [0,T] \times \bR^2$, the map $z \mapsto g^a(t,z)-\langle z,\vartheta \rangle$ is also strictly convex and attains its unique minimum at the point where its gradient vanishes.
With this in mind we can define $\cZ^a : [0,T] \times \bR^2 \rightarrow \bR^2$, $(t,\vartheta) \mapsto \cZ^a(t,\vartheta)$ where $\cZ^a(t,\vartheta)$ is the unique solution, in the unknown $\cZ$, to the equation
\begin{align} 
	\label{eq:caligraphicZ} 
		\nabla_z g^a(t,\cZ) = \vartheta.
	\end{align}

The agents' risk measure given by the above BSDE is strongly time consistent, convex and translation invariant (or monetary). We do not give many details on the class of risk measures described by BSDEs, instead, we point the interested reader to \cites{BarrieuElKaroui2005,Gianin2006,BarrieuElKaroui2009}. 

For convenience, we recall the relevant properties of the risk measures that play a role in this work:  i) \emph{translation invariance}: for any $m\in\bR$ and any $t \in[0,T]$ it holds that $\rho^a_t(\xi^a+m)=\rho^a_t(\xi^a)-m$; ii) \emph{time-consistency} of the process $\big(\rho^a_t(\xi^a)\big)$: for any $t,t+s\in[0,T]$ it holds that $\rho^a_t(\xi^a)=\rho^a_t(\rho^a_{t+s}(\xi^a))$; and iii) \emph{convexity}: for any $t \in[0,T]$ and for $\xi^a,\hat\xi^a$ $\cF_T$-measurable and $\alpha\in[0,1]$ we have $\rho^a_t\big(\alpha\xi^a+(1-\alpha)\hat\xi^a\big)\leq \alpha\rho^a_t(\xi^a)+(1-\alpha)\rho^a_t(\hat\xi^a)$.

\subsubsection*{The individual optimization problem}

Agent $a$'s position $\xi^a$ at maturity is given by the sum of her terminal income $H^a$ and the trading gains $V^a_T$ over the time period $[0,T]$. However, the agent compares her trading gains $V^a_T = V_T(\pi^a)$ with the average  gains of all other agents. Thus, we define the perceived total wealth $\xi^a(\pi^a,\pi^{-a})$ of each of the $N$ agents $a\in\bA$ in the market at time $t=T$ as
\begin{align*}  
\xi^a 
	&= \bigg( H^a + \big(1-\lambda^a\big) V_T(\pi^a) \bigg) + \lambda^a \bigg( V_T(\pi^a) - \frac{1}{N -1}\sum_{b \in \bA \setminus \{a\}} V_T(\pi^b) \bigg)
\\
&= H^a + V_T(\pi^a) - \wt{\lambda}^a \sum_{b \in \bA \setminus \{a\}} V_T(\pi^b),
\quad\text{where}\quad
		\wt{\lambda}^a := \frac{\lambda^a}{N -1}
\end{align*} 
and $\lambda^a \in [0,1]$ is the \emph{concern rate} (or \emph{jealousy factor}) of agent $a \in \bA$ (compare with \eqref{intro-perffunc}).

We make the following assumption on the concern rates $\lambda^{\cdot}$, whose justification will become clear later on in Section \ref{example-general-lambdas}.
\begin{assumption}[Performance concern rates\footnotemark]
\label{assump:ConcernRates}
We have $\lambda^a \in \left[0,1\right]$ for each agent 
and $\prod_{a \in \bA} \lambda^a < 1$.
\end{assumption}
\footnotetext{The case of $\lambda^a>1$ is, as we show in Section \ref{example-general-lambdas}, not necessarily intractable, but the analysis of such a situation is beyond the scope of this paper.}

For notational convenience we introduce the $\big(\bR^{2}\big)^{(N-1)}$-valued vector $\pi^{-a} := (\pi^b)_{b \neq a}$ and 
\begin{align}
\label{eq:averageVaminusa}
	\bar{V}_t^{-a} 
		:= 
		   \sum_{b \in \bA \setminus \{a\}} V_t(\pi^b) 
		 = 
		V_t(\bar{\pi}^{-a})
	\quad\text{where}\quad
		\bar{\pi}^{-a} := \sum _{b \in \bA \setminus \{a\}} \pi^b .
\end{align}
The risk associated with the self-financing strategy $\pi^a$ evolves according to the BSDE 
\begin{align} 
\label{BSDE_1} 
		-\ud Y^{a}_{t} 
		= 
		g^{a}\left(t,Z^{a}_{t}\right) \udt  -\langle Z^{a}_{t},\udwt\rangle 
    \quad\text{and}\quad
		Y^{a}_{T} 
		= -\xi^a(\pi^a,\pi^{-a}) =
		-\Big(H^{a} +V_T(\pi^a)  - \wt{\lambda}^a  V_T(\bar{\pi}^{-a})\Big). 
\end{align} 
Now, we introduce a notion of admissibility for our problem. 
\begin{definition}[Admissibility] 
Let $a\in\bA$, and $\pi^{-a}=(\pi^b)_{b \in \bA\backslash\{a\}}$ be integrable strategies for the other agents. 
The $\bR^2$-valued strategy process $\pi^a$ is called \emph{admissible} with respect to the market price of risk $\theta$ if $\bE^{\theta}  \left[\lev \; V_\cdot(\pi^a) \; \rev_{T}\right]< \infty$, where $\lev  V_\cdot(\pi^a) \rev$ denotes the quadratic variation of $\big(V_t(\pi^a)\big)_{t \in [0,T]}$ and BSDE \eqref{BSDE_1} has a unique solution.
The set of admissible trading strategies for agent $a\in\bA$ is denoted by $\cA^{\theta}(\pi^{-a})$.
\end{definition} 

We point out that in full generality each agent could have her own trading constraints, as in \cite{EspinosaTouzi2015} or \cite{FreiDosReis2011}. Here we assume that the agents have no trading constraints, aside from their strategies being integrable against the prices and leading to well-defined risk.

Each agent $a\in\bA$ solves the following risk-minimization problem 
\begin{align*}
	\min\{ {Y}^a_0(\pi^a,\pi^{-a}) \ |  \ \pi^a\in \cA^\theta(\pi^{-a}) \}. 
\end{align*}

Notice that, a priori, the risk for agent $a$ and the strategy chosen depend on the strategies of all other players, $\pi^{-a}$. This interdependence is an ever-present feature of our model. For the sake of presentation we leave it implicit whenever possible and we write the solution to the BSDE giving the risk for Agent $a$ as $(Y^a,Z^a)$ instead of $\big(Y^a(\pi^a,\pi^{-a}),Z^a(\pi^a,\pi^{-a})\big)$. We will use the latter when the situation requires it.

\subsubsection*{Competitive equilibrium, equilibrium market price of risk and endogenous trading}

We denote by $n \in \bR$ the number of units of derivative present in the market. While each unit of derivative pays $H^D$ at time $T$, the agents are free to buy and underwrite contracts for any amount of $H^D$, so that $n$ is not necessarily an integer. We think essentially of the case $n=0$, where every derivative held by an agent has been underwritten by another agent in $\bA$, entailing essentially that agents share their risks with each other (see \cites{BarrieuElKaroui2005,BarrieuElKaroui2009} or \cite{HorstMuller2007}). Building upon \cite{HorstPirvuDosReis2010} allows for a bit more flexibility as  $n\neq 0$ is possible\footnote{The situation $n \neq 0$ would be possible if, prior to time $t=0$, another agent $a_0 \notin \bA$ was on the derivatives market and then stopped her activity, for instance $a_0$ might have had as objective to buy $m>0$ units, according to her own specific criteria, in which case $n=-m<0$.}. In any case, over $[0,T]$, only the agents in our set $\bA$, with trading objectives as  described above, are active in the market and so the total number $n$ of derivatives present is constant over time.

We assume that each agent seeks to minimize her risk measure independently, without cooperation with the other agents, so we are interested in Nash equilibria. 
\begin{definition}[Equilibrium and equilibrium MPR (EMPR)]
\label{def-equilibrium}
  For a given \emph{Market Price of Risk} (MPR) $\theta=(\theta^S,\theta^R)$, we call $\pi^*=\left(\pi^{*,a}\right)_{a \in \bA}$ an \emph{equilibrium} if, for all $a\in\bA$ , $\pi^{*,a} \in\cA^\theta(\pi^{*,-a})$ and 
\begin{enumerate}
	\item[] for any admissible strategy $\pi^a$ it holds that $Y_0^a(\pi^{*,a},\pi^{*,-a}) \leq Y_0^a(\pi^a,\pi^{*,-a})$,
\end{enumerate}
i.e. individual optimality given the strategies of the other agents.
We call $\theta$ \emph{Equilibrium MPR (EMPR)} and $\theta^R$ \emph{Equilibrium Market Price of external Risk (EMPeR)} if
\begin{enumerate}	 
	\item $\theta=(\theta^S,\theta^R)$ makes $\bP^\theta$ a true probability measure (equivalently, $\cE^{\theta}$ from \eqref{density-process} is a uniformly integrable martingale);

	\item there exists a unique equilibrium $\pi^*$ for $\theta$;

	\item $\pi^*$ satisfies the market clearing condition (or fixed supply condition) for the derivative $H^D$ (where $Leb$ denotes the Lebesgue measure): 
	\begin{align}
		\label{nnetsupplycondition}
			\sum_{a \in \bA} \pi_t^{*,a,2} = \sum_{a \in \bA} \pi_{0}^{*,a,2} = n \quad \bP \otimes Leb-a.e.
		\end{align}
	\end{enumerate}

\end{definition}
%
%
%

%
%
%
\section{The single agent's optimization and unconstrained equilibrium} 
\label{section-SingleAgentSection}  

In this section and the one following, we study the solvability of the problem and the structure of the equilibrium for general risk measures induced by a BSDE. 
Finding an equilibrium market price of external risk is, essentially, an optimization under the fixed-supply constraint. 
So, prior to looking whether such an equilibrium MPR exists (postponed to Section \ref{sec4-repagent}), we start by fixing an arbitrary MPR $\theta \in \cH_{BMO}$ and solve for the behavior of the system of agents given that MPR, without the fixed-supply constraint. 
For this, we first solve for the behavior of the individual agents given that the others have chosen their strategies (the so-called best response problem), and then solve for the Nash equilibrium for the system of agents.

\subsection{Optimal response for one agent}

In this subsection, in addition to a MPR $\theta$ being fixed, we assume given the strategies $\pi^{-a} = (\pi^b)_{b \in \bA\backslash\{a\}}$ of the other agents, for a fixed agent $a \in \bA$, 
and we study the investment problem for a single agent whose preferences are encoded by $g^ a$.

\subsubsection*{Optimizing the residual risk}

To solve the optimization problem for agent $a$, 
we first recall from \cite{HorstPirvuDosReis2010} that, at each time $t \in [0,T]$,
the strategy chosen must minimize the residual risk: the additivity of the risk measure implies (writing $V_T = (V_T - V_t) + V_t$ and using the translation invariance) that
	\begin{align*}
		Y_t^a 
			=\rho^a_t\Big( H^a + V^a_T - \wt{\lambda}^a \bar{V}^{-a}_T \Big)		
			= \rho^a_t\Big( H^a + (V^a_T-V^a_t) - \wt{\lambda}^a (\bar{V}^{-a}_T-\bar{V}^{-a}_t) \Big) - \big( V^a_t - \wt{\lambda}^a \bar{V}^{-a}_t \big) .
	\end{align*}
This suggests applying the following change of variables to \eqref{BSDE_1} (using \eqref{eq:averageVaminusa}),
\begin{align}		
\label{eq:BSDEchangeVArs}
	\begin{aligned}
		\wt{Y}_t^a &:= Y_t^a + \big( V_t^{a} - \wt{\lambda}^a \bar{V}_t^{-a} \big),
		\\
		\wt{Z}_t^a &:= Z_t^a +\zeta_t^a 	, 
			\quad \text{where} \quad
			\zeta_t^a = \big( \pi_{t}^{a,1} \sigma_t +\pi_{t}^{a,2} \kappa^\theta_t \big) - \wt{\lambda}^a \big( \bar{\pi}^{-a,1}_t \sigma_t + \bar{\pi}_{t}^{-a,2} \kappa^\theta_t \big) \in \bR^2.
	\end{aligned}
\end{align}	
If the strategies are not clear from the context, we also write $\zeta^a = \zeta^a(\pi)=\zeta^a(\pi^a,\pi^{-a})$. 
Direct computations yield a BSDE for $(\wt{Y}^a,\wt{Z}^a)$ given by 
\begin{align}		
\label{agentresidualBSDEtransf}
		- d\wt{Y}^a_t 
		= 
		\wt{g}^a\Big(t,\pi_t^a,\pi_t^{-a},\wt{Z}_t^a\Big) \udt - \langle \wt{Z}_t^a, \udwt \rangle 
	\quad \text{with terminal condition } \quad 
		\wt{Y}_T^a 	= -H^a,
\end{align}
where the driver $\widetilde{g}^a \colon \Omega\times \left[0,T\right]  \times \bR^2 \times (\bR^2)^{N-1}\times \bR^2 \to \bR$ is defined as 
\begin{align}		
		\label{eq-version2-Gab-v00}
		\wt{g}^a(t,\pi^a_t,\pi^{-a}_t,z^a) 
		:
		&
		= g^a\big(t, z^a - \zeta_t^a \big) - \lev \zeta_t^a , \theta_t \rev 
		\\
		\label{eq-version2-Gab}	
		&
		=
		g^a\bigg(t, 
		    z^a 
		- \Big( 
			   \big( \pi^{a,1}_t - \wt{\lambda}^a \bar{\pi}^{-a,1}_t \big) \sigma_t 
					+\big( \pi^{a,2}_t - \wt{\lambda}^a \bar{\pi}^{-a,2}_t \big) \kappa^\theta_t
			\Big) 
   \bigg)
\\
\nonumber
&\hspace{3cm} 
		- \lev 
			\big( \pi^{a,1}_t - \wt{\lambda}^a \bar{\pi}^{-a,1}_t \big) \sigma_t 
		+ \big( \pi^{a,2}_t - \wt{\lambda}^a \bar{\pi}^{-a,2}_t \big) \kappa^\theta_t
		 , 
		\theta_t \rev.
\end{align}
Each individual agent $a\in \bA$ seeks to minimize $\widetilde Y^a_0$,
the solution to \eqref{agentresidualBSDEtransf}, via her choice of investment strategy $\pi^a\in \cA^\theta(\pi^{-a})$, in other words she aims at solving
\begin{align}
\label{eq:ResidualRiskOptimizationProblem}
		\min_{\pi^a\in \cA^\theta(\pi^{-a})} \widetilde{Y}^a_0 (\pi^a,\pi^{-a}).
\end{align} 
Before we solve the individual optimization, we assume that the derivative $H^D$ does indeed complete the market. This must then be verified a posteriori (once the solution is computed) and case-by-case depending on the specific model. 
\begin{assumption}
\label{assump-kappa}
Assume that $\kappa^R_t \neq 0$, for any $t \in [0,T]$, $\bP$-a.s. .
\end{assumption}

\subsubsection*{The pointwise minimizer for the single agent's residual risk}

In \eqref{agentresidualBSDEtransf}, the strategy $\pi^a$ appears only in the driver $\wt g^a$. The comparison theorem for BSDEs suggests that in order to minimize $\widetilde{Y}^a_0(\pi^a)$ over $\pi^a$
one needs only to minimize the driver function $\widetilde g^a$ over $\pi^a_t$, for each fixed $\omega$, $t$, $\pi^{-a}_t$ and $z^a$. We define such pointwise minimizer as the random map 
\begin{align*} 
		\Pi^a(\omega,t,\pi^{-a}_t,z) := \arg\min_{\pi^a \in\bR^2} \widetilde g^a(\omega,t,\pi^a,\pi^{-a}_t,z), \qquad 
		(\omega,t,\pi^{-a}_t,z) \in \Omega \times [0,T]\times(\bR^2)^{N-1} \times \bR^2,
\end{align*}
and 
$\wt{G}^a(t,\pi^{-a}_t,z^a) : = \wt{g}^a \big(t,\Pi^a(t,\pi^{-a},z^a),\pi^{-a}_t,z^a\big)$ as the minimized driver.

The pointwise minimization problem has, under Assumption \ref{ass:driver}, a unique minimizer which is characterized by the \emph{first order conditions (FOC)} for $\widetilde g^a$, i.e. $\nabla_{\pi^a} \widetilde g^a(t,\pi^a,\pi^{-a}_t,z^a)=0$. Recall that $\sigma = (\sigma^S S,0)$. Using \eqref{eq-version2-Gab}, the FOC is equivalently written as
\begin{align}
\nonumber
	\partial_{\pi^{a,1}}\widetilde g^a (t,\pi^a,\pi^{-a},z^a) = 0 
		&\Leftrightarrow 
				\big\langle (\nabla_z g^a)(t,z^a-\zeta^a), - \sigma \big\rangle  - \langle \sigma, \theta\rangle =0
\\	\label{FOC-number-one-widetilde-G}
		&\Leftrightarrow 
				g^a_{z^1} (t,z^a-\zeta^a) = -\theta^S,
\\ \nonumber
	\partial_{\pi^{a,2}}\widetilde g^a (t,\pi^a,\pi^{-a},z^a) = 0
		&\Leftrightarrow 
				\big\langle (\nabla_z g^a)(t,z^a-\zeta^a) , -\kappa^\theta\big\rangle -\langle \kappa^{\theta}, \theta \rangle=0 
\\ \nonumber
		&\Leftrightarrow 
				-\theta^S \, \kappa^S + g^a_{z^2} (t,z^a-\zeta^a) \, \kappa^R = - \kappa^S \theta^S - \kappa^R \theta^R
\\	\label{FOC-number-two-widetilde-G}
		&\Leftrightarrow 
				g^a_{z^2}(t,z^a-\zeta^a) = -\theta^R,
\end{align}
where we used \eqref{FOC-number-one-widetilde-G} to obtain \eqref{FOC-number-two-widetilde-G} under Assumption \ref{assump-kappa}.

With $\cZ^a$ from \eqref{eq:caligraphicZ}, the FOC system \eqref{FOC-number-one-widetilde-G}--\eqref{FOC-number-two-widetilde-G} is equivalent to $z^a-\zeta^a_t=\cZ^a(t,-\theta_t)$. The expression for $\zeta$ in \eqref{eq:BSDEchangeVArs} and elementary re-arrangements allow to rewrite $z^a-\zeta^a_t=\cZ^a(t,-\theta_t)$ as
\begin{align}		
\label{equation-general.structure.of.the.map.Pi}
\begin{aligned}
	\Pi^{a,1}(t,\pi^{-a}_t,z^a) - \wt{\lambda}^a \bar{\pi}^{-a,1}_t 
			&
			= 
			\frac{ z^{a,1} - \cZ^{a,1}(t,-\theta_t) }{\sigma^S S_t} 
		- \frac{ z^{a,2} - \cZ^{a,2}(t,-\theta_t) }{\kappa^R_t} 
		\ \frac{\kappa^S_t}{\sigma^S S_t},				
			\\
	\Pi^{a,2}(t,\pi^{-a}_t,z^a) - \wt{\lambda}^a \bar{\pi}^{-a,2}_t 
			&
			= 
			\frac{ z^{a,2} - \cZ^{a,2}(t,-\theta_t) }{\kappa^R_t}.
\end{aligned}
\end{align}
Plugging $z^a-\zeta^a_t=\cZ^a(t,-\theta_t)$ into \eqref{eq-version2-Gab-v00} yields an expression for the minimized (random) driver  
	\begin{align}		
	\label{equation---minimized.driver.wtGa.for.general.ga} 
		\wt{G}^a(t,\pi^{-a}_t,z^a)   
			= g^a\big( t,\cZ^a(t,-\theta_t) \big) +  \scalar{\cZ^a(t,-\theta_t)}{\theta_t} - \scalar{z^a}{\theta_t}  =: \wt{G}^a(t,z^a) .
	\end{align}
		Since $g^a$ is generic at this point, the process $\cZ^a(t,-\theta_t)$ in not known precisely. Nonetheless, the general structure of $\Pi^a$ and the minimized driver $\wt{G}^a$ are determined. 	
		We stress two important things. First, $\wt{G}^a$ is an affine driver with stochastic coefficients. Second, $\wt{G}^a$ does not depend at all on $\pi^{-a}$. 
		This means that while the optimal strategy $\pi^{*,a}$ (see below) depends on the strategies of the other agents, the minimized risk does not. In \cite{HorstPirvuDosReis2010}, 
		the authors did not obtain this general form for the minimized driver.

\subsubsection*{Single-agent optimality}

Since $\wt{G}^a$ is an affine driver, and since $\nabla_z \wt{G}^a = -\theta \in \cH_{BMO}$, we have a unique solution to the BSDE with driver $\wt{G}^a$ and terminal condition $-H^a$ provided that the process $(\omega,t) \mapsto \wt{G}^a(t,0) = g^a\big( t,\cZ^a(t,-\theta_t) \big) +  \scalar{\cZ^a(t,-\theta_t)}{\theta_t}$ is integrable enough, which we assume. 
Let then $(\wt{Y}^a,\wt{Z}^a)$ be the solution to BSDE \eqref{agentresidualBSDEtransf} with driver \eqref{equation---minimized.driver.wtGa.for.general.ga} and define the strategy $\pi^{*,a}_{\cdot} := \Pi^a(\cdot,\pi^{-a}_\cdot,\wt{Z}^a_\cdot)$. We now prove that the above methodology indeed yields the solution to the individual risk minimization problem, in other words the so-called best response.

\begin{theorem}[Optimality for one agent]	
	\label{theorem--single_agent_optimality} 
	Assume the market price of risk $\theta=(\theta^S,\theta^R) \in \cH_{BMO}$ and let Assumption \ref{assump-kappa} hold. 
	Fix an agent $a\in\bA$ and a set of integrable strategies $\pi^{b}$ for $b \in \bA \setminus \left\{a\right\}$. Assume further that
		\begin{itemize}
			\item for $\wt{G}^a$ given by \eqref{equation---minimized.driver.wtGa.for.general.ga}, ${|\wt{G}^a(\cdot,0)|}^\frac12 \in \cHBMO$, and 
			\item $\pi^{*,a}_\cdot=\Pi^a(\cdot,\pi^{-a}_\cdot,\wt{Z}^a_\cdot)$ is admissible .
		\end{itemize}
Then BSDE with driver \eqref{equation---minimized.driver.wtGa.for.general.ga} and terminal condition $-H^a$ has a unique solution $(\wt{Y}^a,\wt{Z}^a)\in\cS^\infty \times \cHBMO$. Moreover, $\wt{Y}^a_0$ is the value of the optimization problem \eqref{eq:ResidualRiskOptimizationProblem} (i.e.~the minimized risk) for agent $a$ and $\pi^{*,a}$ is the unique optimal strategy. 
\end{theorem}

\begin{proof} 
Given the structure of $\wt G^a$ in \eqref{equation---minimized.driver.wtGa.for.general.ga} and the integrability assumption made, the existence and uniqueness of the BSDE's solution $(\wt{Y}^a,\wt{Z}^a)$ in $\cS^\infty \times \cHBMO$ is straightforward.  

	We first use the comparison theorem to prove the minimality of $\wt{Y}^a$, and hence the optimality of $\pi^{*,a}$.	Let $t\in[0,T]$. 
	Take any strategy $\pi^a \in \cA^\theta(\pi^{-a})$.
	First, from the definition of $\wt{G}^a$ as a pointwise minimum, we naturally have that 
	$\wt{G}^a(t,z^a) = \wt{g}^a\left(t,\Pi^a(t,\pi^{-a}_t,z^a),\pi^{-a}_t,z^a\right) \leq \wt{g}^a(t,\pi^a_t,\pi^{-a}_t,z^a)$ for all $t$ and $z^a$, 
	that is, $\wt{G}^a(\cdot,\cdot) \le \wt{g}^a(\cdot,\pi^a_\cdot,\pi^{-a}_\cdot,\cdot)$.
	Second, $\wt{G}^a$ is affine and thus Lipschitz, with Lipschitz coefficient process $-\theta \in \cH_{BMO}$. 
	By the comparison theorem, we therefore have, for any $t\in[0,T]$ and in particular for $t=0$, that $\wt{Y}^{a}_t = \wt{Y}^a_t(\pi^{*,a},\pi^{-a}) \leq \wt{Y}^a_t(\pi^a,\pi^{-a})$. 
	As this holds for any $\pi^a\in \cA^\theta(\pi^{-a})$, this proves the minimality of $\wt{Y}^{a}_0 = \rho_0^a\big(\xi^a(\pi^{*,a},\pi^{-a})\big)$ and thus the optimality of $\pi^{*,a}$.

We now argue the uniqueness of the optimizer $\pi^{*,a}$. 
Let $\pi^a$ be an admissible strategy and let $(\wt{Y}^a(\pi^{a}), \wt{Z}^a(\pi^a))$ be the corresponding risk, i.e. solution to the BSDE \eqref{agentresidualBSDEtransf} with strategy $\pi^a$. 
We compute the difference $\wt{Y}^a_t(\pi^a) - \wt{Y}^a_t(\pi^{*,a})$, adding-and-substracting in the Lebesgue integral
$\wt{g}^a\left(t,\Pi^a\big(t,\pi^{-a}_t,\wt{Z}^a_t(\pi^a)\big),\pi^{-a}_t,\wt{Z}^a_t(\pi^a)\right) = \wt{G}^a(t,\wt{Z}^a_t(\pi^a))$, and using the affine form of $\wt{G}^a$ :
\begin{align}
\label{eq--uniquenessBSDEproof} 
 \wt{Y}^a_t(\pi^a) -& \wt{Y}^a_t(\pi^{*,a}) \nonumber \\
= & \int_t^T \left[\wt{g}^a\left(s,\pi^a_s,\pi^{-a}_s,\wt{Z}^a_s(\pi^a)\right) - \wt{G}^a\left(s,\wt{Z}^a_s(\pi^{*,a})\right)  \right] \uds 
\nonumber 
  - \int_t^T [\wt{Z}^a_s(\pi^a) - \wt{Z}^a_s(\pi^{*,a})] \ud W_s
\nonumber \\
= &\int_t^T \Big[\wt{g}^a\left(s,\pi^a_s,\pi^{-a}_s,\wt{Z}^a_s(\pi^a)\right) - \wt{g}^a\left(s,\Pi^a\big(s,\pi^{-a}_s,\wt{Z}^a_s(\pi^a)\big),\pi^{-a}_s,\wt{Z}^a_s(\pi^a)\right)\Big] \ud s
\\
\nonumber
& \qquad -\int_t^T  [\wt{Z}^a_s(\pi^a) - \wt{Z}^a_s(\pi^{*,a})] \ud W^\theta_s .
\end{align}
By construction of $\Pi^a$ as a minimizer, the difference in \eqref{eq--uniquenessBSDEproof} is always positive. In particular, taking $\bP^\theta$-expectation w.r.t. $\cF_t$ implies that $\wt{Y}^a_t(\pi^a) - \wt{Y}^a_t(\pi^{*,a}) \ge 0$ for all $t \in [0,T]$.
Assume that $\pi^a$ is an optimal strategy. Then $\wt{Y}^a_0(\pi^a) = \wt{Y}^a_0(\pi^{*,a})$ and the LHS for $t=0$ vanishes. 
Under $\bP^\theta$-expectation, the stochastic integral on the RHS also vanishes and we can conclude that the integrand in \eqref{eq--uniquenessBSDEproof} is zero $\bP^\theta \otimes Leb$-a.e.
Consequently, we obtain $\wt{Y}^a(\pi^a) = \wt{Y}^a(\pi^{*,a})$ and hence $\wt{Z}^a(\pi^a) = \wt{Z}^a(\pi^{*,a})$. 
By uniqueness of the minimizer, we then have $\pi^a_{\cdot} = \Pi^a\big(\cdot,\pi^{-a}_{\cdot},\wt{Z}^a_{\cdot}(\pi^a)\big) = \Pi^a\big(\cdot,\pi^{-a}_{\cdot},\wt{Z}^a_{\cdot}(\pi^{*,a})\big) = \pi^{*,a}_{\cdot}$.
\end{proof}

	\begin{remark}
	\label{remark---ParticularCaseTheo3.3}
		While Theorem \ref{theorem--single_agent_optimality} is stated as the optimal response of a single agent $a$ in the system $\bA$ with the other strategies $\pi^{-a}$ being fixed,
		it is clear that it can more generally describe the optimal investment of an agent with preferences described by $g^a$ (equivalently, $\rho^a_{\cdot}$)
		who trades in the assets $S$ and $B$, which have the given MPR $\theta$ (one can think of making $\bA=\{a\}$, or doing $\lambda^a = 0$). 
		Following the same methods, the result could be generalized to a higher number of assets, with price processes given exogenously. This applies similarly to an agent trading in fewer assets, by setting the respective components to zero -- see Theorem \ref{theorem:EMPR_aggregate_2_2}. 
	\end{remark}

We now state a characterization of the optimal strategy via the FOC.

	\begin{lemma} 
	\label{remark---FOC.for.processes.pi.and.Z.implies.optimality}
		Under the assumptions of Theorem \ref{theorem--single_agent_optimality}, 
		let $\wh{\pi}^a$ be an admissible strategy and $(\wh{Y}^a,\wh{Z}^a)$ be the associated risk process, 
		solution to the BSDE with driver $\wt{g}^a(t,\wh{\pi}^a_t,\pi^{-a}_t,\cdot)$ and terminal condition $-H^a$.
		Assume that they satisfy the FOC \eqref{FOC-number-one-widetilde-G}--\eqref{FOC-number-two-widetilde-G} in the sense that
			\begin{align*}
				\nabla_z g^a(t, \wh{Z}^a_t - \wh{\zeta}^a_t) = -\theta_t
				\qquad \text{where} \qquad
				\wh{\zeta}^a_t = \big( \wh{\pi}_{t}^{a,1} \sigma_t +\wh{\pi}_{t}^{a,2} \kappa^\theta_t \big) - \wt{\lambda}^a \big( \bar{\pi}^{-a,1}_t \sigma_t + \bar{\pi}_{t}^{-a,2} \kappa^\theta_t \big).
			\end{align*}
		Then $(\wh{Y}^a,\wh{Z}^a) = (\wt{Y}^a,\wt{Z}^a)$ and $\wh{\pi}^a = {\pi}^{*,a}$.
	\end{lemma}

	\begin{proof}
		By the assumptions on $g^a$, $\nabla_z g^a(t, \wh{Z}^a_t - \wh{\zeta}^a_t) = -\theta_t$ means that 
		$\wh{Z}^a_t - \wh{\zeta}^a_t = \cZ^a(t,-\theta_t)$, or equivalently $\wh{\pi}_t = \Pi^a(t,\pi^{-a}_t,\wh{Z}^a_t)$. 
		Therefore $\wt{g}^a(t,\wh{\pi}^a_t,\pi^{-a}_t,\wh{Z}^a_t) = \wt{G}^a(t,\wh{Z}^a_t)$ 
		-- recall \eqref{eq-version2-Gab-v00}. By uniqueness of the solution to the BSDE with driver $\wt{G}^a(t,\cdot)$ and terminal condition $-H^a$, we have $(\wh{Y}^a,\wh{Z}^a) = (\wt{Y}^a,\wt{Z}^a)$.
		Consequently, by the uniqueness of the FOC's solution, $\wh{\pi}^a_t = \Pi^a(t,\pi^{-a}_t,\wh{Z}^a_t) = \Pi^a(t,\pi^{-a}_t,\wt{Z}^a_t) = {\pi}^{*,a}_t$. 
	\end{proof} 

\subsection{The unconstrained Nash equilibrium}

Having solved the optimization problem for one agent, we now look at the existence and uniqueness of a Nash equilibrium, still for the MPR $\theta \in \cH_{BMO}$ fixed at the beginning of this section, and still with no fixed-supply constraint.

Assume $\pi^*=(\pi^{*,a})_{a \in \bA}$ is a Nash equilibrium. Fix an agent $a \in \bA$. From the uniqueness of the optimal strategy, given by Theorem \ref{theorem--single_agent_optimality}, one must have 
	\begin{align*}
		\pi^{*,a}_t = \Pi^a(t,\pi^{*,-a}_t,\wt{Z}^a_t),\quad t \in [0,T],
	\end{align*}
where $(\wt{Y}^a,\wt{Z}^a)$ is the solution to the BSDE with terminal condition $-H^a$ and driver $\wt{G}^a$ given in \eqref{equation---minimized.driver.wtGa.for.general.ga}.
From the characterization \eqref{equation-general.structure.of.the.map.Pi} of $\Pi^{a}$, we therefore have, for all $a \in \bA$ and $t \in [0,T]$, 
	\begin{align}		\label{eq:ApiJlinearsystem}		
		\begin{aligned}
			\pi^{*,a,1}_t - \wt{\lambda}^a \bar{\pi}^{*,-a,1}_t 
				&= 
				\frac{ \wt{Z}^{a,1}_t - \cZ^{a,1}(t,-\theta_t) }{\sigma^S S_t} - \frac{ \wt{Z}^{a,2}_t - \cZ^{a,2}(t,-\theta_t) }{\kappa^R_t} \ \frac{\kappa^S_t}{\sigma^S S_t} =: J^{a,1}_t,				
			\\
			\pi^{*,a,2}_t - \wt{\lambda}^a \bar{\pi}^{*,-a,2}_t 
				&= 
				\frac{ \wt{Z}^{a,2}_t - \cZ^{a,2}(t,-\theta_t) }{\kappa^R_t} =: J^{a,2}_t.
		\end{aligned}
	\end{align}
Note that, for any $a \in \bA$, the process $(\wt{Y}^a,\wt{Z}^a)$ does not depend on $\pi^{*,-a}$ nor $\pi^{*,a}$, seeing as neither $-H^a$ nor $\wt{G}^a$ does. Therefore, $J^a_t$ is also independent of the unknown $\pi^{*}_t$, which is only present in the LHS of \eqref{eq:ApiJlinearsystem}. 

Conversely, assume we can solve for $\pi^*$ in Equation \eqref{eq:ApiJlinearsystem} and that $\pi^*$ is integrable against the prices. Then, since $\pi^{*,a}_t = \Pi^a(t,\pi^{*,-a}_t,\wt{Z}^a_t)$ by \eqref{equation-general.structure.of.the.map.Pi}, Theorem \ref{theorem--single_agent_optimality} guarantees that $\pi^{*,a}$ is the best response to $\pi^{*,-a}$, and we therefore have a Nash equilibrium. 

So, the existence and uniqueness of a Nash equilibrium $\pi^*$ is equivalent to the existence and uniqueness of solutions to Equation \eqref{eq:ApiJlinearsystem}.

Define the matrix $A_N \in \bR^{N \times N}$ by\footnote{Recall the notation that for sums and products over certain subsets of $\bA$ we identify $\bA$ with the set $\left\{1,2,\ldots,N\right\}$, where $N \in \bN$ is the fixed finite number of agents.}
\begin{align}
\label{eq:ANmatrix}
A_{N}=
\begin{bmatrix}
   1      &   & -\frac{\lambda^1}{N-1}\\
   &  \ddots      &   \\
  -\frac{\lambda^{N}}{N-1} & & 1
\end{bmatrix},
\end{align}
i.e. the $j$-th line has the entries $- \wt \lambda^j=-{\lambda^{j}}/{(N-1)}$, everywhere but for the $j$-th one which is $1$.
Equation \eqref{eq:ApiJlinearsystem} can be rewritten as
	\begin{align}		\label{eq:ApiJlinearsystemMatrixForm}
		A_{N} \ \pi^{*,\cdot,i} = J^{\cdot,i},
	\end{align}
where $\pi^{*,\cdot,i} = (\pi^{*,a,i})_{a \in \bA}$ and $J^{\cdot,i} = (J^{a,i})_{a \in \bA}$, for $i \in \{1,2\}$.

\begin{theorem}
	\label{theorem:invertibilityAN}
	Assume the market price of risk $\theta=(\theta^S,\theta^R) \in \cH_{BMO}$, 
	that Assumption \ref{assump-kappa} and Assumption \ref{assump:ConcernRates} hold, 
	that, for all $a \in \bA$, $J^a$ is integrable against the prices 
	and $\abs{\wt{G}^a(\cdot,0)}^{1/2} \in \cH_{BMO}$.
	
	Then there exists a unique Nash equilibrium $\pi^* = (\pi^{*,a})_{a \in \bA}$ associated with the MPR $\theta$, which is given by the unique solution to \eqref{eq:ApiJlinearsystemMatrixForm}.
\end{theorem}
\begin{proof}
	The determinant of the $A_N$ is
		\begin{equation*}
			\det (A_{N}) = 1 - \sum_{i < j} \wt \lambda^i \wt \lambda^j
									- 2 \sum_{i < j < k} \wt \lambda^i \wt \lambda^j \wt \lambda^k
									- 3 \sum_{i < j < k < l} \wt \lambda^i \wt \lambda^j \wt \lambda^k \wt \lambda^l
									- \ldots
									- (N-1) \prod_{i=1}^N \wt \lambda^i,			
		\end{equation*}
	where the sums run over indices $i, j, k, l$ from $\bA = \left\{1,\ldots,N\right\}$.
	If $\lambda^a=1$ for all $a \in \bA$, then $\det (A_N) = 1 - \sum_{k=2}^N \frac{k-1}{(N-1)^k} \binom{N}{k} = 0$, so the matrix is not invertible. 
	The determinant is strictly decreasing in each $\wt \lambda^a$ ($a \in \bA$) and therefore also in $\lambda^a$. 
	Hence, if $\lambda^a \in [0,1]$ for all $a \in \bA$ and the product $\prod_{a \in \bA} \lambda^a < 1$, then at least one factor must be strictly smaller than one and the determinant must be strictly positive (i.e. $\det(A_N) >0$). The invertibility of $A_N$ follows. This guarantees that one can solve system \eqref{eq:ApiJlinearsystem} (or, equivalently, \eqref{eq:ApiJlinearsystemMatrixForm}) for each $i \in \{1,2\}$ to obtain $(\pi^{*,\cdot,i})$. The integrability of $\pi^*$ follows from fact that each component $\pi^{*,a}$ is a linear combination of the integrable $J^a$'s. 
	Finally, the Nash-optimality of $\pi^*$ was argued in the identification of the core Equation \eqref{eq:ApiJlinearsystem}.
\end{proof}
We can now comment on Assumption \ref{assump:ConcernRates}. If $\lambda^b = 0$ for all $b \in \bA \setminus \left\{a\right\}$, then $A_N$ is invertible independent of $\lambda^a$, i.e. in particular for $\lambda^a=1$. This shows that the condition that $\lambda^a \in [0,1)$ for all $a \in \bA$ is not necessary, but merely a sufficient condition. Finally, if we were to allow for $\lambda^a > 1$, then $\prod_a \lambda^a <1$ is not sufficient for invertibility of $A_N$, e.g. in the case $N=3$ take $\lambda^a=\lambda^b = 2$ and $\lambda^c=0$.

From now on we assume the agents' optimization problems have a solution so that it makes sense to discuss the notion of equilibrium market price of risk (EMPR).
\begin{remark}
Notice that at this point, as $\theta$ is given exogenously, we do not yet have a system of coupled BSDEs. For each $a \in \bA$, we obtain the value process $(\wt{Y}^a,\wt{Z}^a)$ as the solution to a BSDE where the terminal condition $-H^a$ and driver $\wt{G}^a$, which does not depend on the strategies and only takes $\wt{Z}^a$ as argument.	
These processes are then used to solve for the Nash equilibrium of strategies $\pi^*$, and so $(\wt{Y}^a,\wt{Z}^a)_{a \in \bA}$ is the value process of the Nash equilibrium.
This feature (no coupling of the BSDEs when solving for the optimal values), as well as the fact that solving for the optimizers $\pi^*$ of the Nash equilibrium reduces to solving a linear system (for which existence and uniqueness of the solution is equivalent to the invertibility of a matrix) is a consequence of the structure of the problem, i.e. the form of the concerns over the relative performance. In particular, it does not depend on the specific form of the individual risk measures $\rho^a$ (or equivalently, the drivers $g^a$). Finding the EMPR, later on in Section \ref{sec4-repagent}, leads to multi-dimensional quadratic BSDEs.
\end{remark}

\subsection{An example: the entropic risk measure case}  
\label{subsection-unconstrained.case-entropic.case}

We now illustrate the methodology and result of Theorem \ref{theorem--single_agent_optimality} for a particular risk measure, and prepare the ground for the model we study in Sections \ref{section-EntropicCase} and \ref{section-NumSection}.
We give a sequence of examples, in increasing order of complexity, that show how the structure of the optimal strategies is changing as features are added. As in the above, the examples do not yet take into account the market clearing condition but rather assume that a market price of risk $\theta=(\theta^S,\theta^R) \in \cH_{BMO}$ is given. Nonetheless they give a flavor for the next section where the equilibrium market price of risk is derived.

Each agent $a\in \bA$ is assessing her risk using the entropic risk measure $\rho^a_0$ for which the driver $g^a: \bR^2 \to \bR $ if given by
	\begin{align*}
	g^a(z):= \frac1{2\gamma_a}{|z|^2},
	\quad\text{where }\quad
	\gamma_a > 0
	\quad \text{is agent $a$'s \emph{risk tolerance},}
	\end{align*}
and ${1}/{\gamma_a}$ is agent $a$'s \emph{risk aversion}. 
This choice of $g^a$ relates to exponential utilities, and we have (see e.g. \cite{Carmona2009}, \cite{FoellmerKnispel2011}, \cite{HuImkellerMuller2005} or \cite{Rouge2000}) 
	\begin{align*}
		\rho^a_0(\xi) 
		= Y_0^a = \gamma_a\ln \bE[ e^{ - \xi/\gamma_a }] 
		= \gamma_a \ln\big(- U_{\gamma_a}(\xi) \big) 
		\qquad\text{with} \qquad U_{\gamma_a}(\xi) 
		= \bE[ - e^{ - \xi/\gamma_a }] ,
	\end{align*}
so that, equivalently, the agents are maximizing their expected (exponential) utility.


In what follows, the optimal strategies were computed using the techniques described so far and hence we omit the calculations. They boil down to finding the map $\cZ^a$ arising from \eqref{eq:caligraphicZ}, then injecting it in \eqref{equation-general.structure.of.the.map.Pi} and \eqref{equation---minimized.driver.wtGa.for.general.ga} to obtain $\wt{G}^a$. 
We denote throughout by $\pi^{*,a}$ (for $a\in\bA$) the optimal Nash equilibrium strategy.

\subsubsection{The reference case of a single agent} 

For comparison, we first give the optimal strategy for a single agent who could trade liquidly in the stock of price $S$ and the derivative of price $B$ with (arbitrary and exogenously-given) market price of risk $\theta=(\theta^S,\theta^R)$. She aims at minimizing her risk, with terminal endowment and trading gains $\xi^a = H^a + V^a_T(\pi^a)$. Here, other agents do not play a role. Since $g^a_{z^i}(z) = z^i/{\gamma_a}$, it is easily found that $\cZ^a(t,-\theta_t) = (- \gamma_a \theta^S_t, - \gamma_a \theta^R_t) = - \gamma_a \theta_t$. 
Injecting this in \eqref{equation---minimized.driver.wtGa.for.general.ga} yields the minimized driver $\wt{G}^a$, 
	\begin{align*}
		\wt{G}^a(t,z^a) 
		= -\frac{\gamma_a}{2} \abs{\theta_t}^2 - \scalar{z^a}{\theta_t} , \quad t\in [0,T].
	\end{align*}
The minimized risk is then given by $Y^a_0 = \wt{Y}^a_0$ where $(\wt{Y}^a,\wt{Z}^a)$ is the solution to the BSDE with terminal condition $-H^a$ and driver $\wt{G}^a$, while the optimal strategy is then given by 
	\begin{align*}
		&
		\pi^{*,a,1} 
		=  
		\frac{  \wt{Z}^{a,1} + \gamma_a \theta^S }{\sigma^S S} - \frac{ \wt{Z}^{a,2} + \gamma_a \theta^R} {\kappa^{R}} \, \frac{\kappa^S}{\sigma^S S} 
		\qquad \text{and} \qquad 
		\pi^{*,a,2} 
		=
		\frac{ \wt{Z}^{a,2} + \gamma_a \theta^R} {\kappa^{R}} .
	\end{align*}
This result is expected and in line with canonical mathematical finance results. The particular structure of the optimal strategy follows from the fact that the second asset is correlated to the first when $\kappa^S\neq 0$, and the inversion of the volatility matrix for the 2-dimensional price $(S,B)$,
	\begin{align*}
		\matrix{\sigma^S S & 0 \\ \kappa^S  & \kappa^R} . 
	\end{align*}
The market faced by $a$ is complete, the driver for the minimized residual risk $\wt{Y}^a$ is affine and we have the explicit solution 
\begin{align*}
\wt{Y}^a_0 
= 
\bE^\theta \left[ -H^a - \frac{\gamma_a}{2} \int_0^T \abs{\theta_u}^2 \udu \right] 
= 
-
\bE\left[ 
\cE^{-\theta}_T \cdot
 \ \Big(H^a + \frac{\gamma_a}{2} \int_0^T \abs{\theta_u}^2 \udu \Big) \right].
\end{align*}
The minimized risk measure is affine with respect to $H^a$ : the trend ($\theta \neq 0$) in the prices leads to a constant risk reduction and the completeness of the market leads to an affine dependence on $H^a$.

	\begin{remark}
		Note that $\wt{G}^a(t,0) = -\frac{\gamma_a}{2} \abs{\theta_t}^2$. 
		Therefore, since $\theta \in \cH_{BMO}$, the assumption we made in Theorems \ref{theorem--single_agent_optimality} and \ref{theorem:invertibilityAN}
		that $\abs{\wt{G}^a(t,0)}^{1/2} \in \cH_{BMO}$ is satisfied, 
		ensuring the well-posedness of the minimized-risk BSDEs.
	\end{remark}

\subsubsection{The reference case of a single agent that cannot trade in the derivative} 
\label{subsubsection-example-single.agent.no.derivative}

It is also instructive, and will be useful later on, to look at the case where this single agent cannot trade in the derivative, and hence faces an incomplete market. 
We first enforce $\pi^{a,2}=0$ on \eqref{eq-version2-Gab},	then we optimize over $\pi^{a,1}$ (see Remark \ref{remark---ParticularCaseTheo3.3}). 
The minimized driver following the calculations is 
	\begin{align*}
		\wt{G}^a(t,z) = -\frac{\gamma_a}{2} (\theta^S_t)^2 - z^1 \, \theta^S_t + \frac{1}{2\gamma_a}(z^2)^2, \quad t\in [0,T].
	\end{align*}
Notice that $\wt{G}^a$ is affine in the variable $z^1$ but retains the quadratic term in $z^2$.
The minimized risk is then given by $Y^a_0 = \wt{Y}^a_0$ where $(\wt{Y}^a,\wt{Z}^a)$ is the solution to the BSDE with terminal condition $-H^a$ and the above driver $\wt{G}^a$, while the optimal strategy is
	\begin{align*}
		\pi^{*,a,1} 
	 		= \frac{  \wt{Z}^{a,1} + \gamma_a \theta^S }{\sigma^S S}  
		\quad \text{and}\quad 
		\pi^{*,a,2} = 0 .
	\end{align*}

\subsubsection{The case of multiple agents without relative performance concerns}
\label{subsubsection:zeroconcern}

We return to the full set of agents $\bA$ and take $\lambda^a=0$ for all $a\in \bA$; this is the setting covered in \cite{HorstPirvuDosReis2010}. We find the minimized risk-driver for agent $a$ to be 
	\begin{align*}
		\wt{G}^a(t,z^a) 
		=
		-\frac{\gamma_a}{2} \abs{\theta_t}^2 - \scalar{z^a}{\theta_t}, \quad t\in [0,T].
	\end{align*}
The minimized risk is given by $Y^a_0 = \wt{Y}^a_0$, where $(\wt{Y}^a,\wt{Z}^a)$ solves the BSDE with terminal condition $-H^a$ and minimized driver $\wt{G}^a$, while the optimal strategies are given by 
	\begin{align}		
	\label{eq-optimPi-HPdR2010}
		\pi^{\lambda=0,a,1}  
			:= 
			\frac{  \wt{Z}^{a,1} + \gamma_a \theta^S }{\sigma^S S} 
			- \frac{ \wt{Z}^{a,2} + \gamma_a \theta^R }{\kappa^R} 
			            \, \frac{\kappa^S}{\sigma^S S} 
		\quad \text{and}\quad 
		\pi^{\lambda=0,a,2} 
			:= 
			\frac{ \wt{Z}^{a,2} + \gamma_a \theta^R  } {\kappa^{R}}. 
	\end{align}
Observe that in this case the strategy $\pi^{\lambda=0,a}$ followed by $a$ does not depend \emph{directly} on the strategies of the other agents; its structure is the same as for the single agent case. However, when the price dynamics of the derivative is not fixed but emerges from the equilibrium, later on, the other agents' strategies will appear \emph{indirectly} via $\theta^R$ and $\kappa$.

\subsubsection{The case of multiple agents without relative performance concerns in zero net supply}
\label{subsubsec-prescience-on-shortcut}

If one would want to take into account the endogenous trading of the derivative in the particular situation of pure risk trading, where one takes $n=0$ in \eqref{nnetsupplycondition}, then the market price of external risk $\theta^R$ must be endogenously computed instead of being fixed arbitrarily as we have done so far.  

It is not difficult to see, summing the last equation in \eqref{eq-optimPi-HPdR2010} over $a\in\bA$ and imposing the zero-net supply condition, $\sum_a \pi^{\lambda=0,a,2} = 0$, that this requires that $\theta^R = - {\sum_a \wt{Z}^{a,2}} / \sum_a \gamma_a$. However the $\wt{Z}^{\cdot,2}$s are themselves found by solving a system of $N$ BSDEs which involve $\theta^R$. Replacing $\theta^R$ by the expression above in the said system of equations leads to a fully coupled system of quadratic BSDEs that is hard to solve in general. We solve this problem with an alternative tool in Section \ref{sec4-repagent}.

\subsubsection{The general case: multiple agents with performance concerns} 
\label{example-general-lambdas}

In the general case (not assuming $\lambda^a=0$ for all $a$), we obtain the minimal driver as being still 
	\begin{align}		
	\label{equation-minimized.risk.driver.for.entropic.agents-general.case}
		\wt{G}^a(t,z^a) = -\frac{\gamma_a}{2} \abs{\theta_t}^2 - \scalar{z^a}{\theta_t}, \quad t\in [0,T].
	\end{align}
The minimized risk is then given by $Y^a_0 = \wt{Y}^a_0$ where $(\wt{Y}^a,\wt{Z}^a)$ is the solution to the BSDE with terminal condition $-H^a$ and driver $\wt{G}^a$, while the optimal strategies $\pi^* = (\pi^{*,a})_{a\in\bA}$ are given by
	\begin{align}		
		\label{eq-optimPi-version2-pia1}
		\pi^{*,a,1} - \wt{\lambda}^a \sum_{b \in \bA \setminus \{a\}} \pi^{*,b,1}
				&= \frac{\wt{Z}^{a,1} + \gamma_a \theta^S}{\sigma^S S} - \frac{\wt{Z}^{a,2} + \gamma_a \theta^R}{\kappa^R} \, \frac{\kappa^S}{\sigma^S S},
		\\
		\label{eq-optimPi-version2-pia2}
		\pi^{*,a,2} - \wt{\lambda}^a \sum_{b \in \bA \setminus \{a\}} \pi^{*,b,2}
				&= \frac{\wt{Z}^{a,2} + \gamma_a \theta^R}{\kappa^R}.
	\end{align}
The general invertibility of the systems \eqref{eq-optimPi-version2-pia1} and \eqref{eq-optimPi-version2-pia2} given $\theta$ is guaranteed by Proposition \ref{theorem:invertibilityAN}. 

\subsubsection{The general case: multiple agents with performance concerns in zero net supply}
\label{example-general-lambdas-with-zero-supply}

If one imposes \eqref{nnetsupplycondition} with $n=0$, implying that $\sum_{b \in \bA \setminus \{a\}} \pi^{*,b,2} = -\pi^{*,a,2}$, then the linear system \eqref{eq-optimPi-version2-pia2} for the investment in the derivative simplifies greatly and its solution is explicitly given by 
\begin{align}		
\label{eq-optimPi-version2-pia2-if.net.supply.condition.holds}
\pi^{*,a,2} = \frac{1}{1+\wt{\lambda}^a} \ \frac{\wt Z^{a,2} + \gamma_a \theta^R}{\kappa^R}
\qquad \text{for all $a \in \bA$}.
\end{align}
Notice how the structure of the optimal investment strategy for the derivative in \eqref{eq-optimPi-version2-pia2-if.net.supply.condition.holds} is that of \eqref{eq-optimPi-HPdR2010}, scaled down by the factor $({1+\wt{\lambda}^a})^{-1}$. In Section \ref{section-NumSection} we study a model with two agents and computations will be done explicitly for the investment in the stock (i.e. the inversion of the system \eqref{eq-optimPi-version2-pia1}).

\subsection{Reduction to zero net supply}

We now give an auxiliary result allowing to simplify the condition \eqref{nnetsupplycondition}. We show how the initial holdings $\pi^{a,2}_{0^-} = \pi^{a,2}_0\neq 0$ before/at the beginning of the game can be reduced to the case where $\pi^{a,2}_{0^-} = \pi^{a,2}_0= 0$. This allows us to apply \eqref{nnetsupplycondition} with $n=0$, which will prove crucial in later computations.  The reduction to $n=0$ is based on the monotonicity of the risk measures and the following lemma, stated from the point of view of one agent $a \in \bA$.  The result is based on Lemma 3.9 in \cite{HorstPirvuDosReis2010}.  

To avoid a notational overload, we omit explicit dependencies on $\pi^{-a}$ in this subsection.
\begin{lemma} 
	For a given MPR $\theta$ and admissible strategies $\pi^{-a} = (\pi^b)_{b \in \bA \setminus \left\{a\right\}}$,  consider the dynamics of the residual risk BSDE 
	\begin{align}
	\label{BSDE2}
		-\ud\wt{Y}_t^a(\pi^a) 
		&
		= \wt{g}^a\big(t,\pi^a_t,\wt{Z}_t^a(\pi^a)\big) \udt 
		- \langle \wt{Z}_t^a(\pi^a) , \ud W_t \rangle 
	\end{align}
associated with the preferences of agent $a$ using an admissible strategy $\pi^a$. Assume further that \eqref{BSDE2} has a unique solution for any given $\cF_T$-measurable bounded terminal condition $\wt{Y}_T$. 
Let $\nu \in \bR$. Then,
\begin{itemize}
	\item if ${\pi}^{a} := ({\pi}^{a,1},{\pi}^{a,2})$ minimizes the solution $\wt{Y}_0({\pi}^a)$ to \eqref{BSDE2} for a terminal condition $-H^a$,\\ 
		then $\check{\pi}^a := ({\pi}^{a,1},{\pi}^{a,2}-\nu)$ is optimal for the terminal condition $-(H^a+\nu H^D)$;
	\item if ${\pi}^{a} := ({\pi}^{a,1},{\pi}^{a,2})$ minimizes the solution $\wt{Y}_0({\pi}^a)$ for a terminal condition $-(H^a+\nu H^D)$,
	\\  
		then $\wh{\pi}^a := ({\pi}^{a,1},{\pi}^{a,2}+\nu)$ is optimal for the terminal condition $-H^a$.
\end{itemize}
\end{lemma}
\begin{proof}
We prove only the first assertion, as the second is equivalent.
Let $t\in[0,T]$. Assume that ${\pi}^{*,a}\in\cA^\theta(\pi^{-a})$ is optimal for \eqref{BSDE2} with $\wt Y^a_T:=-H^a$, i.e. for any $\pi^a\in\cA^\theta$ one has $\wt Y^{a}_0({\pi}^{*,a})\leq \wt Y^a_0(\pi^a)$. Define further, for any $\pi^a\in\cA^\theta$, the strategies 
\begin{align*}
\check{\pi}^a :
               =\pi^{a}-(0,\nu)
							 =(\pi^{a,1}, \pi^{a,2}-\nu)
\qquad\text{and}\qquad
\check{\pi}^{*,a}:=\pi^{*,a}-(0,\nu).
\end{align*}
To show that $Y^{a}_0(\check{\pi}^{*,a})\leq Y^a_0(\check\pi^a)$ for any $\check \pi^a$ where $Y^{a}$ solves \eqref{BSDE2} with $Y^{a}_T=-(H^a+\nu H^D)$ we first show an identity result between the BSDEs with different terminal conditions. The second step is the optimality.

\emph{Step 1:} We show that the process $\left(Y(\check{\pi}^a),Z(\check{\pi}^a)\right) := (\tY^{a}(\pi^a) - \nu B^{\theta}, \tZ^a(\pi^a) - \nu \kappa^{\theta})$ solves BSDE
\begin{align} 
\label{Claim0sum}
	Y_t(\check{\pi}^a) &= -(H^a+ \nu H^D) + \int_t^T \wt{g}^a\left(s,\check\pi^{a}_s, Z_s(\check{\pi}^a)\right) \uds - 
\int_t^T \langle Z_s(\check{\pi}^a),\udws\rangle .
\end{align}
To this end, we reformulate \eqref{Bond-Price-MRT-decomposition} as a BSDE:
\begin{align} 
\label{BSDEforDerivativeHD}
	B_t^{\theta}
	 = H^D  
	- \int_t^T \langle  \kappa^\theta_s, \theta_s \rangle \uds
	- \int_t^T \langle \kappa_s^\theta, \ud W_s\rangle .
\end{align}
The difference between \eqref{BSDE2} and $\nu$  times \eqref{BSDEforDerivativeHD} yields 
\begin{align*}
	\tY_t(\pi^a)-\nu B_t^{\theta} 
	&
	=- (H^a+\nu H^D) 
	 + \int_t^T \left[\wt{g}^a(s,\pi^{a}_s,\tZ^a_s(\pi^a))
	              +\nu\langle  \kappa^\theta_s, \theta_s \rangle\right] \uds 
	- \int_t^T \langle \tZ^a_s(\pi^a) - \nu \kappa^{\theta}_s, \udws \rangle 
	\\
	\Leftrightarrow\quad	Y_t(\check\pi^a)
	&
	=- (H^a+\nu H^D) 
	 	- \int_t^T \langle Z_s(\check \pi^{a}), \udws\rangle 
		\\
		&
		\qquad 
		+ \int_t^T \left[
		   \wt{g}^a\left(s,\check \pi^{a}_s+(0,\nu)
		                  , Z_s(\check \pi^{a})+\nu\kappa^\theta_s
							\right)
	              +\nu\langle \kappa^\theta_s, \theta_s \rangle\right] \uds. 
\end{align*}
In view of \eqref{eq-version2-Gab}, we can manipulate the terms inside driver $\wt{g}^a$ above and obtain 
\begin{align*}
	  	&{} 
			\wt{g}^a\left(\cdot,\check \pi^{a}+(0,\nu)
		                  , Z^a(\check \pi^a)+\nu\kappa^\theta\right)
	              +\nu\lev  \kappa^\theta, \theta \rev
	  	\\
		&
		\qquad 
		 = g^a \Big(\cdot,(Z^a(\check \pi^a) +\nu \kappa^\theta) 
								   - \check\pi^{a,1} {\sigma}						
								   -(\check\pi^{a,2} +\nu) {\kappa^\theta}
									 +
									\wt \lambda^a \big( \bar{\pi}^{-a,1}\sigma
									               + \bar{\pi}^{-a,2} \kappa^\theta\big)
									\Big)
		  \\
			& \qquad \qquad \qquad 						
				-\check\pi^{a,1} \lev \sigma ,\theta\rev 
				-(\check\pi^{a,2}+\nu) \lev \kappa^{\theta}, \theta \rev 
				+ \wt \lambda^a 
			         \lev\bar{\pi}^{-a,1}\sigma + \bar{\pi}^{-a,2} \kappa^\theta , \theta\rev
				+\nu\lev  \kappa^\theta, \theta \rev
			\\
			&
			\qquad
			= \wt{g}^a\big(\cdot,\check \pi^{a}
		                  , Z^a(\check \pi^a)\big).
\end{align*}
Given the assumed uniqueness of BSDE \eqref{BSDE2} the assertion follows.

\emph{Step 2:} Given that $\left(Y(\check{\pi}^a),Z(\check{\pi}^a)\right)$ solve \eqref{Claim0sum} and that $\pi^{*,a}$ is the minimizing strategy for $\pi^a\mapsto \tY^{a}(\pi^a)$, then manipulating $Y(\check{\pi}^a)= \tY^{a}(\pi^a) - \nu B^{\theta}$, we have
\begin{align*}
Y_0(\check{\pi}^a)
& = 
  \tY^{a}_0(\pi^a) - \nu B^{\theta}_0
  \geq 
	\tY^{a}_0(\pi^{*,a}) - \nu B^{\theta}_0 
  = 
	Y_0(\check\pi^{*,a}),
\end{align*}
and hence $\check\pi^{*,a}:=\pi^{*,a}-(0,\nu)$ is optimal for BSDE \eqref{BSDE2} with terminal condition $-(H^a+\nu H^D)$.
\end{proof}
This lemma intuitively states that an agent $a$, owning at time $t=0$ a portion $\nu^a=\pi^{a,2}_{0^-}=\pi^{a,2}_{0}$ of units of $H^D$, can be regarded as being in fact endowed with $\check{H}^a = H^a + \nu^a H^D$. One then looks only at the relative portfolio $\check{\pi}^{a,2}=\pi^{a,2}-\nu^a$, which counts the derivatives bought and sold only from time $t=0$ onwards: the optimization problem is equivalent. The argument can be extended to all other agents.
We note that this reduction is only possible because we do not consider trading constraints in this work, so that the strategies $\pi^{a,2}$ and $\check{\pi}^{a,2}$ are equally admissible. 

For the rest of this work we assume that each agent receives at $t=T$ a portion\footnote{Many possibilities for this reduction to zero net supply exist, including endowing one agent with the total amount $n$ of derivatives $H^D$ or endowing each agent with their initial portions of the derivative $\nu^a$. We make the judicious choice of $n/N$ for simplicity.} $n/N$ of the derivative $H^D$. By doing so, the market clearing condition in Definition \ref{def-equilibrium} transforms into 
\begin{align*}
\sum_{a\in\bA} \pi^{a,2}_t=0 \quad \bP \otimes Leb-a.e.,
\end{align*}
and we refer to it as the \emph{zero net supply condition}. 

For clarity, we recall that agent $a\in \bA$ now assesses her risk by solving the dynamics provided by BSDE \eqref{BSDE_1}
with terminal condition 
\begin{align} 
\label{BSDE_1_terminal_0netsupply}
		Y^{a}_{T} 
		&= -\Big(H^{a} + \frac n{N} H^D
		                +V^{a,\theta}_T(\pi^a)  - \wt{\lambda}^a
				\sum_{b \in \bA \setminus\{a\}} V_T^{b,\theta}(\pi^b)\Big)
\end{align} 
(instead of that in \eqref{BSDE_1}). Moreover, by applying the change of variables \eqref{eq:BSDEchangeVArs} to BSDE \eqref{BSDE_1} with terminal condition \eqref{BSDE_1_terminal_0netsupply}, we reach
\begin{align}
\label{eq:SingleAgentBSDE} 
-\ud\wt Y_t^a &= \wt g^a(t,\pi^a_t,\pi^{-a}_t,\wt Z_t^a)\udt -\langle \wt Z_t^a, \ud W_t\rangle , 
\qquad 
\wt Y_T^a := - \left(H^a + \frac n {N} H^D\right),
\end{align}
with $\wt g^a$ given by \eqref{eq-version2-Gab} (and $(\wt Y^a,\wt Z^a)$ relates to $(Y^a,Z^a)$ via the change of variables \eqref{eq:BSDEchangeVArs}).

It is straightforward to recompile the results of Section \ref{subsection-unconstrained.case-entropic.case} under the \emph{zero net supply} condition. It entails no changes in the strategies or drivers, only the terminal condition of the involved BSDEs need to be updated from $-H^a$ to $-(H^a+\frac nN H^D)$ as in \eqref{eq:SingleAgentBSDE}.

%
%
%
%
%
%
\section{The equilibrium market price of external risk} 
\label{sec4-repagent}

In the previous section we saw how to compute the Nash equilibrium for a given market price of risk $\theta = (\theta^S,\theta^R)$, without the global constraint on trading (market clearing condition). In this section we solve the equilibrium problem, as posed by Definition \ref{def-equilibrium}, by finding the Equilibrium Market Price of external Risk (EMPeR) $\theta^R$. 

The literature contains many results on equilibria in complete markets that link competitive equilibria to an optimization problem for a representative agent, and this is the approach we use here. The preferences of the representative agent are usually given by a weighted average of the individual agents' preferences with the weights depending on the competitive equilibrium to be supported by the representative agent, see \cite{Negishi1960}. This dependence results in complex fixed point problems which renders the analysis and computation of equilibria quite cumbersome. The many results on risk sharing under translation invariant preferences, in particular \cites{BarrieuElKaroui2005,JouiniSchachermayerTouzi2006,FilipovicKupper2008}, suggest that when the preferences are translation invariant, then all the weights are equal. This was an effective strategy in \cite{HorstPirvuDosReis2010} and it would be so here if, for all $a\in\bA$, $\lambda^a=0$, or $\lambda^a=\lambda\in[0,1)$. 

In a market without performance concerns, \cites{HorstPirvuDosReis2010,BarrieuElKaroui2009} show that the infimal convolution of risk measures gives rise to a suitable risk measure for the representative agent which, for $g$-conditional risk measures, corresponds to infimal convolution of the drivers. Due to the performance concerns, we use a weighted-dilated infimal convolution and, in Theorem \ref{theorem:EMPR_aggregate_2_2} we show that indeed minimizing the risk of our representative agent is equivalent to finding a competitive equilibrium in our market.

\subsection{The representative agent}

\subsubsection*{Aggregation of risks and aggregation of drivers} 

Inspired by the above mentioned results and having in mind \cite{Rueschendorf2013} (see Remark \ref{justify_weighted_convolution} below) we deal with the added inter-dependency arising from the fixed-supply condition and the additional unknown $\theta^R$ (see examples in \ref{subsubsec-prescience-on-shortcut} and \ref{example-general-lambdas-with-zero-supply}) by defining a new risk measure $\rho^w_0$. For a set of positive weights $w=(w^a)_{a\in\bA}$ satisfying $\sum_{a\in\bA} w^a = 1$, we define 
\begin{align}
\label{eq:defi-of-inf-conv-rhow}
\rho^w_0(X) 
=
\inf\left\{ \sum_{a\in\bA} w^a \rho^a_0\big(X^a\big)
 \;\bigg|\; 
(X^a) \in (L^\infty)^{N} : \sum_{a\in\bA} w^a X^a = X \right\} 
\quad\text{for any } X \in L^\infty.
\end{align}
In the case of risk measures induced by BSDEs, \cite{BarrieuElKaroui2005} shows that the measure defined by inf-convolution of risk measures $(\rho^a_0)_{a\in\bA}$ is again induced by a BSDE, whose driver is simply the inf-convolution of the BSDE drivers $g^a$ for the risk measures $(\rho^a_0)_{a\in\bA}$. 
For the set of weights $w=(w^a)_{a\in\bA}$, we define the driver $g^w$ as the \emph{weighted-dilated} inf-convolution of the drivers $g^a$ for $(t,z) \in [0,T] \times \bR^2$,
\begin{align}
\label{eq:weightedinfconv}
	g^w(t,z) = \square_w \Big( (g^{a})_{a\in\bA} \Big)(t,z) = \inf \left\{\sum_{a \in \bA} w^a g^a(t,z^a) \; \bigg| \; (z^a) \in (\bR^2)^{N} \text{ s.t. } \sum_{a \in \bA} w^a z^a = z\right\} ,
\end{align}
where the notation $\square \big( (g^{a})_{a\in\bA} \big)$ is that of the standard inf-convolution. 
	\begin{lemma}[Properties of $g^w$] 
	\label{lemma---properties.of.the.infconvol.driver.and.optimizer}
		The map $g^w:[0,T]\times \bR^2\to \bR$ defined by \eqref{eq:weightedinfconv} is a deterministic continuous function, strictly convex and  continuously differentiable. 
		Moreover, there exists a unique solution of $\nabla_z g^w(t,\cZ) =-\vartheta$ in $\cZ$. 
		
		For a $z^\bA=(z^a)$ such that $\sum_a w^a z^a = z$, one has $g^w(t,z) = \sum_{a} w^a g^a(t,z^a)$ if and only if there exists $\vartheta \in \bR^2$ such that, for all $a \in \bA$,
		$\nabla_z g^a(t,z^a) = -\vartheta$. In that case, one has $\nabla_z g^w(t,z) = -\vartheta$.
	\end{lemma}
	\begin{proof}
		The weighted inf-convolution transfers the properties of the $g^a$'s to $g^w$, in particular continuity, strict convexity and differentiability. We do not show these as they follow from a simple adaptation of known arguments, see \cites{BarrieuElKaroui2005,BarrieuElKaroui2009,HorstPirvuDosReis2010}. 

		Since the function being minimized ($z^\bA=(z^a) \mapsto \sum_a w^a g^a(z^a)$) is convex and the function defining the constraint ($z^\bA \mapsto \sum_a w^a z^a$) is also convex, because affine, 
		the minimization defining $g^w$ is equivalent to finding a critical point for the associated Lagrangian, $L(z^\bA,\vartheta) = \sum_a w^a g^a(z^a) + \vartheta (\sum_a w^a z^a - z)$.
		Therefore, for a $z^\bA=(z^a)$ such that $\sum_a w^a z^a = z$, $z^\bA$ is a minimizer 
		if and only if there exists $\vartheta \in \bR^2$ such that, for all $a \in \bA$,
		$\nabla_z g^a(t,z^a) = -\vartheta$.  
		Then, $\nabla_z g^w(t,z) = -\vartheta$ where $\vartheta$ is the Lagrange multiplier associated with $z$.
	\end{proof}

\paragraph*{}
The risk of the random terminal wealth $\xi^w$, measured through $\rho^w_0$, is given by $\rho^w_0(\xi^w) := Y^w_0$ where $(Y^w,Z^w)$ is the solution to the BSDE
	\begin{align}		
	\label{equation-risk.BSDE.for.the.representative.agent}
		- \ud Y^w_t = g^w(t,Z^w_t) \udt - \langle Z^w_t, \udwt\rangle ,
		\qquad\text{with terminal condition}\qquad
		Y^w_T = - \xi^w .
	\end{align}
Since the weights $(w^a)_{a\in\bA}$ are required to satisfy $\sum_a w^a = 1$, the risk measure $\rho_0^w$ associated to the BSDE with the above driver is a monetary risk measure. Translation invariance and monotonicity follow from the fact that the driver $g^w$ is independent of $y$. Convexity follows from the convexity of $g^w$, which in turns follows from that of the $g^a$'s by the envelope theorem.

\begin{remark} 
Notice that \eqref{eq:weightedinfconv} can be rewritten
\begin{equation*}
	g^w(t,z) = \inf \left\{\sum_{a \in \bA} w^a g^a\big(t,\frac{z^a}{w^a}\big) \; \bigg| \; \sum_{a \in \bA} z^a = z\right\} .
\end{equation*}
In this way, $g^w$ is seen as the usual $w$-weighted infimal convolution of the $w^a$-dilated drivers $g^a$, in the terminology from \cite{BarrieuElKaroui2009} (p.137).
For more on dilated risk measures, see Proposition 3.4 in \cite{BarrieuElKaroui2009}.
\end{remark}

\begin{example}[Entropic risk measure]	
For entropic agents, i.e. with drivers $g^a(z^a) = \frac{|z^a|^2}{2 \gamma_a}$, one obtains 
	\begin{align}
	\label{eq:defi-gammaR}
		g^{w}(z) 
		= 
		\frac{\abs{z}^2}{2\gamma_R},
		\quad \text{with} \quad 
		\gamma_R
		:	= 
		\sum_{a\in\bA} w^a \gamma_a .
	\end{align}
\end{example}
\subsubsection*{Trading and the risky position of the representative agent}
Having defined the aggregated risk measure $\rho_0^w$ and the associated driver $g^w$, we now introduce a strategy $\pi^w$ and associated trading gains $V_\cdot(\pi^w) = \int_0^\cdot \pi^{w,1}_t \ud S_t +\int_0^\cdot \pi^{w,2}_t \ud B_t$ for a representative agent whose preferences are described by $g^w$. Direct computations from \eqref{eq:defi-of-inf-conv-rhow} entail that we assign to the representative agent the terminal gains
\begin{align*}
		\xi^w 
		:= 
		\sum_{a\in\bA} w^a \xi^a 
		&= 
		\sum_{a\in\bA} w^a \bigg( 
		 {H}^a + \frac nN H^D + V^a_T - \wt{\lambda}^a \bar V^{-a}_T 
		                   \bigg) 
    \\
		& 
		= 
		\sum_{a\in\bA} w^a ({H}^a + \frac nN H^D) 
		 + \sum_{a\in\bA} w^a\Big( (1+\wt{\lambda}^a) V^a_T - \wt{\lambda}^a \sum_{b\in\bA} V^{b}_T  \Big)
		\\
		&
		=
		\sum_{a\in\bA} w^a ({H}^a + \frac nN H^D) 
		 + \sum_{a\in\bA} V^a_T\Big( (1+\wt{\lambda}^a) - \wt{\lambda}^a \sum_{b\in\bA} w^b \wt \lambda^b  \Big)
		= 
		{H}^w + V_T(\pi^w),
\end{align*}
where $c^a := w^a (1+\wt{\lambda}^a) - \sum_{b\in\bA} w^b\wt{\lambda}^b$, $\pi^w = \sum_{a\in\bA} c^a \pi^a$ is the representative agent's portfolio, $V_T(\pi^w)=\sum_{a\in\bA} c^a V_T(\pi^a)$ is the representative agent's wealth process and 
\begin{align}
\label{eq:rep-ag-term-cond}
	{H}^w &:= \sum_{a\in\bA} w^a ({H}^a+\frac nN H^D) =  \frac nN H^D+  \sum_{a\in\bA} w^a {H}^a
\end{align}
 is defined as the representative agent's terminal endowment.

We now choose the weights $(w^a)_{a\in\bA}$ such that $c^a=c$ for any $a\in\bA$ for some $c\in (0,+\infty)$, i.e. 
\begin{equation}
\label{eq:defi-weightsw}
		w^a 
		:=  
		\frac{1}{\Lambda(1+\wt{\lambda}^a)}  
		\quad\text{for all $a\in\bA$, where }\quad 
		\Lambda 
		:= 
		\sum_{a\in\bA} \frac{1}{1+\wt{\lambda}^a} .
\end{equation}
Direct verification yields $\sum_a w^a = 1$ and, furthermore, for all $a\in\bA$, 
\begin{equation*}
	c^a 
	=
	c
	:= 
	\frac1\Lambda - \frac1\Lambda \sum_{b\in\bA} \frac{\wt{\lambda}^b}{1+\wt{\lambda}^b}
	.
\end{equation*}
Notice that $	\pi^{w,2} =  \sum_{a\in\bA} c^a \pi^{a,2} = c \sum_{a\in\bA} \pi^{a,2}$. In other words, the zero net supply condition for the individual agents (i.e. $\sum_{a\in\bA} \pi^{a,2}=0$) is equivalent 
to the representative agent not investing in $H^D$ (i.e. $\pi^{w,2}=0$).
From now on, the family of weights $w$ is fixed and is given by \eqref{eq:defi-weightsw}.

\subsubsection*{The pointwise minimizer for the representative agent's residual risk}
We now show that the approach by aggregated risk and representative agent, as motivated above, allows to identify the equilibrium market price of risk as a by-product of minimizing the risk of the representative agent. This risk is given by the solution to BSDE \eqref{equation-risk.BSDE.for.the.representative.agent} with terminal condition $Y^w_T  = - \xi^w = -{H}^w - V_T(\pi^w)$, for admissible strategies $\pi^w$ of the form $\pi^w=(\pi^{w,1},0)$.  
The $\bR^2$-valued strategy process $\pi^w$ is said to be admissible ($\pi^w \in \cA^w$) if $\bE^{\theta}  \left[\lev \; V_\cdot(\pi^w) \; \rev_{T}\right]< \infty$  
 and the BSDE \eqref{equation-risk.BSDE.for.the.representative.agent} has a unique solution. Following Section \ref{section-SingleAgentSection} we introduce the residual risk processes 
	\begin{align*} 
		\wt{Y}^w_t 
		:= 
		Y^w_t + V^w_t 
		\qquad\text{and accordingly}\qquad 
		\wt{Z}^w_t 
		:= 
		Z^w_t + \big( \pi^{w,1}_t\sigma_t + 0 \big).
	\end{align*}
The pair $(\wt{Y}^w,\wt{Z}^w)$ satisfies the BSDE with terminal condition $\wt{Y}^w_T=-{H}^w$ and random driver $\wt{g}^w$, defined for $(\omega,t,\pi^w_t,z) \in \Omega\times [0,T] \times \bR^2 \times \bR^2$, by
	\begin{align*}		
		\wt{g}^{w}(t,\pi^{w}_t,z) 
		:= 
			{g}^{w}\big(t,z - \zeta^w_t \big) 
			- \scalar{\zeta^w_t \,}{\, \theta_t} ,
	\end{align*}
where $\zeta^w = \pi^{w,1} \sigma + 0$ (compare with \eqref{eq:BSDEchangeVArs}-\eqref{eq-version2-Gab}).
Since $\wt{Y}^w_0 = Y^w_0$, the representative agent then equivalently aims at solving for $\min\{ \wt{Y}^w_0(\pi^w) | \pi^w \in \cA^w \}$.

Following the methodology used for the single agent in Section \ref{section-SingleAgentSection}, we first look at minimizing the driver $\wt{g}^w$ pointwise.
We define $\Pi^{w,1}(t,z)$ as the optimizer for $\min\big\{ \wt{g}^w(t,(p,0),z) \,|\, p \in \bR \big\}$, setting $\Pi^{w,2}(t,z)=0$ as to enforce the zero-net supply condition. 
Since $g^w$ is strictly convex, so is the function $\wt{g}^w$, and the minimum is characterized by the solution to first-order condition
		\begin{align*}
			g^{w}_{z^1}\left(t, z - \Pi^{w,1}(t,z) \, \sigma_t \right) = - \theta^S_t .
		\end{align*}
We denote the minimized (random) driver by
	\begin{align}
	\label{eq:MinimizedDriverGw}
		\wt{G}^w(t,z) = \wt{g}^{w}\big(t,\Pi^{w}(t,z),z\big).
	\end{align}

	\begin{remark}[The structure for the optimized driver \eqref{eq:MinimizedDriverGw} under a separation assumption]
		Here, unlike for the optimization of individual agents who trade in $S$ and $B$ under a fixed MPR $\theta=(\theta^S,\theta^R)$, we do not have a nice structure like in \eqref {equation---minimized.driver.wtGa.for.general.ga} for $\wt{G}^w$ in all generality on $g^w$ (hence on the $g^a$'s). 
		
		Assume that for some $g^{1,w},g^{2,w}: [0,T]\times \bR \to \bR$ we have $g^w(t,z) = g^{1,w}	(t,z^1) + g^{2,w}(t,z^2)$, then the first-order condition would translate to $g^{1,w}_{z^1}\left(t, z^1 - \Pi^{w,1}(t,z) \sigma^S S_t \right) = - \theta^S_t $. 
		Denoting by $\cZ^{w,1}(t,-\theta^S_t)$ the solution in $\cZ^1 \in \bR$ to the equation $g^{1,w}_{z^1}\left(t, \cZ^1 \right) = - \theta^S_t$, the structure for $\Pi^{w,1}$ is given by
			\begin{align*}
				\Pi^{w,1}(t,z) = \frac{z^1 - \cZ^{w,1}(t,-\theta^S_t)}{\sigma^S S_t}
			\end{align*}
		and the structure for the optimized driver $\wt{G}^w$ is given by 
			\begin{align*}
				\wt{G}^w(t,z) 
				&
				= 
				 g^{w}\big(  t, (\cZ^{w,1}(t,-\theta^S_t),z^2) \big) 
				 + \cZ^{w,1}(t,-\theta^S_t) \theta^S_t -  z^{1} \theta^S_t
				\\
				&= \big[ g^{1,w}\big(  t, \cZ^{w,1}(t,-\theta^S_t) \big) 
				     + \cZ^{w,1}(t,-\theta^S_t) \theta^S_t  \big]
				     	-  z^{1} \theta^S_t 
						 + g^{2,w}(t, z^2) .
			\end{align*}		
			The special case of entropic drivers, that falls in this category, is discussed below in Example \ref{example_EntropicOptimizedDriver}.
	\end{remark}

\subsubsection*{Optimimality for the representative agent and the equilibrium market price of external risk} 

We assume that the BSDE with driver $\wt{G}^w$ defined in \eqref{eq:MinimizedDriverGw} and terminal condition $-H^w$ has a unique solution $(\wt{Y}^w,\wt{Z}^w)$ in $\cS^\infty \times \cH_{BMO}$.
Define the strategy $\pi^{*,w}$ by $\pi^{*,w}_t := \big( \Pi^{w,1}(\omega,t,\wt{Z}^w_t) , 0 \big)$.
Like for the individual agents in Section \ref{section-SingleAgentSection}, the following theorem asserts that $\pi^{*,w}$ is the optimal strategy and $\wt{Y}^w_0$ is the minimized risk for the representative agent. Moreover, the theorem relates the equilibrium market price of risk (EMPR) $\theta = (\theta^S,\theta^R)$ (recall Definition \ref{def-equilibrium}) to the solution of the representative agent's optimization problem. 
Recall the family of weights $w$ given by \eqref{eq:defi-weightsw}. 

\begin{theorem}		
\label{theorem:EMPR_aggregate_2_2}
	Assume that
		\begin{itemize}
			\item the BSDE with driver $\wt{G}^w$, \eqref{eq:MinimizedDriverGw}, and 
			$\wt{Y}^w_T=-H^w$ has a unique solution $(\wt{Y}^w,\wt{Z}^w)$ in $\cS^\infty \times \cHBMO$,
			\item the comparison theorem holds for the BSDE with driver $\wt{G}^w$,
      \item $\pi^{*,w}_{\cdot}=\big( \Pi^{w,1}(\omega,\cdot,\wt{Z}^w_\cdot) , 0 \big)$ is integrable against the prices $S$ and $B$,
		\end{itemize}
	then $\wt{Y}^w_0$ is the minimized risk for the representative agent and $\pi^{*,w}$ is the unique optimal strategy that minimizes his risk.
	
	If, for the process $\theta^*=(\theta^S,\theta^{R})$, with $\theta^{R}$ defined by
		\begin{align}
		\label{eq:ConditionDefiningThetaR}
			g^{w}_{z^2}\left(t, \wt{Z}^{w}_t - \pi^{*,w,1}_t \, \sigma_t \right) = - \theta^R_t ,
		\end{align}
the conditions of Theorem \ref{theorem:invertibilityAN} hold, then $\theta^*$ is the unique EMPR for the agents in $\bA$.
	
	Additionally, the minimized aggregated risk $\wt Y^w$ is linked to the individual minimized risks $(\wt{Y}^a)_{a\in\bA}$ through the identity $\wt{Y}^w = \sum_{a} w^a \wt{Y}^a$ (the same holds for $\wt Z$). Moreover, the Nash equilibrium for the agents in $\bA$ satisfies $\pi^{*,w}=c \sum_a \pi^{*,a}$.
\end{theorem}

\begin{example}[The entropic case]
	\label{example_EntropicOptimizedDriver}
		In the entropic case, we have found $g^w(z) = \frac{\abs{z}^2}{2\gamma_R}$, so we have $\cZ^{w,1}(t,-\theta^S_t) = -\gamma_R \theta^S_t$.
		The minimized driver is then
			\begin{equation*}
				\wt{G}^w(t,z) = -\frac{\gamma_R}{2} (\theta^S_t)^2 -  z^{1} \theta^S_t + \frac{1}{2\gamma_R}(z^2)^2,
			\end{equation*}		
		as was found in Subsection \ref{subsection-unconstrained.case-entropic.case}. 
		This driver is quadratic and regular, and the terminal condition $-H^w$ is bounded. From \cites{Kobylanski2000,ImkellerDosReis2010} 
		there is a unique solution $(\wt{Y}^w,\wt{Z}^w)$ in $\cS^\infty \times \cHBMO$ and the comparison theorem applies (see \cites{Kobylanski2000,MaYao2010}). 
		The optimal strategies are 
		\begin{equation*}
			\pi^{*,w,1} = \frac{\wt{Z}^{w,1} + \gamma_R \theta^S}{\sigma^S S} \text{ and }
			\pi^{*,w,2} = 0.
		\end{equation*} 
		With $\wt{Z}^{w} \in \cHBMO$ and $\theta^S$ bounded, $\pi^{*,w,1}$ is integrable against $S$. This verifies the first three assumptions of the theorem.
		Furthermore, with \eqref{eq:ConditionDefiningThetaR} and since $\wt{Z}^{w} \in \cHBMO$ and $\theta^S$ is bounded, we find that
	\begin{equation} \label{eq:entropic-EMPeR}
		\theta^R = -\frac{\wt{Z}^{w,2}}{\gamma_R} 
		\quad \text{ and } \quad
		\theta^*=(\theta^S,\theta^R) \in \cHBMO .
	\end{equation}
\end{example}

Following on Remark \ref{remark---ParticularCaseTheo3.3}, 
the optimality of $\pi^{*,w}$ and $(\wt{Y}^w,\wt{Z}^w)$, for an agent $w$ with preferences described by $g^w$ and trades in $S$, is obtained exactly in the same way as the optimality for a single agent $a \in \bA$ in Theorem \ref{theorem--single_agent_optimality}. 
So we prove only the claims of Theorem \ref{theorem:EMPR_aggregate_2_2} related to the EMPR $\theta^*$.
First, however, we state a counterpart to Lemma \ref{remark---FOC.for.processes.pi.and.Z.implies.optimality} to the case when no trading in $B$ is possible. 

	\begin{lemma} \label{remark---FOC.for.processes.pi.and.Z.implies.optimality-version.when.trading.S.only}
		Under the assumptions of Theorem \ref{theorem:EMPR_aggregate_2_2}, 
		let $\wh{\pi}^w=(\wh{\pi}^{w,1},0)$ be an admissible strategy and $(\wh{Y}^w,\wh{Z}^w)$ be the associated risk process, i.e. the
		solution to the BSDE with driver $\wt{g}^w(t,\wh{\pi}^w_t,\cdot)$ and terminal condition $-H^w$.
		Assume that the FOC holds for these processes, i.e.  
			\begin{align*}
				g^w_{z^1}(t, \wh{Z}^w_t - \wh{\zeta}^w_t) = -\theta^S_t
				\qquad \text{where} \qquad
				\wh{\zeta}^w_t = \wh{\pi}_t^{w,1} \sigma_t .
			\end{align*} 
		Then $(\wh{Y}^w,\wh{Z}^w) = (\wt{Y}^w,\wt{Z}^w)$ and $\wh{\pi}^w = {\pi}^{*,w}$.
		\end{lemma}

	\begin{proof}
		Recalling the properties of $g^w$ (see Lemma \ref{lemma---properties.of.the.infconvol.driver.and.optimizer}) and the definition of $\Pi^{w,1}$, 
		the condition $g^w_{z^1}(t, \wh{Z}^w_t - \wh{\pi}_t^{w,1} \sigma_t ) = -\theta^S_t$ means that
		$\wh{\pi}^{w,1}_t = \Pi^{w,1}(t,\wh{Z}^w_t)$.
		We have then $\wt{g}^w(t,\wh{\pi}^w_t,\wh{Z}^w_t) = \wt{G}^w(t,\wh{Z}^w_t)$  (recall \eqref{eq:MinimizedDriverGw}). 
		By the assumed uniqueness of the solution to the BSDE with driver $\wt{G}^w(t,\cdot)$ and terminal condition $-H^w$, we have $(\wh{Y}^w,\wh{Z}^w)=(\wt{Y}^w,\wt{Z}^w)$.
		Consequently, by the uniqueness of the FOC's solution, $\wh{\pi}^{w,1}_t = \Pi^{w,1}(t,\wh{Z}^w_t) = \Pi^{w,1}(t,\wt{Z}^w_t) = \pi^{*,w,1}_t$.
		Since both strategies have second component equal to zero, we have therefore $\wh{\pi}^w = \pi^{*,w}$.
	\end{proof}
The next result, to be used in the proof of Theorem \ref{theorem:EMPR_aggregate_2_2}, states that aggregating the solutions to the individual optimization problems leads to an optimum for the aggregated preference $g^w$ and identifies the BSDE of the aggregation with the weighted sum of the agents' BSDEs. 
	\begin{lemma}	
	\label{lemma--aggregating.individual.optima.leads.to.optimum.for.aggregated.preferences}
		Let $\vartheta \in \cH_{BMO}$ be a MPR and assume the conditions of Theorem \ref{theorem:invertibilityAN}. 
		Let then $(\pi^{*,a})_{a \in \bA}$ be the unconstrained Nash equilibrium associated with $\vartheta$, and let $(\wt{Y}^a,\wt{Z}^a)$ be the solution to the minimized-risk BSDE for each agent $a\in \bA$ (BSDE \eqref{agentresidualBSDEtransf} with driver \eqref{equation---minimized.driver.wtGa.for.general.ga}). 
		Define $(\wh{Y}^w,\wh{Z}^w) := \sum_a w^a (\wt{Y}^a,\wt{Z}^a)$ and $\wh{\pi}^w := \sum_a c^a \pi^{*,a} = c \sum_a \pi^{*,a}$.

		Then $(\wh{Y}^w,\wh{Z}^w)$ and $\wh{\pi}^w$ are the minimal risk and optimal strategy for a single agent whose preferences are given by $g^w$, who can invest in $(S,B)$ (without trading constraints).
	\end{lemma}

\begin{proof}
		Firstly, we sum the individual risk BSDEs to obtain $(\wh{Y}^w,\wh{Z}^w)$ and its BSDE. We have $\wh{Y}^w_T = -\sum_a w^a H^a = -H^w$ and also 
			\begin{align*}
				\ud \wh{Y}^w_t 
					&= - \bigg[ \sum_{a\in\bA} w^a \Big\{g^a\Big(t, \wt{Z}^a_t - \zeta^a_t(\pi^*) \Big) - \scalar{\zeta^a_t(\pi^*)}{\vartheta_t}\Big\} \bigg] \ud t
					+ \sum_{a\in\bA} w^a \langle \wt{Z}^a_t ,\ud W_t\rangle
					\\
					&= - \bigg[ \sum_{a\in\bA} w^a g^a\Big(t, \wt{Z}^a_t - \zeta^a_t(\pi^*) \Big) - \scalar{\wh{\zeta}^w_t}{\vartheta_t} \bigg] \ud t + \langle \wh{Z}^w_t ,\ud W_t\rangle  ,				
			\end{align*}
		where $\wh{\zeta}^w = \wh{\pi}^{w,1} \sigma + \wh{\pi}^{w,2} \kappa=\sum_a w^a\zeta^a(\pi^*)$.
		We remark that, on the one hand,
			\begin{align*}
				\sum_a w^a ( \wt{Z}^a_t - \zeta^a_t(\pi^*) ) = \sum_a w^a \wt{Z}^a_t - \Big( \sum_a c^a \pi^{*,a,1} \sigma_t + \sum_a c^a \pi^{*,a,2} \kappa_t \Big)
						= \wh{Z}^w_t - \wh{\zeta}^w_t ,
			\end{align*}
		and, on the other hand, for all $a \in \bA$, by the optimality of $\wt \pi^a$ and $(\wt Y^a, \wt Z^a)$ 
			\begin{align*}
				\nabla_z g^a\big(t, \wt{Z}^a_t - \zeta^a_t \big) = -\vartheta_t . 
			\end{align*}		
		Therefore, we know by Lemma \ref{lemma---properties.of.the.infconvol.driver.and.optimizer} that $g^w(t,\wh{Z}^w_t - \wh{\zeta}^w_t) = \sum_a w^a g^a(t, \wt{Z}^a_t - {\zeta}^a_t)$. 
		This implies that
			\begin{align*}
				\ud \wh{Y}^w_t 
					&= - \big[ g^w(t,\wh{Z}^w_t - \wh{\zeta}^w_t) - \langle{\wh{\zeta}^w_t},{\vartheta_t}\rangle \big] \ud t + \langle \wh{Z}^w_t ,\ud W_t\rangle  
					\\
					&
					= - \wt{g}^w(t,\wh{\pi}^w_t,\wh{Z}^w_t) \ud t + \langle \wh{Z}^w_t ,\ud W_t\rangle  .
			\end{align*}
		Secondly, by Lemma \ref{lemma---properties.of.the.infconvol.driver.and.optimizer}, we also know that  $\nabla_z g^w\big(t, \wh{Z}^w_t - \wh{\zeta}^w_t \big) = -\vartheta_t$.

		Therefore, by Lemma \ref{remark---FOC.for.processes.pi.and.Z.implies.optimality}, we obtain that $(\wh{Y}^w,\wh{Z}^w)$ is the solution to the minimized-risk BSDE
		for an agent with preferences given by $g^w$ from \eqref{equation---minimized.driver.wtGa.for.general.ga}, terminal condition $-H^w$, who trades in $S$ and $B$ under the given MPR $\vartheta$ with $\wh{\pi}^w$ as the optimal strategy.	
	\end{proof}

	\begin{proof}[Proof of Theorem \ref{theorem:EMPR_aggregate_2_2}]
	The first part of the proof of the theorem, the optimization for the representative agent, follows through arguments similar to those used in the single agent case, see Theorem \ref{theorem--single_agent_optimality} and Remark \ref{remark---ParticularCaseTheo3.3}. Hence we omit it.

		\emph{$\triangleright$ Existence of the EMPR.} Here we prove that $\theta^*=(\theta^S,\theta^R)$, defined through \eqref{eq:ConditionDefiningThetaR}, is indeed an EMPR. 
		Since $\theta^* \in \cHBMO$ 
		and the conditions of Theorem \ref{theorem:invertibilityAN} hold, let $(\pi^{*,a})_{a \in \bA}$ be the unique unconstrained Nash equilibrium under the MPR $\theta^*$, 
		and let $(\wt{Y}^a,\wt{Z}^a)$ be the solution to the minimized-risk BSDE for each agent $a\in \bA$. 
		Our goal now is to prove that $\sum_a \pi^{*,a,2}=0$.
		
		Let us introduce $(\wh{Y}^w,\wh{Z}^w) := \sum_a w^a (\wt{Y}^a,\wt{Z}^a)$ and $\wh{\pi}^w := \sum_a c^a \pi^{*,a} = c \sum_a \pi^{*,a}$.
		From Lemma \ref{lemma--aggregating.individual.optima.leads.to.optimum.for.aggregated.preferences}, $\wh{\pi}^w$ and $\wh{Y}^w$ are the optimal strategy and risk for a single agent with risk preferences encoded by $g^w$ trading $S$ and $B$ under $\theta^*$ without trading constraints. 
		
		Meanwhile, we defined $\pi^{*,w} = (\pi^{*,w,1},0)$ as the optimal strategy for an agent $w$ with preferences encoded by $g^w$ and who can only invest in $S$ (with MPR $\theta^S$).
		By construction of $\theta^R$, we have 
			\begin{align*}
				\nabla_z g^w\big(t, \wt{Z}^w_t - \zeta^w_t \big) = -\theta^*_t ,
				\quad \text{where} \quad 
				\zeta^w_t = \pi^{*,w,1}_t \sigma_t.
			\end{align*}		
		It results from Lemma \ref{remark---FOC.for.processes.pi.and.Z.implies.optimality} that $\pi^{*,w}$ 
		is also the optimal strategy for an agent with preferences $g^w$ and who can invest in $S$ and $B$, with given MPR $\theta^*$. 
		By the uniqueness in Lemma \ref{remark---FOC.for.processes.pi.and.Z.implies.optimality} we therefore have $\wh{\pi}^w = {\pi}^{*,w}$. 
		This implies in particular that $\sum_a \pi^{*,a,2} = \wh{\pi}^{w,2} =  \pi^{*,w,2} = 0$. 
		We have therefore proved that the Nash equilibrium associated with $\theta^*$ satisfies the zero-net supply condition, hence the constructed $\theta^*$ is an EMPR.

		\emph{$\triangleright$ Uniqueness of the EMPR.} 
		Assume that $\vartheta=(\theta^S,\vartheta^R)$ is also an EMPR and let $(\pi^{*,a,\vartheta})_{a \in \bA}$ be the associated Nash equilibrium for which, by definition of EMPR, the zero-net supply condition $\sum_a {\pi}^{*,a,\vartheta,2}=0$ is satisfied.
		Let also $(\wt{Y}^{a,\vartheta},\wt{Z}^{a,\vartheta})$ be the solution to the minimized-risk BSDE for each agent $a\in \bA$. 
		As above, we define $(\wh{Y}^{w,\vartheta},\wh{Z}^{w,\vartheta}) := \sum_a w^a (\wt{Y}^{a,\vartheta},\wt{Z}^{a,\vartheta})$ and $\wh{\pi}^{w,\vartheta} := \sum_a c^a \pi^{*,a,\vartheta} = c \sum_a \pi^{*,a,\vartheta}$.	By Lemma \ref{lemma--aggregating.individual.optima.leads.to.optimum.for.aggregated.preferences}, we obtain that $(\wh{Y}^{w,\vartheta},\wh{Z}^{w,\vartheta})$ and $\wh{\pi}^{w,\vartheta}$ are optimal for an agent $w$ who trades in $S$ and $B$ under the given MPR $\vartheta$ for a single agent economy. 
		Consequently, using the characterization between the  optimizer and the FOC condition, we have 
			\begin{align*}
				g^w_{z^1}\big(t, \wh{Z}^{w,\vartheta}_t - \wh{\pi}^{w,\vartheta,1}_t \sigma_t \big) = -\theta^S_t
				\qquad \text{and} \qquad 
				g^w_{z^2}\big(t, \wh{Z}^{w,\vartheta}_t - \wh{\pi}^{w,\vartheta,1}_t \sigma_t \big) = -\vartheta^R_t,
			\end{align*}
where $\wh{\pi}^{w,\vartheta,2}=0$ as $\vartheta$ is an EMPR. 
By Lemma \ref{remark---FOC.for.processes.pi.and.Z.implies.optimality-version.when.trading.S.only}, the first equation guarantees that $(\wh{Y}^{w,\vartheta},\wh{Z}^{w,\vartheta})$ and $\wh{\pi}^{w,\vartheta}$ are optimal for an agent with preferences $g^w$ who trades in $S$.
		By the construction of $ (\wt{Y}^w,\wt{Z}^w)$ and ${\pi}^{*,w}$ (for the MPR $\theta^*$), and the uniqueness recalled in Lemma \ref{remark---FOC.for.processes.pi.and.Z.implies.optimality-version.when.trading.S.only}, we have $(\wh{Y}^{w,\vartheta},\wh{Z}^{w,\vartheta}) = (\wt{Y}^w,\wt{Z}^w)$ and $\wh{\pi}^w = {\pi}^{*,w}$. As a consequence, we have from the second FOC equation 
			\begin{align*}
				-\vartheta^R_t = g^w_{z^2}\big(t, \wh{Z}^{w,\vartheta}_t - \wh{\pi}^{w,\vartheta,1}_t \sigma_t \big) = g^w_{z^2}\big(t, \wt{Z}^w_t - {\pi}^{*,w,1}_t \sigma_t \big) = -\theta^R_t .
			\end{align*}
		Hence the uniqueness of the EMPR $\theta^*$.
	\end{proof}

From Theorem \ref{theorem:EMPR_aggregate_2_2} we point out that $\theta^*$ is only \emph{a} MPR for the representative agent's economy as the representative agent trades in an incomplete market where she is not able trade the risk from $(R_t)$ --- recall \eqref{diffusion-R}. 
Nonetheless, $\theta^*$ is the only MPR leading to a complete market for the agents where the Nash equilibrium they form satisfies the zero net supply condition. 
We close this remark by adding  that the representative agent approach for complete market leads to Arrow-Debreu equilibria for the acting agents (see \cites{HorstPirvuDosReis2010,KardarasXingZitkovic2015} and references therein).

\begin{remark} 
\label{justify_weighted_convolution}
In \cite{Rueschendorf2013} a ``weighted minimal convolution'' of risk measures is introduced via
\begin{align*}
		\left(\bigwedge \rho_i\right)_{\gamma}(X) := 
		\inf \left\{
			\sum_{i=1}^N \gamma_i\rho_i(X_i); \; \; X_1,\ldots,X_N \in L^p, \; \sum_{i=1}^N X_i = X
		\right\}
\end{align*}
(see page 271, Equation (11.25)) for $\gamma=(\gamma_i) \in \bR^N_{>0}$ and for some $p \geq 1$. 

Observe that aggregation in our context would not work without the dilation weights ${1}/{w^a}$ in the argument of the driver. This can be seen in Step 2 of the proof of Theorem \ref{theorem:EMPR_aggregate_2_2}. The reason is that $\wt G^a$ is the sum of $g^a$ with the strategies plugged in as arguments and of 
an additional term with the strategies multiplied by the weights. For the aggregation as a single strategy this adjustment is necessary.
\end{remark}

\subsection{A shortcut to the EMPeR in the case of entropic risk measures} 

In the previous subsection we gave a result on the existence and uniqueness of the equilibrium market price of risk via the inf-convolution of the risk measures, for general preferences. 
In the particular case of the entropic risk measure, the general computations are considerably simpler and an easier path allows to compute what the EMPeR $\theta^R$ is (if it exists) without the representative agent.  Although the BSDE for the representative agent derived above will appear in the following computations, with only these computations one cannot show that the computed $\theta$ is indeed an EMPR.
This shorter path consists, as was hinted in Section \ref{subsubsec-prescience-on-shortcut}, in a direct linear combination of the BSDEs \eqref{agentresidualBSDEtransf} with the minimized driver $\wt{G}^a$ given by \eqref{equation-minimized.risk.driver.for.entropic.agents-general.case}.

Following the computations from Section \ref{example-general-lambdas}, we see that the market clearing condition requires 
\begin{align*}
		0 = \sum_{a\in\bA} \pi^{*,a,2}_t 
				= \sum_{a\in\bA} \frac1{1+\wt \lambda^a}\frac{\wt Z^{a,2}_t + \gamma_a \theta^R_t}{ \kappa^{R}_t}
		\quad \Leftrightarrow \quad 
		\theta^R_t 
				= -\frac{\sum_{a\in \bA}\frac{\wt Z^{a,2}_t}{(1+\wt\lambda^a)}}{\sum_{a\in \bA}\frac{\gamma_a}{(1+\wt\lambda^a)}} 
				= -\frac{\sum_{a\in \bA} w^a \wt Z^{a,2}_t}{\gamma_R} ,
\end{align*}
if we define $\gamma_R = \sum_{a\in \bA} w^a \gamma_a$, with $w^a = 1/(\Lambda (1+\wt{\lambda}^a))$ and $\Lambda = \sum_{a\in \bA} {1}/({1+\wt{\lambda}^a})$. Notice that here we do not need to normalize the family $w=(w^a)$ so that $\sum_{a\in \bA} w^a = 1$, since we are not considering an aggregated risk measure. Any rescaling $\Lambda'$ of $w$ would give the same $\theta^R$. We present it in this way for consistency with the general case.

Now, replacing the term $\theta^R$ by the above value in the minimized driver given by \eqref{equation-minimized.risk.driver.for.entropic.agents-general.case}, we find that the optimal risk processes for each agent solve the BSDEs with driver given by
	\begin{align}		\label{eq:changeddriverG-aggregationG}
		\wt{G}^a(t,\wt{Z}^\bA_t) 
			= -\frac{\gamma_a}{2} \big(\theta^S_t\big)^2 
			  - \wt{Z}^{a,1}_t\theta^S_t 
				+ \frac{1}{\gamma_R} \; \wt{Z}^{a,2}_t \; 
				          \bigg(\sum_{b\in \bA} w^b \wt Z^{b,2}_t\bigg) 
			  - \frac{\gamma_a}{2\gamma_R^2} 
				          \bigg(\sum_{b\in \bA} w^b \wt Z^{b,2}_t\bigg)^2.
	\end{align}
The BSDEs with these drivers form a system of $N$ coupled BSDEs with quadratic growth, which, in general, are difficult to solve, see \cite{EspinosaTouzi2015}, \cite{Espinosa2010}, \cite{FreiDosReis2011} or more recently \cites{Frei2014,KramkovPulido2016}. 
Fortunately, one can take advantage of the structure of \eqref{eq:changeddriverG-aggregationG} and find a simpler BSDE for the process $(\wh{Y}^w,\wh{Z}^w) = \sum_{a\in \bA} w^a (\wt{Y}^a,\wt{Z}^a)$. It is easily seen that $\wh{Y}^w_T=-\sum_{a\in \bA} w^a(H^a+nH^D/N)=-H^w$, as in \eqref{eq:rep-ag-term-cond}.
Linearly combining the BSDEs \eqref{agentresidualBSDEtransf} with drivers expressed as in \eqref{eq:changeddriverG-aggregationG}, we find
	\begin{align}
	\label{eq-help-rep-ag}
		-\ud \wh{Y}^w_t 
		= 
		\Big[ -\frac{\gamma_R}{2} \big(\theta^S_t\big)^2 - \wh{Z}^{w,1}_t\theta^S_t + \frac{1}{2\gamma_R}\big(\wh{Z}^{w,2}_t\big)^2  \Big] \udt - \langle \wh{Z}^w_t, \udwt\rangle 
		\quad \text{with} \quad
		\wh{Y}^w_T 
		=
		- H^w.
	\end{align}
This is exactly the same BSDE as in Example \ref{example_EntropicOptimizedDriver}. Given that $H^w$ and $\theta^S$ are bounded, this BSDE falls in the standard class of quadratic growth BSDE and the existence and uniqueness of $(\wh{Y}^w,\wh{Z}^w)$ is easily guaranteed. This allows one to compute $\theta^R$ as $-\wh{Z}^{w,2}/\gamma_R$ and in turn one can finally solve the BSDEs giving the minimized risk processes for each agents, using the driver $\wt{G}^a$ as given by \eqref{equation-minimized.risk.driver.for.entropic.agents-general.case}.

\begin{remark}[No trade-off between risk tolerance and performance concern rate] 
Each agent's individual preferences are specified by the parameters $\gamma_a$ and $\lambda^a$, i.e.  risk tolerance and performance concern respectively. One may ask whether a parametric relation between those parameters exists such that an agent with $(\gamma_a,\lambda^a)$ and another agent with $(\gamma_{b}, \lambda^{b})$ would exhibit the same behaviour and have the same optimal strategies. Indeed, in most formulas the two parameters appear as coupled. However, one can see that the terminal condition $H^w$ is independent of the risk tolerance parameter $\gamma^{\cdot}$, hence by changing $\lambda^a$ and $\gamma^a$ of any one fixed agent $a \in \bA$, one cannot obtain the same outcome.
\end{remark}

%
%
%

%
%
%
\section{Further results on the entropic risk measure case}
\label{section-EntropicCase}
In this section we investigate further the entropic case. We introduce a structure that allows to use the theory developed in the previous section and, moreover, to design $H^D$ such that Assumption \ref{assump-kappa} holds true. The ultimate goal of this section is to understand how the concern rates $\lambda$ affect prices and risks. The first two parts of the section verify that Assumption \ref{assump-kappa} holds and the third sheds light on the behavior of the aggregated risk and derivative price as the parameters vary. 

We now make further assumptions (commented below) on the structure of the random variables introduced Section \ref{sec:model}. Namely, we assume that the endowments $H^a$ for $a\in \bA$ and the derivative $H^D$ have the form 
\begin{align}
\label{agents-payoff-Ha-bond-payoff-HD}
    H^a 
		= h^a (S_T,R_T) 
		\quad \text{and}\quad 
    H^D 
		= h^D (S_T,R_T) 
\end{align}
for some deterministic functions $h^\cdot$. 
This structure for the derivative and endowments is interpreted as each agent receiving a lump sum at maturity time $T$.

To ease the analysis we will assume throughout a Black-Scholes market (i.e $\mu^S,\sigma^S$ are constants). Such an  assumption is not strictly necessary for the results we obtain here, but we wish to focus on the qualitative analysis and not on obfuscating mathematical techniques. Throughout the rest of this section the next assumption holds.
\begin{assumption}
\label{assump-DiffData}
Let Assumption \ref{assump-StandardData} hold. 
Let $\sigma^S\in(0,\infty)$ and $\mu^R,\mu^S\in \bR$ (and hence also $\theta^S \in \bR)$.  For any $a\in\bA$ the functions  $h^D,h^a\in C^1_b(\bR^2;\bR)$ are strictly positive, their derivatives are uniformly Lipschitz continuous w.r.t.~the non-financial risk and 
satisfy $(\partial_{x_2} h^D)(x_1,x_2)\neq 0$ for any $(x_1,x_2)\in \bR\times \bR$.
\end{assumption}
The assumption concerning the strict positivity of the involved maps or that $\partial_{x_2} h^D\neq 0$ are the key in proving that Assumption \ref{assump-kappa} is indeed verified for the example we present. 
The assumption on the form of $H^a$ and $H^D$ reduces the BSDE to the Markovian case, giving us access to the many existing BSDE regularity results, which we will use below in their full scope. 
It would be possible (this is left open to future research) to remain in the non-Markovian setting of general $\cF_T$-measurable $H^D$ and $H^a$ and use link between non-Markovian BSDE and path-dependent PDEs (see e.g. \cite{EkrenKellerTouziEtAl2014}). Indeed, tools on general Malliavin differentiability of BSDE solutions in the non-Markovian setting can be found in \cite{AnkirchnerImkellerReis2010} or in more generality in \cites{DosReis2011,MastroliaPossamaiReveillac2014}. 

We recall that our goal, in the example below, is to analyze the impact of the parameters $\lambda^\cdot,n,\gamma_\cdot$ on the risk processes (single and representative agent), derivative price process and EMPeR. 

\begin{remark}[On notation for the section]
In this section we work mainly with the representative agent BSDE (see Example \ref{example_EntropicOptimizedDriver} or \eqref{eq-help-rep-ag}) and the derivative price BSDE \eqref{BSDEforDerivativeHD}.

To avoid a notation overload in what the BSDE for the representative agent is concerned, we drop the tilde notation and define $(Y^w,Z^w)$ as the solution to the mentioned BSDE; not to be confused with \eqref{equation-risk.BSDE.for.the.representative.agent} which plays no role here. The solution to the derivative price BSDE is denoted by $(B,\kappa)$.
\end{remark}

\subsection{The aggregated risk}

The BSDE \eqref{eq-help-rep-ag} is not difficult to analyze given the existing literature on BSDEs of quadratic growth. Recall that $\theta^S\in \cS^\infty$ and $Y^w_T\in L^\infty$ (since it is a weighted sum of bounded random variables). We shortly recall that $\bD^{1,2}$ is the space of 1st order Malliavin differentiable processes and $D$ denotes the Malliavin derivative operator, we point the reader to Appendix \ref{sec-MalliavinCalculus} for further Malliavin calculus references.
\begin{theorem} 
\label{theo-thetaRisbounded}
The BSDE \eqref{eq-help-rep-ag} has a unique solution $(Y^w,Z^w)\in (\cS^\infty \cap \bD^{1,2}) \times (\cHBMO\cap \bD^{1,2})$. 
Moreover, there exists a strictly negative function $u^w\in C^{0,1}([0,T]\times \bR^2, \bR)$ such that for any $t\in[0,T]$
	\begin{align*} 
	Y^w_t=u^w(t,S_t,R_t) 
	\qquad \text{and}\qquad 
	Z^{w,2}_t = (\partial_{x_2} u^w)(t,S_t,R_t)b,\qquad \bP\text{-a.s.}.
	\end{align*}
\begin{enumerate}[i)]
	\item For any $r,u\in [0,t]$, $t\in[0,T]$ it holds that 
		\begin{align*}
		D^{W^R}_u Y^w_t=D^{W^R}_r Y^w_t\quad \bP\text{-a.s.}
		\qquad \text{and}\qquad
		D^{W^R}_u Z^w_t=D^{W^R}_r Z^w_t \quad\bP \otimes Leb\text{-a.e.}
		\end{align*}
and in particular $D^{W^R}_t Y_t=Z^w_t$ $\bP$-a.s. for any $t\in[0,T]$.
	\item  There exists a constant $C> 0$ such that $|Z^{w,2}_t|\leq C$ for any $t\in[0,T]$, i.e.  $Z^{w,2}\in \cS^\infty$ and $\partial_{x_2} u^w\in C_b$. Moreover, $\theta^R \in \cS^\infty$. 
	\item The process $D^{W^R}_\cdot Z^w$ belongs to $\cHBMO$.
\end{enumerate}
\end{theorem}

\begin{proof}
Let $a\in \bA$ and $0\leq u\leq t\leq T$. Existence and uniqueness of the SDEs 
\eqref{diffusion-R} and \eqref{diffusion-S} follow from Proposition \ref{prop-SDEresults}.

By assumption we have $Y^w_T\in L^\infty$ and $\theta^S\in \cS^\infty$ which allows to quote Theorem 2.6 in \cite{ImkellerDosReis2010} and hence that $(Y^w,Z^w)\in \cS^\infty \times \cHBMO$. Moreover, given that $Y^w_T<0$, a strict comparison principle for quadratic BSDEs (see e.g. \cite{MaYao2010}*{Property (5)}) yields easily that $Y_t^w<0$ for any $t\in[0,T]$ and hence that $u^w<0$. 

Proposition \ref{prop-SDEresults} ensures that the payoffs $H^D$ and $H^a$, and hence $H^w$, are  Malliavin differentiable with bounded Malliavin derivatives. Combining this further with $\theta^S\in \bR$, the Malliavin differentiability of \eqref{eq-help-rep-ag} follows from  Theorem 2.9 in \cite{ImkellerDosReis2010}. Under Assumption \ref{assump-DiffData} the results in \cite{ImkellerDosReis2010} (or Chapter 4 of \cite{DosReis2011}) along with Theorem 7.6 in \cite{AnkirchnerImkellerReis2010} yield the Markov property for $Y^w$ and the  parametric differentiability result for the (quadratic) BSDE. 

\emph{$\triangleright$ Proof of i):} 
Since $u^w\in C^{0,1}$ by direct application of the Malliavin differential we have for $0\leq u\leq t\leq T$
\begin{align*}
D^{W^R}_u Y^w_t
& 
=
D^{W^R}_u \big(u^w(t,S_t,R_t)\big)
\\
&
=
(\partial_{x_2} u^w)(t,S_t,R_t)(D^{W^R}_u R_t)
=
(\partial_{x_2} u^w)(t,S_t,R_t) b 
=
D^{W^R}_t Y^w_t.
\end{align*}
It now follows that $D^{W^R}_t Y^w_t=D^{W^R}_u Y^w_t= Z^w_t$ for any $0 \leq u\leq t\leq T$ $\bP$-a.s.. 

\emph{$\triangleright$ Proof of ii):}
Define now the probability measure $\bQ$ (equivalent to $\bP$) as 
\begin{align}
\label{eq:DifferentiabilityMeasureChange}
\frac{\ud \bQ}{\ud \bP}= \cE\left( 
-\int_0^T \left\langle  (\theta^S_s,-\frac{Z^{w,2}_s }{\gamma_R}) ,
 \udws \right \rangle
\right).
\end{align} 
The measure $\bQ$ is well defined since $\theta^S\in \cS^\infty$ and
 $Z^{w,2}\in \cHBMO$. Then for $0\leq u\leq t\leq T$ we have (Theorem 2.9 in \cite{ImkellerDosReis2010})
\begin{align}
\label{eq:BSDEforDWRYw}
D^{W^R}_u Y^w_t
&= D^{W^R}_u Y^w_T + 
\int_t^T 
[-\theta^S_s D^{W^R}_u Z^{w,1}_s 
+\frac1{\gamma_R} Z^{w,2}_s D^{W^R}_u Z^{w,2}_s] \uds
-\int_t^T \langle D^{W^R}_u Z^{w}_s, \udws\rangle 
\\
\nonumber
&
\qquad \Rightarrow 
D^{W^R}_u Y^w_t = \bE^{\bQ}[D^{W^R}_u Y^w_T| \cF_t].
\end{align}
The results in Proposition \ref{prop-SDEresults} and the definition of $Y^w_T$ imply that $|D^{W^R}_u Y^w_t|<C$. Path regularity results for BSDEs along with their usual representation formulas (see \cite{ImkellerDosReis2010}) yield that $(D^{W^R}_t Y_t)= (Z^2_t)\in \cS^\infty$; the boundedness of $\partial_{x_2} u^w$ follows in an obvious way. As a consequence, $\theta^R\in\cS^\infty$ since $Z^{w,2}\in\cS^\infty$ and \eqref{eq:entropic-EMPeR} holds. 

\emph{$\triangleright$ Proof of iii):}
Using now the fact that $\theta^S,Z^{w,2}\in \cS^\infty$, we apply Theorem 2.6 in \cite{ImkellerDosReis2010} to \eqref{eq:BSDEforDWRYw} and obtain that  $D^{W^R}_\cdot Z^w \in\cHBMO$. The BMO norm of $D^{W^R} Z^w$ depends only on some real constants and $T$, $\gamma_R$, $\sup_u \|D^{W^R}_u Y^w_T\|_{L^\infty}$ and $\|(\theta^S,Z^{w,2})\|_{\cS^\infty\times \cS^\infty}$ (see again Theorem 2.6 in \cite{ImkellerDosReis2010}).
\end{proof}
In the next result we show that the mapping $x_2\mapsto (\partial_{x_2}u^w)(t,x_1,x_2)$ is Lipschitz. Denote by $R$ and $\wt R$ the solutions to \eqref{diffusion-R} with $R_0=r_0$ and $R_0=\wt r_0$ respectively; denote as well by $(Y,Z)$ and $(\wt Y,\wt Z)$ the solutions to BSDE \eqref{eq-help-rep-ag} for the underlying processes $R$ and $\wt R$ respectively. 
\begin{proposition}\label{propo:DDZw2isbounded}
For any $(t,x_1)\in[0,T]\times \bR$ the map $\bR\ni x_2\mapsto (\partial_{x_2} u^w)(t,x_1,x_2)$ is Lipschitz continuous uniformly in $t$ and $x_1$. In particular the process $D^{W^R} Z^w$ is $\bP$-a.s. bounded.
\end{proposition}
\begin{proof}
Let $0\leq u\leq t\leq T$ and define $\delta DY:=D^{W^R} Y^w-D^{W^R} \wt Y^w$, $\delta DZ^i:=D^{W^R} Z^{w,i}-D^{W^R} \wt Z^{w,i}$ for $i\in\{1,2\}$ and (intuitively) $\delta DZ:=(\delta DZ^1,\delta DZ^2)$. Then, following from \eqref{eq:BSDEforDWRYw} written under $\bQ$ from \eqref{eq:DifferentiabilityMeasureChange}, we have 
\begin{align*}  
\delta D_u Y_t
&= \delta D_u Y_T-\int_t^T \langle \delta D_uZ_s ,\udws^\bQ\rangle +
\int_t^T 
\frac1{\gamma_R} (Z^{w,2}_s -\wt Z^{w,2}_s )D^{W^R}_u Z^{w,2}_s
 \uds.
\end{align*}
Define now the process
\begin{align}
\label{eq:eprocessbasedonZw2} 
	e_t:=\exp\left\{\int_0^t \frac1{\gamma_R} D^{W^R}_u Z^{w,2}_s \uds \right\},
	\quad t\in[0,T] 
	\qquad \text{with}\qquad
	(e_t)\in \cH^p,\quad \forall p>1,
\end{align} 
where the $\cH^p$ integrability of $(e_t)$ follows from Lemma \ref{lemma-BMOproperties}. Observe next that by the results of Theorem \ref{theo-thetaRisbounded} one has $\delta D_u Y_t=\delta D_t Y_t= Z^{w,2}_t -\wt Z^{w,2}_t$. Applying It\^o's formula to $(e_t \delta D_\cdot Y_t)$, using the just mentioned identity and taking $\bQ$-conditional expectations it follows at $u=t=0$ that
\begin{align*}  
|(\partial_{x_2} u^w)(0,s_0,r_0)-(\partial_{x_2} u^w)(0,s_0,\wt r_0)|
=
\frac1b|(Z^{w,2}_0 -\wt Z^{w,2}_0)|
=
|\frac1b \bE^\bQ\left[ e_T \delta D_0 Y_T\right]|
\leq C|r_0-\wt r_0|.
\end{align*}
The last line is a consequence of Proposition \ref{prop-SDEresults} combined with the fact that $\bE^\bQ[ e_T^p ]$ ($\forall p> 1$) is finite due to the BMO properties of $D^{W^R} Z^{w,2}$, see Lemma \ref{lemma-BMOproperties}. The constant $C$ is independent of $u,r_0,\wt r_0$ and $s_0$. Although $D^{W^R} Z^{w,2}$ is a BMO martingale under $\bP$, the integrability still carries under $\bQ$; this is the same argument as in the final step of the proof of Lemma 3.1 in \cite{ImkellerDosReis2010} (see also Lemma 2.2 and Remark 2.7 of the cited work).

The extension of the above result to the whole time interval $[0,T]$ follows via the Markov property of the BSDE solution. This relates to the close link between BSDEs of the Markovian type and certain classes of quasi-linear parabolic PDEs (see e.g. Section 4 in \cite{ElKarouiPengQuenez1997}).

Finally, the boundedness of $D^{W^R} Z^w$ follows from the Lipschitz property of $x_2\mapsto (\partial_{x_2} u^w)(\cdot,\cdot,x_2)$ and the boundedness of $D^{W^R} R$, see Proposition \ref{prop-SDEresults}, ii).
\end{proof}

\subsection{The EMPR and the derivative's BSDE}
\noindent
We next show that Assumption \ref{assump-DiffData} implies Assumption \ref{assump-kappa} holds for the model with entropic risk.
\begin{theorem}[Market completion]	
	\label{theo:kappaRisEMPR}
The derivative $H^D$ completes the market, i.e. $\kappa^R\neq0$ $\bP$-a.s. for any $t\in[0,T]$. Moreover, $\kappa^R\in \cS^\infty$ and $\sgn(\kappa^R_t)=\sgn(b\partial_{x_2}h^D)$ for any $t\in[0,T]$.
\end{theorem}
Before proving the above result we need an intermediary one. Recall that BSDE \eqref{BSDEforDerivativeHD} describes the dynamics of the price process $B^\theta$, that $H^D\in L^\infty$ and $\theta \in \cS^\infty\times (\cHBMO\cap \bD^{1,2})$ (following from Assumption \ref{assump-DiffData} and Theorem \ref{theo-thetaRisbounded}).
\begin{proposition}	\label{MallDiffofB} 
The pair $(B,\kappa)$ belongs to $(\cS^\infty \cap \bD^{1,2})\times (\cHBMO \cap \bD^{1,2})$ and their Malliavin derivatives satisfy for $0\leq u\leq t\leq T$ the dynamics
\begin{align} 
\label{eq:MallDiffBkappa} 
 D^{W^R}_u B_t^{\theta}
	 & 
	= D^{W^R}_u H^D  
		- \int_t^T 
		     \kappa^R_s D^{W^R}_u \theta^R_s   
		    +\lev \theta_s ,D^{W^R}_u\kappa^\theta_s \rev 
				 \uds
	- \int_t^T \langle D^{W^R}_u \kappa_s^\theta, \ud W_s\rangle .
	\end{align}
The representation $D^{W^R}_t B^\theta_t= \kappa^R_t$ holds $\bP$-a.s. for any $0\leq t\leq T$.
\end{proposition}
\begin{proof}
Let $0\leq u\leq t\leq T$. Observe that BSDE \eqref{BSDEforDerivativeHD} is a BSDE with a linear driver and a bounded terminal condition. 
The existence and uniqueness of a solution follows from the results of \cite{ElKarouiPengQuenez1997}. Moreover, the estimation techniques used in \cite{ImkellerDosReis2010} yield that $(B,\kappa)\in \cS^\infty\times \cHBMO$ (see Theorem 2.6 in \cite{ImkellerDosReis2010}). The Malliavin differentiability of $(B,\kappa)$ follows from Proposition 5.3 in \cite{ElKarouiPengQuenez1997} and the remark following it since $(\theta^S,\theta^R)\in \bR\times (\cS^\infty\cap \bD^{1,2})$ (see Theorem \ref{theo-thetaRisbounded}). The quoted result and Proposition \ref{prop-SDEresults} yield \eqref{eq:MallDiffBkappa} for $D^{W^R}B^\theta$. Moreover, from  Theorem 2.9 in \cite{ImkellerDosReis2010} we have $\lim_{u \nearrow t} D^{W^R}_u B^\theta_t= \kappa^R_t$ for $0\leq u\leq t\leq T$ $ \bP \otimes Leb$-a.e.. 

We now prove a finer result on $B$ and $\kappa$, namely that $D^{W^R}_t B^\theta_t= \kappa^R_t$ holds $\bP$-a.s. for any $0\leq t\leq T$ instead of just $ \bP \otimes Leb$-a.e.. This is done by showing that $(u,t)\mapsto D^{W^R}_u B^\theta_t$ is jointly continuous.

Remark that the map $t\mapsto D^{W^R}_u B^\theta_t$ for $u\leq t$ is given by \eqref{eq:MallDiffBkappa} and hence it is continuous in time ($\forall t\in[u,T]$). Note now that  Proposition \ref{propo:DDZw2isbounded} and Proposition \ref{prop-SDEresults} yield  that $D^{W^R} Z^{w,2}$ is bounded and $D^{W^R}_u Z^{w,2}_t=D^{W^R}_r Z^{w,2}_t=D^{W^R}_0 Z^{w,2}_t$ for any $0\leq u,r\leq t\leq T$. These properties hold as well for $\theta^R$ via the identity $-\gamma_R \theta^R =Z^{w,2}$.

Using the measure $\bP^\theta$ (introduced in \eqref{density-process}), that $D^{W^R} \theta^S=0$ and the identity $-\gamma_R \theta^R =Z^{w,2}$, one can rewrite \eqref{eq:MallDiffBkappa} as
\begin{align}
\label{eq:DWRB-cheaptrick} 
 D^{W^R}_u B_t^{\theta}
	&
		= D^{W^R}_u H^D  
		+ \frac1{\gamma_R}\int_t^T 
		     \kappa^R_s D^{W^R}_u Z^{w,2}_s   
						 \uds
	- \int_t^T \langle D^{W^R}_u \kappa_s^\theta, \ud W^{\theta}_s\rangle .
\end{align}
Writing the same BSDE as above, but for a parameter $v$ (instead of $u$) we have 
\begin{align*}
 D^{W^R}_v B_t^{\theta}
	&
		= D^{W^R}_v H^D  
		+ \frac1{\gamma_R}\int_t^T 
		     \kappa^R_s D^{W^R}_v Z^{w,2}_s   
						 \uds
	- \int_t^T \langle D^{W^R}_v \kappa_s^\theta, \ud W^{\theta}_s\rangle 
\\
	&
		= D^{W^R}_u H^D  
		+ \frac1{\gamma_R}\int_t^T 
		     \kappa^R_s D^{W^R}_u Z^{w,2}_s   
						 \uds
	- \int_t^T \langle D^{W^R}_v \kappa_s^\theta, \ud W^{\theta}_s\rangle ,
\end{align*}
where we used the results of Proposition \ref{prop-SDEresults}. Since the solution to \eqref{eq:DWRB-cheaptrick} is unique and the BSDE just above has exactly the same parameters as \eqref{eq:DWRB-cheaptrick}, we must conclude that for any $t\in[0,T]$ and for $0\leq u,r\leq t$ it holds $D^{W^R}_u B_t^{\theta} = D^{W^R}_r B_t^{\theta}$. 
From the continuity of $t\mapsto D^{W^R}_\cdot B^{\theta}_t$ follows now the joint continuity of $(u,t)\mapsto D^{W^R}_u B^{\theta}_t$ in its time parameters and hence the representation $D^{W^R}_t B^\theta_t= \kappa^R_t$  holds $\bP$-a.s. for any $0\leq t\leq T$.
\end{proof}

We can now prove Theorem \ref{theo:kappaRisEMPR}.
\begin{proof}[Proof of Theorem \ref{theo:kappaRisEMPR}]
We proceed in the same way as in the proof of Proposition \ref{propo:DDZw2isbounded}. The argument goes as follows: define the process $(e_t)$ just like in \eqref{eq:eprocessbasedonZw2}; apply It\^o's formula to $(e_t D^{W^R}_\cdot B^\theta_t)$ and write the resulting equation under $\bP^\theta$ (just like \eqref{eq:DWRB-cheaptrick}); take $\bP^\theta$ conditional expectations. At this point a remaining Lebesgue integral is still in the dynamics:
\begin{align*}
 D^{W^R}_u B_t^{\theta}
	&
		= 
		(e_t)^{-1}
		\bE^\theta\left[
		e_T D^{W^R}_u H^D  
		+ \frac1{\gamma_R}\int_t^T 
		e_s(\kappa^R_s-D^{W^R}_u B_s^{\theta}) D^{W^R}_u Z^{w,2}_s  \uds
		|\cF_t
\right]
\\
&		= 
		(e_t)^{-1}
		\bE^\theta\left[
		e_T D^{W^R}_u H^D  |\cF_t
\right],
\end{align*}
where from the first to the second line we used Proposition \ref{MallDiffofB}, i.e. that $\kappa^R_s=D^{W^R}_s B^\theta_s=D^{W^R}_u B^\theta_s$  $\bP$-a.s. for any $0\leq u\leq s\leq T$.

Recalling $H^D=h^D(S_T,R_T)$ and the dynamics of $R$ given by \eqref{diffusion-R}, we see that (by the chain rule) $D_u^{W^R}H^D = b\partial_{x_2}h^D$.
Since $b\partial_{x_2} h^D$ is either always positive or always negative and since $\kappa^R_t =D^{W^R}_t B^\theta_t$ $\bP$-a.s.~for any $t \in [0,T]$, 
it follows that $\kappa^R\neq 0$ $\bP$-a.s.~for all $t\in[0,T]$. More precisely, depending on the sign of $b\partial_{x_2}h^D$, $\kappa^R$ is $\bP$-a.s. either always positive or always negative\footnote{For any positive random variable $X$ ($X>0$ $\bP$-a.s.) one has $\bE^{\bP}[X|\cF]> 0$ for any sigma-field $\cF$. Since the measure change is done for a strictly positive density function, the inequality for the new conditional expectation is still strict.}, giving $\sgn(\kappa^R_\cdot)=\sgn(b\partial_{x_2}h^D)$. 
\end{proof}

\subsection{Parameter Analysis}

It is possible to justify at a theoretical level some of the predictable behavior of the processes $Y^w$, $B^\theta$ and $\theta^R$ with relation to  the problem's parameters: $n$, $\gamma_R$, $\lambda^a$ and $\gamma_a$ for $a\in\bA$.

\begin{theorem}		
\label{theo:parambehavior}
Let $\theta$ be the EMPR. The process $(Y^w,Z^w)$ solving BSDE \eqref{eq-help-rep-ag} is differentiable with relation to $\lambda^a$ for any $a\in\bA$, $n$ and  $\gamma_R$ (see \eqref{eq:defi-gammaR} and \eqref{eq:defi-weightsw}). 

Fix agent $a\in \bA$. If the differences
\begin{align}
\label{eq:signalcondition} 
\gamma_R-\gamma_a
\qquad \text{and}\qquad 
\bE^{\theta}\Big[\Big(\sum_{b\in\bA} w^b H^b\Big) - H^a\Big] 
\end{align}
are positive (negative respectively) then $\partial_{\wt \lambda^a} Y^w_t$ is negative (positive respectively) for any $t\in[0,T]$.

For any $a\in \bA$ we have $\bP\text{-a.s}$ that
\begin{align*}
\partial_{\gamma_R} Y^w_t & <0,
\quad 
\partial_{\gamma_a} Y^w_t <0
\qquad
\forall t\in [0,T).
\end{align*}
Furthermore, $\bP\text{-a.s}$
\begin{align*}
\partial_{n} Y^w_t & < 0,
\quad 
\sgn(\partial_n \theta^R_t)=\sgn(b\partial_{x_2}h^D)
\ 
\forall t\in [0,T]
\quad\text{and}\quad 
\partial_{n} B^\theta_t  < 0
\quad
\forall t\in [0,T).
\end{align*}
\end{theorem}
Part of the results are in some way expected. Introducing more derivatives leads to an overall risk reduction and as more derivatives are placed in the market, the derivative is worth less (per unit). 
If $\gamma_R$ is interpreted as the representative agent's risk tolerance, then as $\gamma_R$ increases we have a decrease in risk ($Y^w$ decreases) since it represents an increase in the single  agents' risk tolerance (i.e.  $\gamma_a\nearrow$).  

The main message of the above theorem is that the effect of the performance concern of one agent on the aggregate risk depends essentially on how the agent is positioned with respect to the others, both in terms of risk tolerance as well as the personal endowments. If the agent's risk tolerance $\gamma_a$ is higher than the aggregate risk tolerance $\gamma_R$ and her endowment position dominates by the aggregate endowment position, then an increase in the agent's concern rate leads to an increase of the aggregate risk.

Before proving the above result we remark that condition \eqref{eq:signalcondition} simplifies under certain conditions; such simplifications are summarized in the below corollary. All results follow by direct manipulation of the involved quantities.

\begin{corollary}
\label{coro-parameter-behavior}
Let the conditions of Theorem \ref{theo:parambehavior} hold. 

If $\gamma_a=\gamma$ for all $a\in\bA$, then $\gamma_R-\gamma_a= \gamma\big(\sum_a w^a  - 1\big)=0$. 

If $N=2$, then $w^a+w^b=1\Leftrightarrow w^b=1-w^a$ and hence
\begin{align*}
\big(\sum_{c\in\bA} w^c H^c\big) - H^a  =  -w^b(H^a-H^b)
\qquad\text{and}\qquad
\big(\sum_{c\in\bA} w^c H^c\big) - H^b  =   w^a(H^a-H^b).
\end{align*}
Similarly $\gamma_R-\gamma_a=-w^b(\gamma_a-\gamma_b)$ and $\gamma_R-\gamma_b=w^a(\gamma_a-\gamma_b)$. Moreover, it holds that
\begin{align}
\label{eq:crossbehaviorYw}
\text{sgn}\Big(\partial_{\lambda^a} Y^w_t\Big) 
     = - \text{sgn}\Big(\partial_{\lambda^b} Y^w_t\Big)
		\qquad \bP\text{-a.s. for any } t\in[0,T].
\end{align}
\end{corollary}

\begin{proof}[Proof of Theorem \ref{theo:parambehavior}]
Let $a\in \bA$ and $t\in[0,T]$. Theorem 3.1.9 in \cite{DosReis2011} (see also Theorem 2.8 in \cite{ImkellerDosReis2010}) ensures the differentiability of BSDE \eqref{eq-help-rep-ag}
with respect to $\gamma_R$, $\gamma_a$, $\lambda^a$ and $n$. 

\smallskip
\emph{$\triangleright$ The derivative of $Y^w$ in $\gamma_R$:} Applying $\partial_{\gamma_R}$ to  
BSDE \eqref{eq-help-rep-ag} and writing it under the 
probability measure $\bQ$ defined in 
\eqref{eq:DifferentiabilityMeasureChange}  yields the dynamics
\begin{align*}
\partial_{\gamma_R} Y^w_t
&
= 
0
+ 
\int_t^T 
\left[
-\frac{1}{2} (\theta^S_s)^2 
			-\frac{1}{2 \gamma_R^2} \left(Z^{w,2}_s\right)^2
			\right]
 \uds-\int_t^T \langle \partial_{\gamma_R} Z^{w}_s, \udws^\bQ\rangle .
\end{align*}
Taking $\bQ$-conditional expectations and noticing that the Lebesgue integral term is strictly negative for any $t\in[0,T)$, we have then $\partial_{\gamma_R} Y^w_t<0$ for any $t\in[0,T)$. 

\smallskip
\emph{$\triangleright$ The derivative of $Y^w$ in $\gamma_a$:} This case follows from the previous one as $\gamma_R$ is defined by \eqref{eq:defi-gammaR} and the weights $w^\cdot$ (see \eqref{eq:defi-weightsw}) are independent of $\gamma_\cdot$. 
\begin{align*}
\gamma_R := \sum_{a \in \bA} w^a {\gamma_a}
\quad \text{implies} \quad
\partial_{\gamma_a} \gamma_R = w^a > 0,
\end{align*}
and finally $\partial_{\gamma_a} Y^w = \partial_{\gamma_R} Y^w \cdot \partial_{\gamma_a} (\gamma_R)$. The statement follows. 

\smallskip
\emph{$\triangleright$ The derivatives of $Y^w$ in $\wt \lambda^a$:} We compute only the derivatives with respect to $\wt \lambda^a$ in order to present simplified calculations as $\wt \lambda^a:=\lambda^a/(N-1)$. Calculating the involved derivatives leads to
\begin{align*}  
&
\partial_{\wt \lambda^a} \frac1{1+\wt \lambda^a} = -(\Lambda w^a)^2,  
\quad
\partial_{\wt \lambda^a} \frac 1 \Lambda 
= 
(w^a)^2
,\quad
\partial_{\wt \lambda^a} w^a = ( w^a)^2\Lambda(w^a-1) 
,\quad
\partial_{\wt \lambda^a} w^b = ( w^a)^2 \Lambda w^b, 
\\
& 
\partial_{\wt \lambda^a} \gamma_R 
= \partial_{\wt \lambda^a} \sum_{b\in\bA} w^b \gamma_b = (w^a)^2 \Lambda (\gamma_R-\gamma_a) 
\quad \text{and} \quad 
\partial_{\wt \lambda^a} H^w  
= 
(w^a)^2 \Lambda \Big(\big(\sum_{b\in\bA} w^b H^b\big) - H^a\Big).
\end{align*}
Combining the above results with the BSDE for $\partial_{\wt \lambda^a} Y^w$ under the $\bQ$-measure (just as in the previous two steps) yields 
\begin{align*}
\partial_{\wt \lambda^a} Y^w_t
&
= 
-(w^a)^2 \Lambda
\bE^\bQ\Big[
\Big(\big(\sum_{b\in\bA} w^b H^b\big) - H^a\Big)
+
(\gamma_R-\gamma_a)
\int_t^T 
\big[
\frac{1}{2} (\theta^S_s)^2 
			+\frac{1}{2 \gamma_R^2} \left(Z^{w,2}_s\right)^2
			\big]\uds
\Big| \cF_t\Big].
\end{align*}
Since $\bQ$ is equivalent to $\bP$, the statement follows. 

\smallskip
\emph{$\triangleright$ The derivative of $Y^w$ in $n$:} Applying $\partial_{n}$ to  
BSDE \eqref{eq-help-rep-ag} and writing it under the 
probability measure $\bQ$ defined in \eqref{eq:DifferentiabilityMeasureChange} yields the dynamics
\begin{align*}
\partial_{n} Y^w_t
&
= 
\partial_{n} Y^w_T -\int_t^T \langle \partial_{n} Z^{w}_s, \udws^\bQ\rangle 
\quad \Rightarrow  \quad 
\partial_{n} Y^w_t 
= \bE^\bQ[\partial_{n} Y^w_T | \cF_t]
= - \frac{\bE^\bQ[ H^D | \cF_t]}{N} \sum_{a\in\bA} w^a  <  0,
\end{align*}
where the last sign follows from the definition of $Y^w_T$ and $H^D$.

\smallskip
\emph{$\triangleright$ The derivative of $\theta^R$ in $n$:} The analysis of $Z^{w,2}$ and hence of $\theta^R$ with respect to $n$ and $\gamma_R$ follows from the analysis of \eqref{eq:BSDEforDWRYw}. Given representation \eqref{eq:entropic-EMPeR}, applying $\partial_{n}$ to   BSDE \eqref{eq:BSDEforDWRYw} and writing it under the probability measure $\bQ$ defined in \eqref{eq:DifferentiabilityMeasureChange}  yields the dynamics
\begin{align*}
\partial_{n}D^{W^R}_u Y^w_t
&= \partial_{n} D^{W^R}_u Y^w_T 
-\int_t^T \langle \partial_{n} D^{W^R}_u Z^{w}_s, \udws^\bQ\rangle 
+\int_t^T 
\big[\frac1{\gamma_R}  D^{W^R}_u Z^{w,2}_s \partial_{n} Z^{w,2}_s\big] \uds
\\
& \Leftrightarrow \quad 
\partial_{n}Z^{w,2}_t
= \partial_{n} D^{W^R}_u Y^w_T 
-\int_t^T \langle \partial_{n} D^{W^R}_u Z^{w}_s, \udws^\bQ\rangle 
+\int_t^T 
\big[\frac1{\gamma_R}  D^{W^R}_u Z^{w,2}_s \partial_{n} Z^{w,2}_s\big] \uds
\\
& \Leftrightarrow \quad 
\partial_{n}Z^{w,2}_t
= (e_t)^{-1}\bE^\bQ\left[ e_T\partial_{n} D^{W^R}_u Y^w_T  | \cF_t \right],
\end{align*}
where $(e_t)$ is as in \eqref{eq:eprocessbasedonZw2} and the argumentation is similar to that back there. Notice now that with the terminal condition  $Y_T =-\sum_{a \in \bA} w^a(H^a+nH^D/N)$ we have
\begin{align*}
	\partial_{n} D^{W^R}_t Y^w_T 
	= -\frac1N \left(\sum_{a\in\bA} w^a\right) D^{W^R}_t H^D
	= -\frac1N \left(\sum_{a\in\bA} w^a\right) b (\partial_{x_2} h^D)(S_T,R_T).
\end{align*}
Given Assumption \ref{assump-DiffData}, we are able to conclude that $\sgn(Z^{w,2}_t)=-\sgn(b\partial_{x_2}h^D)$, and hence, from \eqref{eq:entropic-EMPeR} that $\sgn(\partial_n \theta^R_t)=\sgn(b\partial_{x_2}h^D)$.

\smallskip
\emph{$\triangleright$ The derivative of $B^\theta$ in $n$:} We use justifications similar to those used in Proposition \ref{MallDiffofB} and hence we do not give all the details. Recall  \eqref{BSDEforDerivativeHD}, apply the $\partial_n$-operator to the equation and do the usual change of measure (with $\bP^\theta$) to obtain 
\begin{align*}
 \partial_n B_t^{\theta}
	&
		= 0
		- \int_t^T 
		     \kappa^R_s \partial_n \theta^R_s   
						 \uds
	- \int_t^T \langle \partial_n \kappa_s^\theta, \ud W^{\theta}_s\rangle 
	\quad\Rightarrow\quad
	\partial_n B_t^{\theta} 
	=
	-\bE^\theta[\int_t^T \kappa^R_s \partial_n \theta^R_s \uds].
\end{align*} 
By the previous result we have $\sgn(\partial_n \theta^R_t)=\sgn(b\partial_{x_2}h^D)$ and from Theorem \ref{theo:kappaRisEMPR} we have  $\sgn(\kappa^R_t)=\sgn(b\partial_{x_2}h^D)$. It easily follows that $\partial_n B_t^{\theta}<0$.
\end{proof}
Unfortunately the conditions used above do not allow for similar results on the behavior of, say $\gamma_R \mapsto \theta^R$ or $(\gamma_R,n,\lambda)\mapsto \wt Y^a$. The conditions required for such results are too restrictive to be of any usefulness. Nonetheless, we will investigate them in Section \ref{section-NumSection} via numerical simulation.

%
%
%

%
%
%
\section{Study of a particular model with two agents} 
\label{section-NumSection} 

In this section, we investigate a model economy consisting of two agents using entropic risk measures and having opposed exposures to the external non-financial risk. We give particular attention to the impact of the relative performance concern rates on the equilibrium related processes. The model is simple enough to allow extended tractability, when compared with Sections \ref{section-SingleAgentSection}, \ref{sec4-repagent} and \ref{section-EntropicCase}, and nonetheless still sufficiently general as to produce a rich set of results and interpretations. In particular, we explicitly describe the structure of the equilibrium. Using numerical simulations, we are able to explore the dependence of individual quantities (such as the optimal portfolios $\pi^{*a}$ and minimized risks $Y^a_0$) on the various parameters, thus complementing the results in Theorem \ref{theo:parambehavior}.

\subsection{The particular model and numerical methodology}

We consider a stylized market consisting of two agents.
We argue that a larger set of $N$ agents with certain exposures to the external risk $R$ can be clustered in two groups: those profiting from the high values of $R$ and those profiting from the low values of $R$, and we can apply the weighted aggregation technique used in Section \ref{sec4-repagent} to each group. Our two agents can therefore be thought of as representative agents for each group. The external risk process is taken to be the temperature affecting the two agents, who also have access to a stock market.
\subsubsection*{Temperature and Stock models}

We study one period of one month ($T=1$) 
We study one period of $T=1$ month 
where the temperatures follow an SDE \eqref{diffusion-R} with constant coefficients:
\begin{align*}
	R_t&=r_0+\mu^R \, t+b \, W^R_t,   
\end{align*}
and for the stock we take a standard Black--Scholes model:
\begin{align*}
	\frac{dS_t}{S_t} = \mu^Sdt+\sigma^SdW^S_t , 
\end{align*}
where the coefficients are $r_0=18$, $\mu^R=2$ and $b=4$ for the temperature process, and $S_0=50$, $\mu^S=-0.2$ and $\sigma^S=0.25$ (so $\theta^S={\mu^S}/{\sigma^S}=-0.8$) for the stock price process.

\subsubsection*{Agents' parameters, endowments and the derivative}

Define $I(x) := \frac{1}{\pi}\arctan(x) + \frac12 \ \in [0,1]$. The agents' endowments, $H^a$ and $H^b$, are taken to be
\begin{align*}
		H^a 
		&
		= 5 + I\Big(2\big(R_T - 24\big)\Big) \cdot 15 ,
		\\
		H^b 
		&= 5 +I\Big(2\big(16 - R_T\big)\Big) \cdot \Big( 15 + 5\, I\big(S_T-40\big) \Big) .
\end{align*}
Agent $a$ profits from higher temperatures while agent $b$ profits from lower ones. The derivative has a payoff $H^D$ that does not depend on the stock $S$, and is given by
\begin{align*}
	H^D = I(R_T-20), 
\end{align*}
so that it allows to transfer purely the external risk.
All functions satisfy Assumptions \ref{assump-StandardData} and \ref{assump-DiffData}. 
Given the agents' opposite exposures to $R_T$ and the design of $H^D$, agent $a$ will act as a seller while agent $b$ will act as the buyer, thus establishing a viable market for the derivative. 

We assume throughout that the total supply of derivative is zero, $n=0$, i.e. every unit of derivative one agent owns is underwritten by the other. The risk tolerance coefficients of the agents are fixed at $\gamma_a=\gamma_b=1$ unless we are analyzing some behavior with respect to them.  Similarly, unless otherwise specified, the concern rates are fixed to be $\lambda^a=\lambda^b=0.25$ unless we are analyzing some behavior with respect to them.

\subsubsection*{The numerical procedure}

The simulation of the processes involves a time discretization and Monte Carlo simulations. We use directly the forward processes' explicit solutions; all BSDEs are solved numerically. Regarding their time discretization, we use a standard backward Euler scheme, see \cite{BouchardTouzi2004}, and we complement the time-discretization procedure with the control variate technique stated in Section 5.4.2 of \cite{LionnetReisSzpruch2015}. The approximation of the conditional expectations in the backward induction steps is done via projection over basis functions, see the Least-Squares Monte Carlo method used in \cite{GobetTurkedjiev2014}.

We follow Sections \ref{section-SingleAgentSection} and \ref{sec4-repagent}. First, we solve the representative agent's BSDE \eqref{eq-help-rep-ag}. This yields via \eqref{eq:entropic-EMPeR} the EMPeR process $\theta^R$. Once this is obtained, we solve the BSDE for the price $B^\theta$ of the derivative, Equation \eqref{BSDEforDerivativeHD}, obtaining $(\kappa^S,\kappa^R)$ in the process. 
Finally, we solve the BSDE \eqref{eq:SingleAgentBSDE} with driver \eqref{equation-minimized.risk.driver.for.entropic.agents-general.case} for each agent $a\in\bA$ and compute the optimal strategies $\pi^{*,a}=(\pi^{*,a,1},\pi^{*,a,2})$ via \eqref{eq-optimPi-version2-pia1} and \eqref{eq-optimPi-version2-pia2}.	We note that in the case of two agents, the system \eqref{eq-optimPi-version2-pia1} is easily inverted.

All plots below are computed using $200.000$ simulated paths along a uniform time-discretization grid of $20$ time-steps, except the plot of Figure \ref{fig:ASimulatedOutcome} which uses $30$ time-steps.

\subsection{Analysis of the behavior in the model} 

Figure \ref{fig:ASimulatedOutcome} shows a realization of the behavior of the agents over the trading period. One can see that the price of the derivative moves like the temperature, and in particular it is never constant (over a time-interval where the temperature has changed). This means that the derivative does indeed complete the market by providing the agents full exposure to $R$, or equivalently to $W^R$ -- Assumption \ref{assump-kappa} is satisfied. Agent $b$ is always long in the derivative and $a$ always short (the latter following from the former since her position is the opposite of that of $b$). The fact that both agents only go short in the stock is due to its decreasing trend ($\theta^S<0$) and the fact that the endowments depend little on $S$: it is in mainly an optimal investment in the stock that is observed. However agent $b$'s endowment is higher for lower stock prices, hence she does not go as short in the stock as agent $a$, to hedge this variability.
\begin{figure}[htb]
\hspace{-2.0cm}
	\centering 
	\includegraphics[scale=0.55]{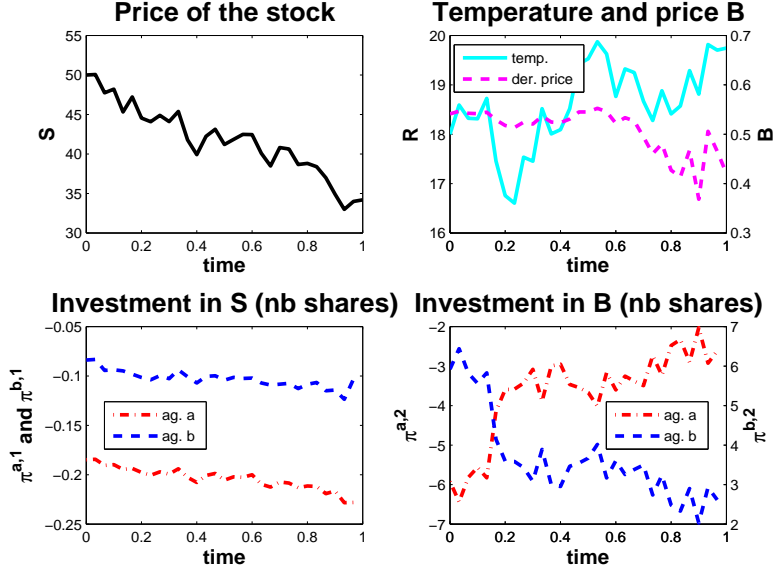}
	\vspace{-0.6cm}
	\caption{Sample paths of the several processes. Stock price on the top left; the temperature and the derivative price on the top right; the investment strategy in the stock on the bottom left and that in the derivative on the bottom right, for each agent. Here $\lambda^a=0.25$ and $\lambda^b=0.0$.}
		\label{fig:ASimulatedOutcome}
\end{figure}

\subsubsection*{Trading activity}
The optimal investment strategies for the derivative were seen in Section \ref{example-general-lambdas-with-zero-supply} and are given by
	\begin{align*}
		\pi^{*,a,2}
				= \frac{1}{1+\lambda^a} \ \frac{\wt{Z}^{a,2} + \gamma_a \theta^R}{\kappa^R}
			\quad\text{and}\quad
		\pi^{*,b,2}
				= \frac{1}{1+\lambda^b} \ \frac{\wt{Z}^{b,2} + \gamma_b \theta^R}{\kappa^R} .
	\end{align*}
The optimal investment strategies in the stock follow easily by inverting $A_2$ from \eqref{eq:ANmatrix}. This yields 
\begin{align*}
\begin{bmatrix}
\phantom{\bigg|}
 \pi^{*,a,1} 
\\
\phantom{\bigg|}
 \pi^{*,b,1} 
\end{bmatrix}
&
=
\begin{bmatrix}
\phantom{\bigg|}
\dfrac{1}{1-\lambda^a\lambda^b}          & \dfrac{\lambda^a}{1-\lambda^a\lambda^b} 
\\
\phantom{\bigg|}
\dfrac{\lambda^b}{1-\lambda^a\lambda^b}  & \dfrac{1}{1-\lambda^a\lambda^b}
\end{bmatrix}
\begin{bmatrix}
\phantom{\bigg|}
	\dfrac{\wt{Z}^{a,1} + \gamma_a \theta^S}{\sigma^S S} 
						 - \dfrac{\wt{Z}^{a,2} + \gamma_a \theta^R}{\kappa^R} \ \dfrac{\kappa^S}{\sigma^S S}
						\\
\phantom{\bigg|}
 \dfrac{\wt{Z}^{b,1} + \gamma_b \theta^S}{\sigma^S S} - \dfrac{\wt{Z}^{b,2} + \gamma_b \theta^R}{\kappa^R} \ \dfrac{\kappa^S}{\sigma^S S}
\end{bmatrix}
.
\end{align*}

\begin{remark}[On the structure of the equilibrium]
\label{remark-strut-of-equilibrium}
The structure of the optimal strategy in the stock's investment appears clearly in view of the examples treated in Section \ref{subsection-unconstrained.case-entropic.case}. Each agent computes her strategy as a weighted sum of the way both would compute their strategy as if there was no relative performance concern (compare with Section \ref{subsubsection:zeroconcern}), using the weights $\big( \frac{1}{1-\lambda^a\lambda^b}, \frac{\lambda^a}{1-\lambda^a\lambda^b} \big)$ for $a$ and $\big( \frac{\lambda^b}{1-\lambda^a\lambda^b}, \frac{1}{1-\lambda^a\lambda^b} \big)$ for $b$. 

These weights can be understood from Equation \eqref{eq-optimPi-version2-pia1} with $\bA=\{a,b\}$: each agent's best response is to invest in the stock according to her natural strategy plus $\lambda^i$ times the strategy played by the other. Assume now that each agent was initially planning to compute her optimal position using 
		\begin{align*}
			\pi^{(0),i,1} 
			= 
			\frac{\wt{Z}^{i,1} + \gamma_i \theta^S}{\sigma^S S} 
			- \frac{\wt{Z}^{i,2} + \gamma_i \theta^R}{\kappa^R} \ \frac{\kappa^S}{\sigma^S S},
			\qquad i \in \left\{a,b\right\},
		\end{align*}
and that they are shown, in turn, the strategy that the other is about to play, so that they can update theirs, yielding a sequence of strategies $\pi^{(1),a,1}$, $\pi^{(1),b,1}$, $\pi^{(2),a,1}$, $\pi^{(2),b,1}$, $\pi^{(3),a,1}$, $\ldots$ for each agent (starting with $a$'s update). Because they both update their strategy according to Equation \eqref{eq-optimPi-version2-pia1}, we observe agent $a$ imitating part of agent $b$, imitating part of agent $a$, imitating part of agent $b$, etc. Summing the corresponding series, agent $a$ ends up investing according to $\sum_n (\lambda^a\lambda^b)^n \ \pi^{(0),a,1} + \lambda^a \sum_n (\lambda^a\lambda^b)^n \pi^{(0),b,1}$, and similarly for $b$.

The structure of the optimal investment in the derivative is much different, following fundamentally from the endogenous trading condition. If an agent is shown the strategy that the other had decided to follow, she could not unilaterally change her strategy. From this emerges the EMPeR $\theta^R$ -- see below.
\end{remark}

We now look at the behavior of the individual portfolios with respect to the rates of relative performance concern. The intensity of the trading activity at time $t=0$ on both the stock ($\pi^{*,a,1}_0$) and the derivative ($\pi^{*,a,2}_0$) as maps of the concern rates $\lambda^a,\lambda^b$ can be found in Figure \ref{fig:TradingActivityInStock}. 
The positions of agent $b$ are similar in some sense: for the stock, the surface looks very similar; for the derivative, it is the exact opposite (due to the zero net supply condition). For readability we plot only the position of agent $a$.
\begin{figure}[ht]
	\centering
	\subfigure
	{
		\includegraphics[scale=0.34]{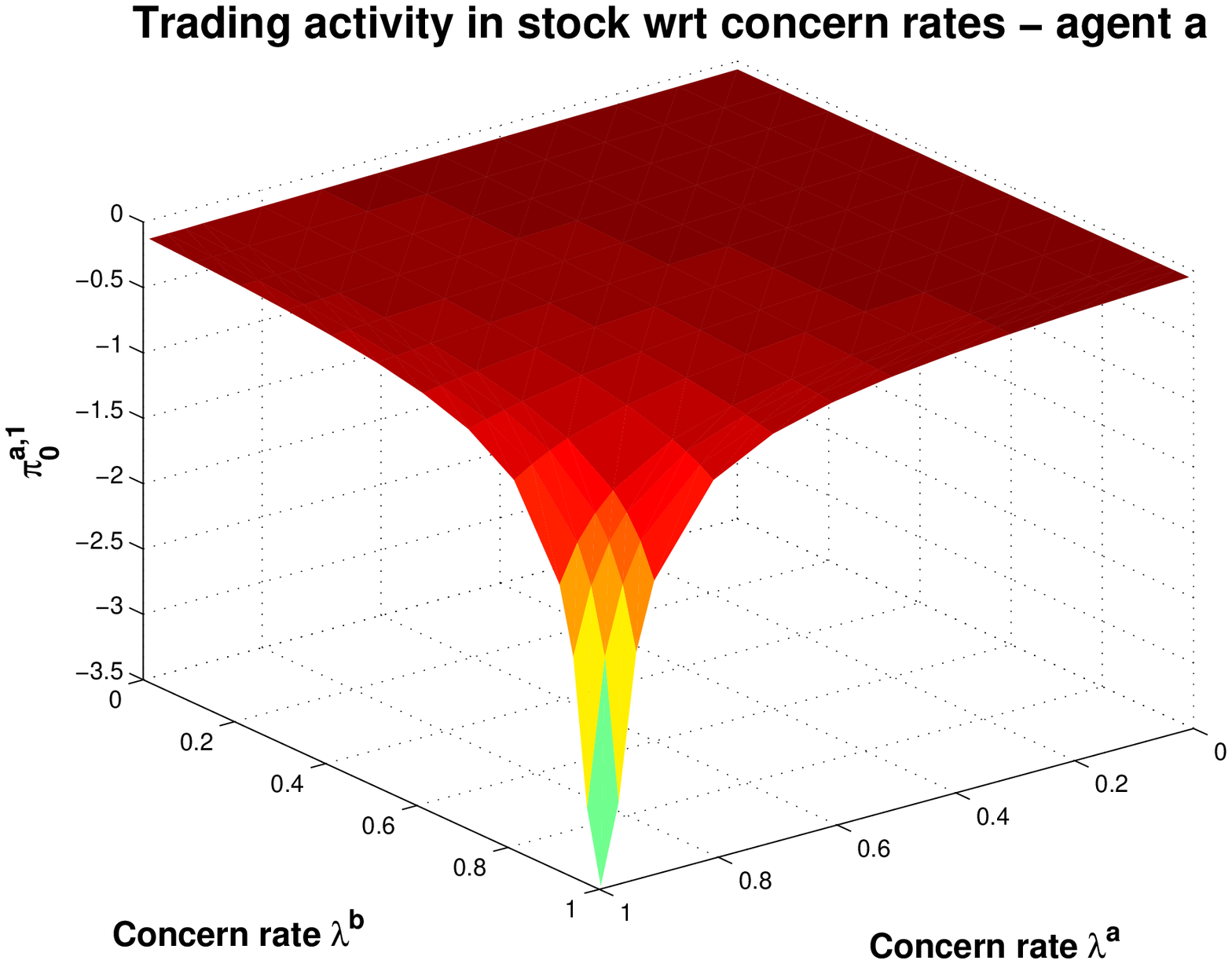}
	}
	\subfigure
		{
		\includegraphics[scale=0.34]{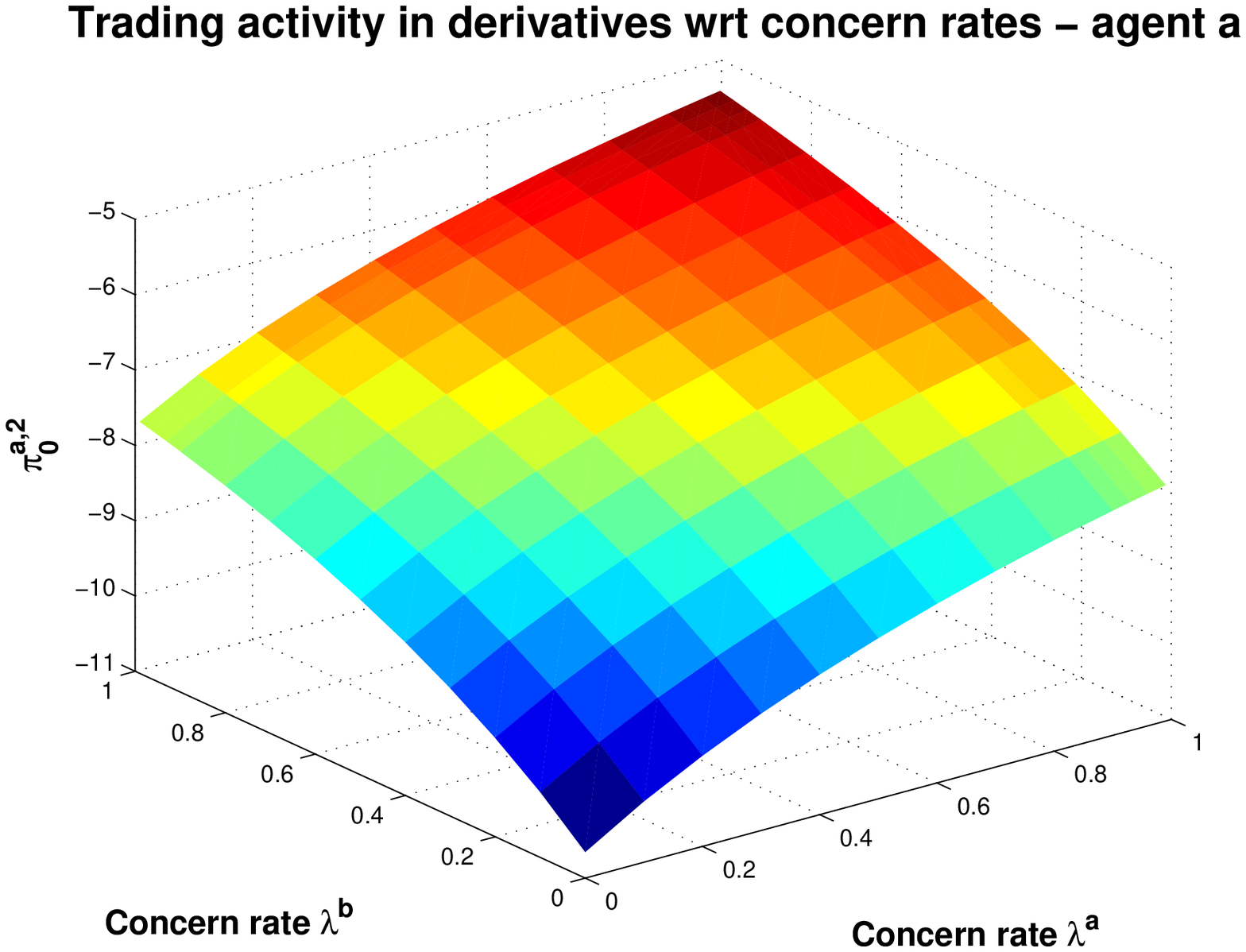}
	}
		\caption{Initial number $\pi^{a,1}_0$ and $\pi^{a,2}_0$ of shares of stock (left) and derivative (right) held by agent $a$, as a function of $(\lambda^a,\lambda^b)$. For visualization purposes the axes on the left picture were inverted.}
\label{fig:TradingActivityInStock}
\end{figure}

The observed behavior in Figure \ref{fig:TradingActivityInStock} is in line with the intuitive idea that the more the agents are concerned (high $\lambda^i$) with their relative performance $V^i_T-V^j_T$, $j \neq i \in \{a,b\}$ (recall \eqref{intro-perffunc}), the more they will invest in a way that neutralizes this source of risk. This is done by adopting a trading strategy that is as close as possible to that of the other agent. 

For the stock, we see from the formulas in Remark \ref{remark-strut-of-equilibrium}  that when $\lambda^a\lambda^b < 1$, the process of $a$ imitating $b$ imitating $a$, etc, results in a finite position. But the volume increases with both $\lambda^a$ and $\lambda^b$, and explodes as $(\lambda^a,\lambda^b) \rightarrow (1,1)$. In our example they would both (short-)sell infinitely many shares of the stock. Note that this is possible only  because the stock is assumed to be exogenously priced and perfectly liquid. For the derivative, they cannot imitate each other and position themselves in the same direction, as the zero net supply condition implies that the agents must hold exactly opposite positions. Agent $b$'s gains on trading the derivative will be exactly agent $a$'s losses. The only way to reduce the difference in performances for a very concerned agent is to engage less (in volume) in the trading of the derivative. The market clearing condition then forces the other agent to also trade less (in volume). This is  seen from the factor $1/(1+\lambda^i)$ in the formulas in Remark \ref{remark-strut-of-equilibrium} and is confirmed in Figure \ref{fig:TradingActivityInStock} (on the right) where agent $a$, identified as the seller, ends up selling fewer units of the derivative as either concern rate  increases. Due to the market clearing condition between the agents, no explosion is possible.

\subsubsection*{Price of the derivative}
Figure \ref{fig:derivatesurface01} shows an opposite dependence of the derivative's price $B^{\theta}_0$ on the concern rates $\lambda^a, \lambda^b$, a behavior not captured by Theorem \ref{theo:parambehavior}. 
One can make sense of this effect by having in mind Figure \ref{fig:TradingActivityInStock}. A higher $\lambda^a$ implies that agent $a$ wants to trade less and, as she is the seller, this drives the price up. Symmetrically, a higher $\lambda^b$ implies that agent $b$ wants to trade less and, as she is the buyer, this drives the price down. 
\begin{figure}[htb]
	\centering
	\subfigure
	{
		\includegraphics[scale=0.37]{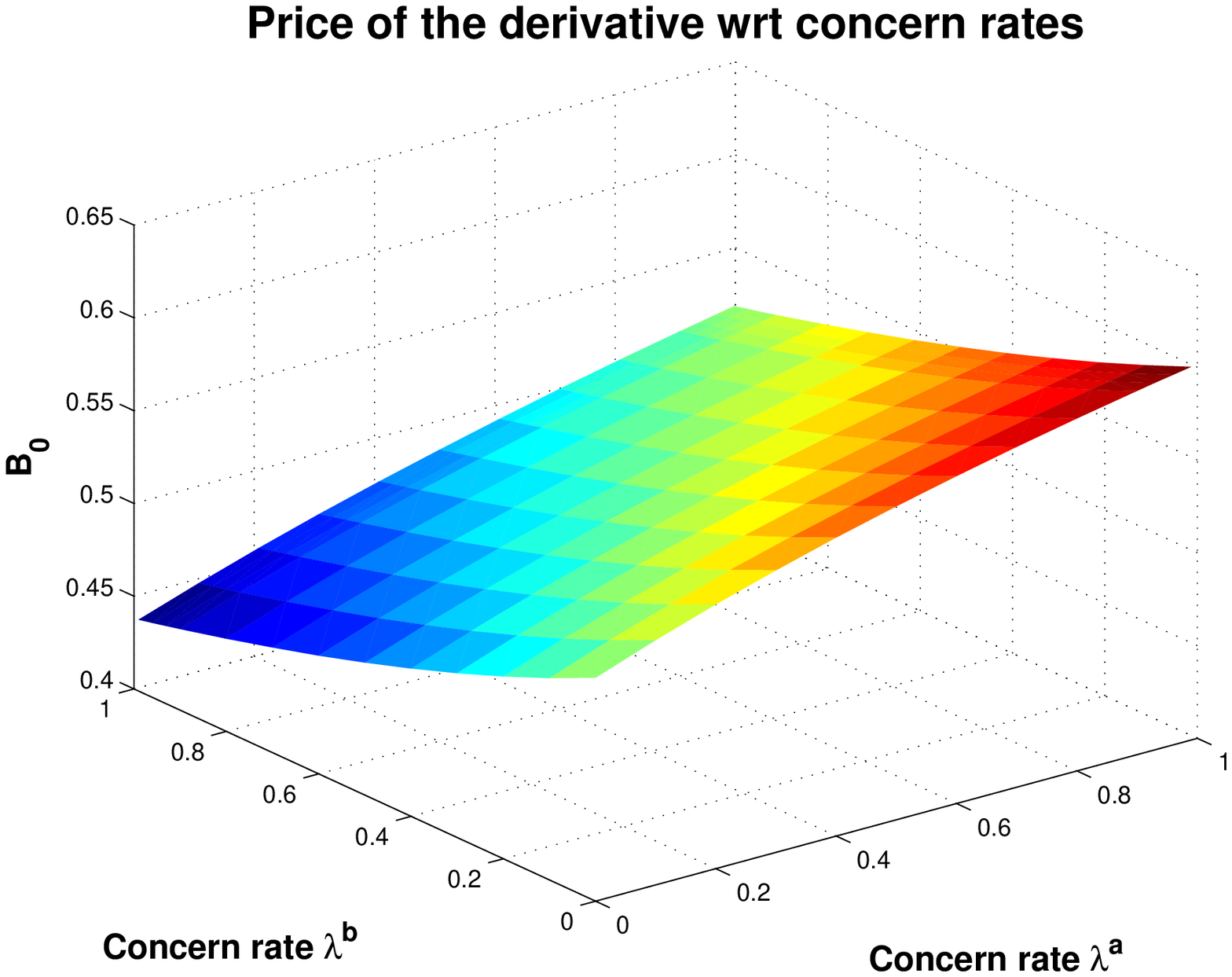}
	}
	\hspace{0.1cm}
	\subfigure
	{
		\includegraphics[scale=0.37]{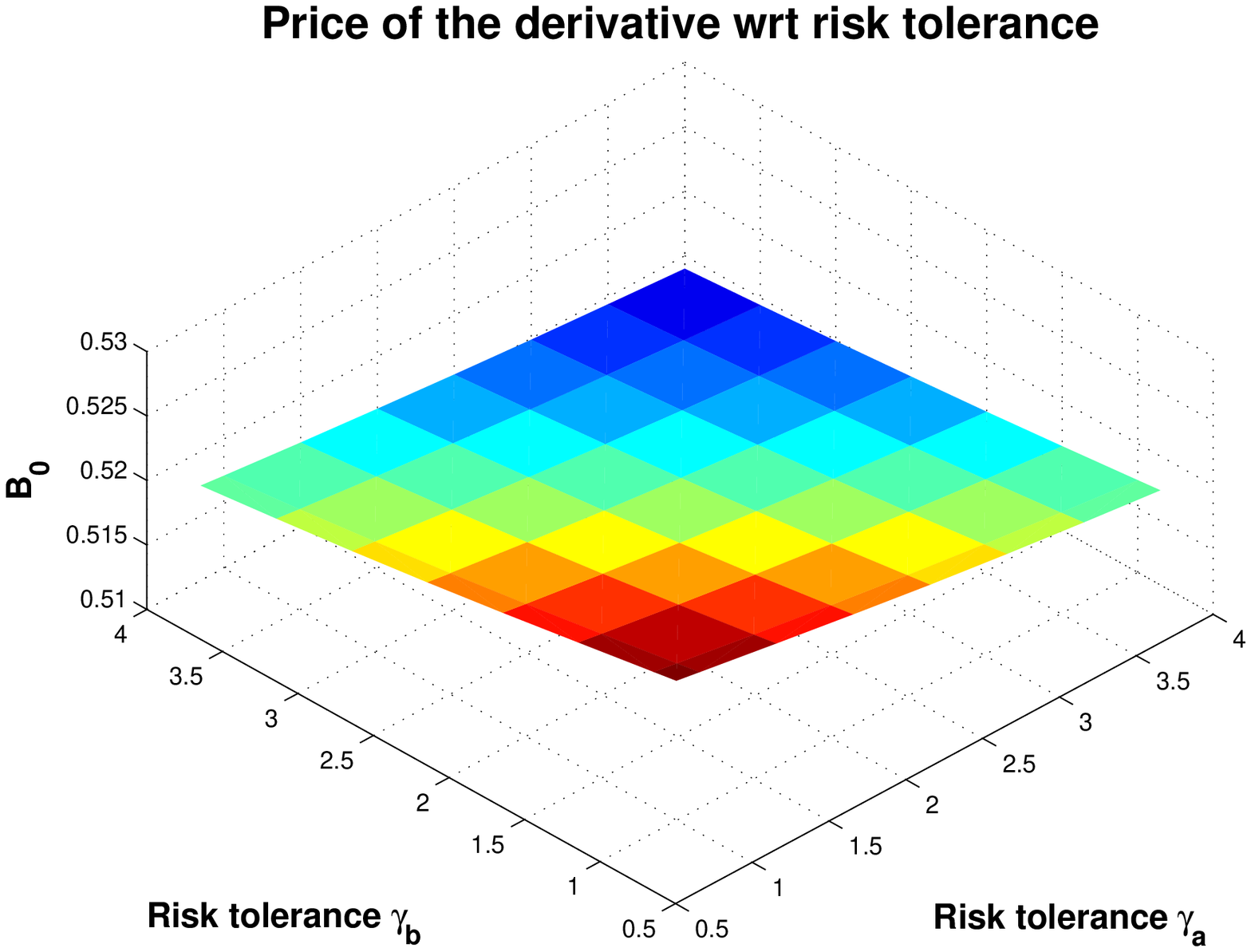}
	}
		\caption{Initial price of the derivative, $B^\theta_0$, as map of $(\lambda^a,\lambda^b)$ on the left, and as map of $(\gamma_a,\gamma_b)$ on the right.}
\label{fig:derivatesurface01}
\end{figure}
\subsubsection*{Aggregated risk}
Figure \ref{fig:weighted-risk-vs-lambdas} confirms the analytical results of Theorem \ref{theo:parambehavior}. First note that $\gamma_a=\gamma_b=1$ and so condition \eqref{eq:signalcondition} simplifies (see Corollary \ref{coro-parameter-behavior}). As predicted, the increase of the risk tolerances lead to a decrease in the aggregated risk (see Figure \ref{fig:weighted-risk-vs-lambdas}, left picture). The picture on the right shows clearly the cross behavior stated in \eqref{eq:crossbehaviorYw}. 
\begin{figure}[htb]
	\centering
	\subfigure
	{
		\includegraphics[scale=0.37]{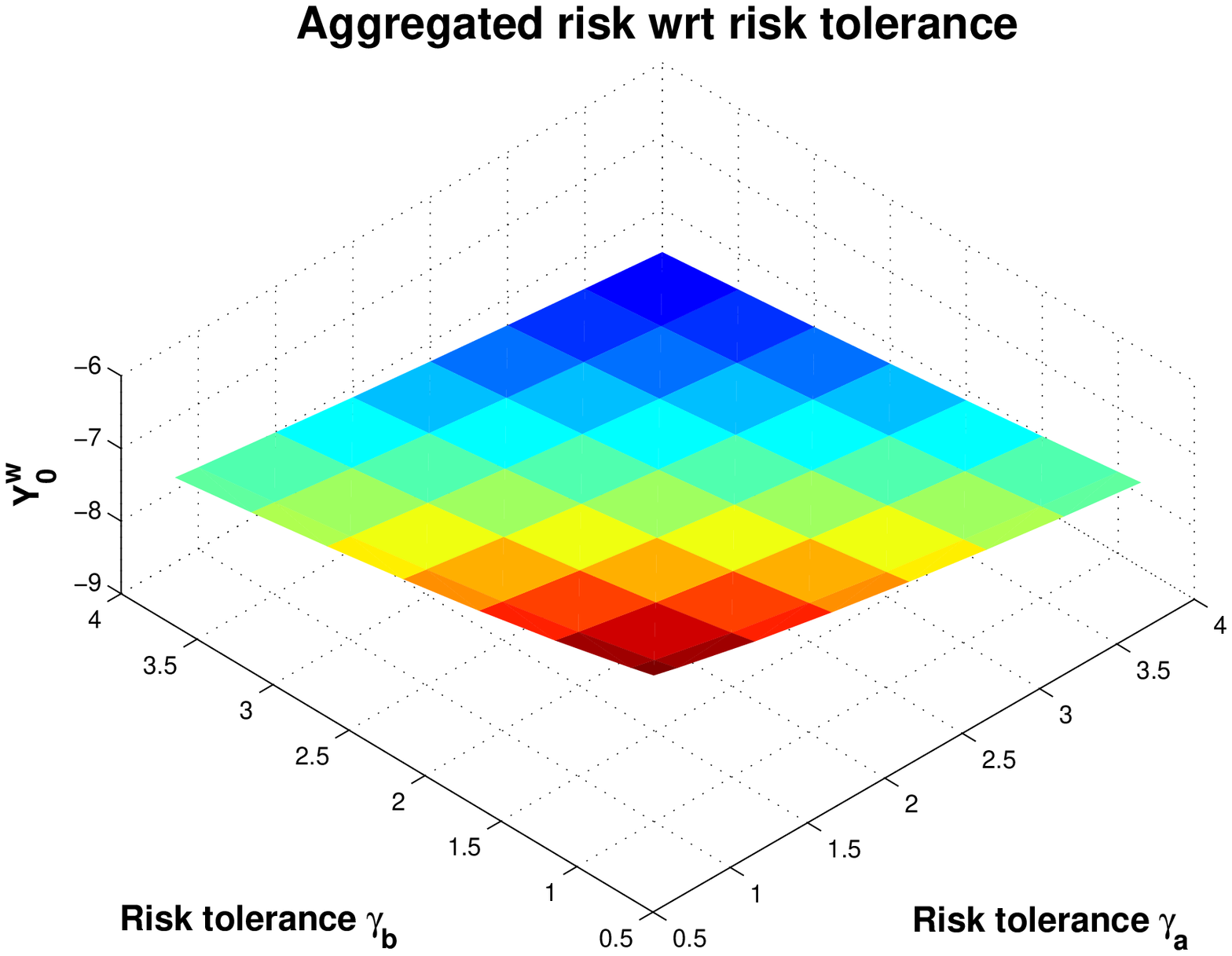}
	}
	\hspace{-0.2cm}
	\subfigure
	{
		\includegraphics[scale=0.37]{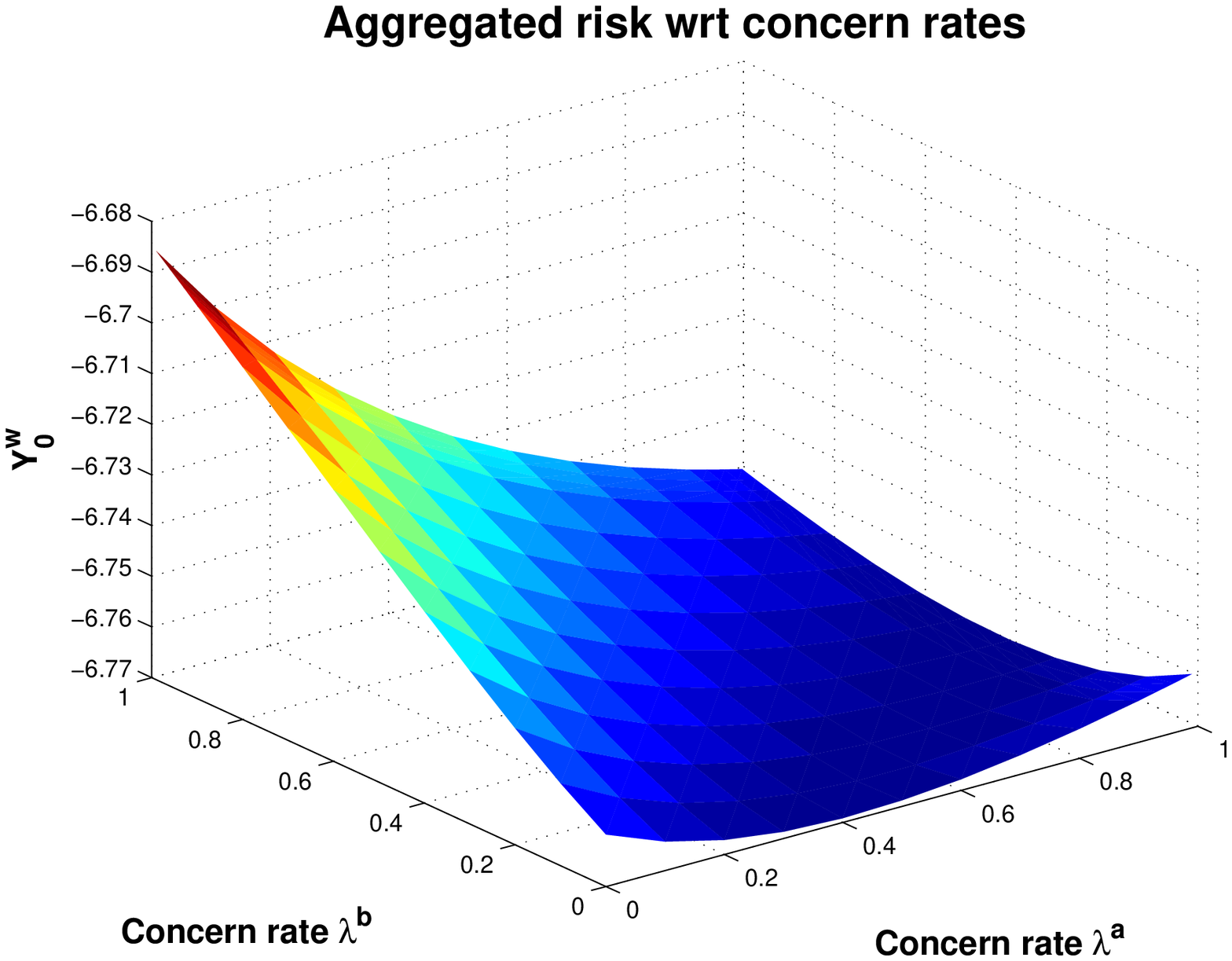} 
	}
	\caption{Aggregated risk $Y^w_0$ as a function of $(\gamma_a,\gamma_b)$ (left) and of $(\lambda^a,\lambda^b)$ (right).
	}
\label{fig:weighted-risk-vs-lambdas}
\end{figure}

\subsubsection*{Risk of each agent}
Theorem \ref{theo:parambehavior} does not capture the behavior of each agent's risk assessment as a function of the concern rates $\lambda^{\cdot}$.
\begin{figure}[htb]
	\centering
	\subfigure
	{
		\includegraphics[scale=0.39]{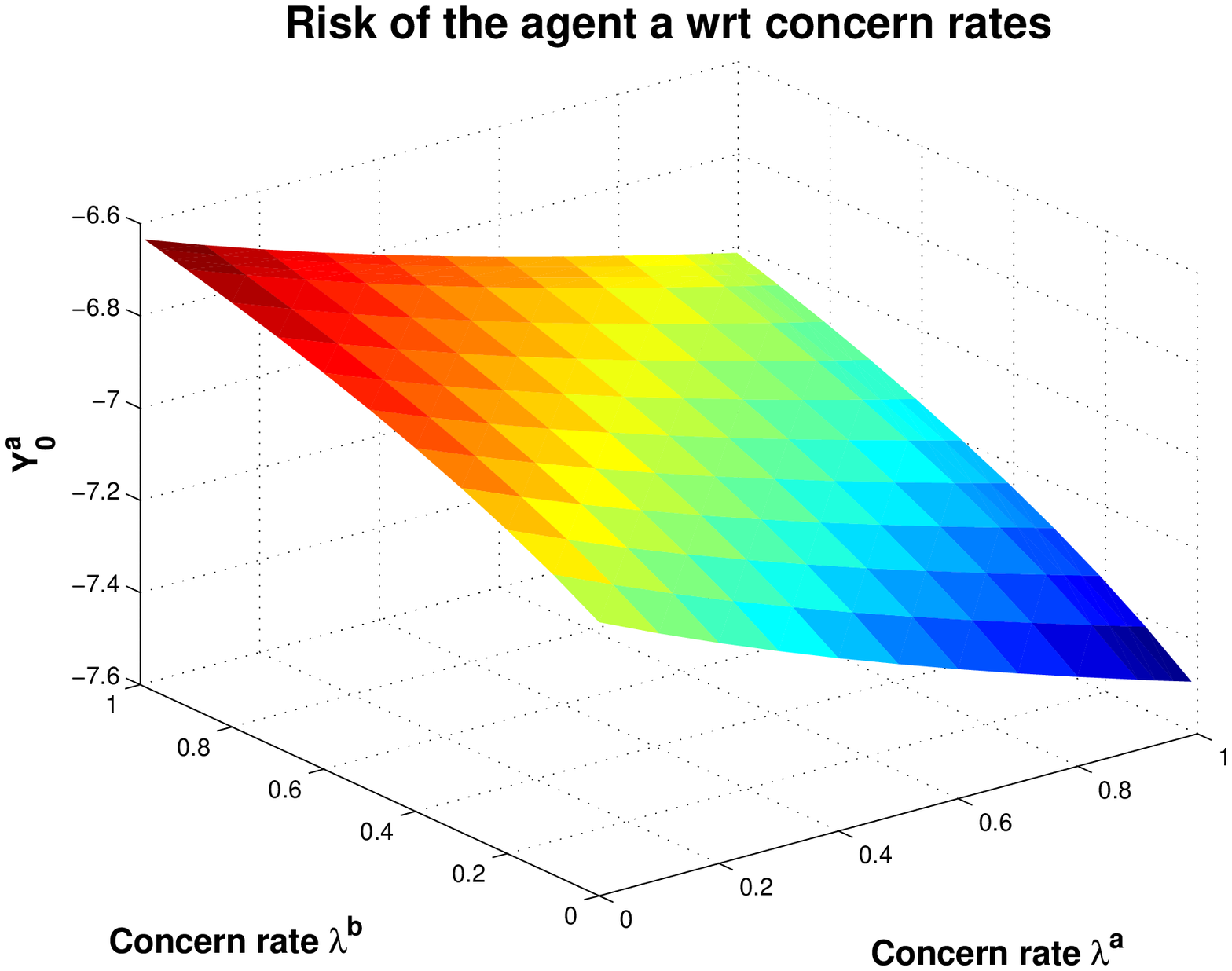}
	}
	\subfigure
		{
		\includegraphics[scale=0.39]{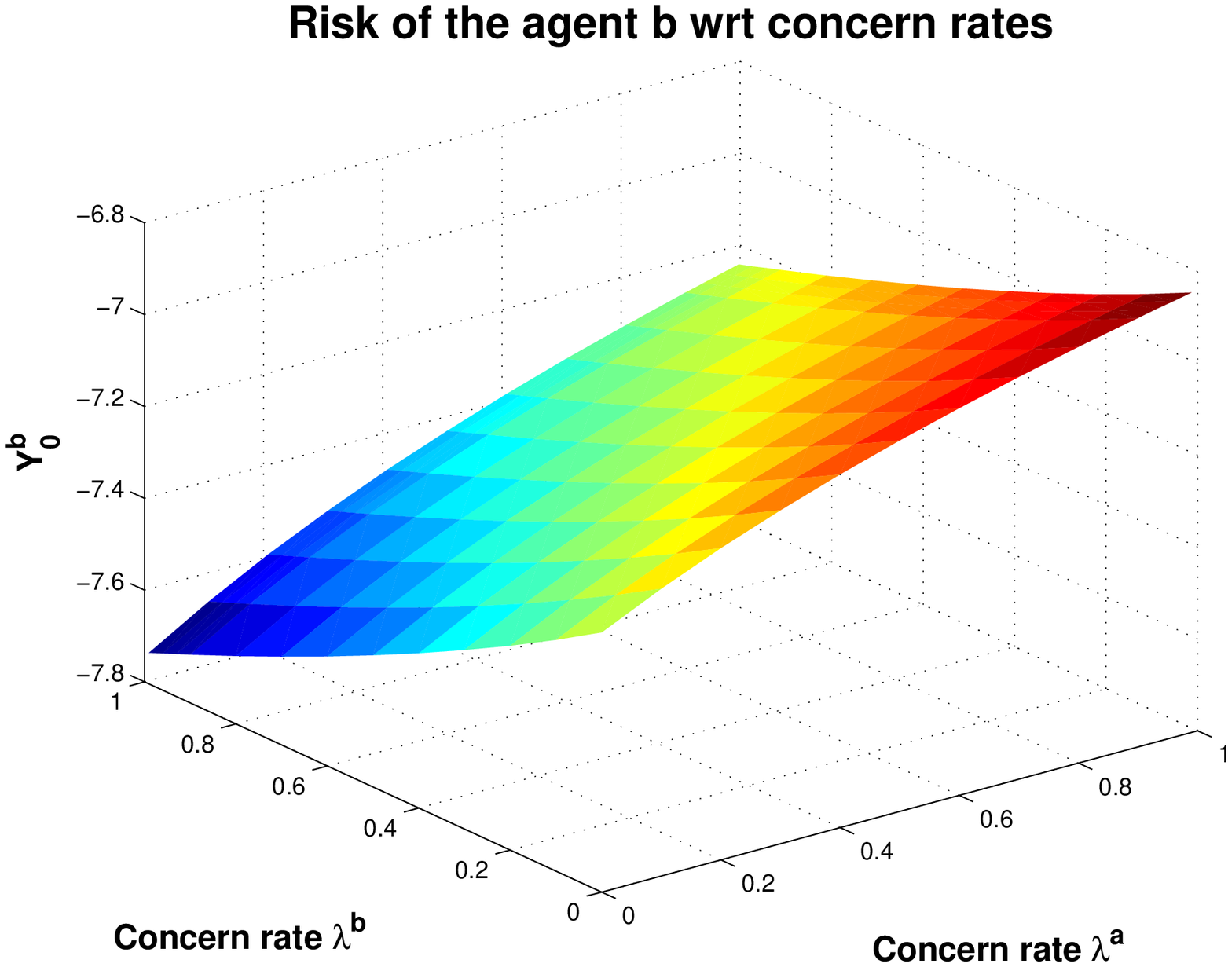}
	}
	\caption{Risk $Y^a_0$ (left) and $Y^b_0$ (right) as a function of $(\lambda^a,\lambda^b)$.}
	\label{fig:SingleRiskwrtLambda}
\end{figure}
Figure \ref{fig:SingleRiskwrtLambda} portrays the risk perceptions of each agent as $\lambda^a, \lambda^b$ change.  Agent $a$'s risk $Y^a_0$ increases in $\lambda^b$ and decreases in $\lambda^a$. 
A possible explanation for the latter behavior ($Y^a_0$ decreases with $\lambda^a$) from the perspective of, say $a$, and having \eqref{intro-perffunc} or \eqref{BSDE_1} in mind is as follows. If $a$ gives more importance to her relative performance concern then she weighs the term $V^a_T-V^b_T$ more than the hedging of the random endowment or the optimization of the personal performance and trades in a way that mimics more of what $b$ does. The net result of this seems to be the ability to neutralize more of the performance risk (as a fluctuation around the mean) and less ability to neutralize the endowment risk. The former apparently carries more weight as $Y^a_0$ does indeed decrease with $\lambda^a$.

The explanation of the first behavior ($Y^a_0$ increases with $\lambda^b$) seems more direct. As $\lambda^b$ increases, agent $b$ engages in less trading of the derivative in order to reduce her relative performance concern, and this affects agents $a$, in particular her ability to hedge $H^a$.

\subsection{Effect of introducing the derivative}		

We now comment on the effects of introducing the derivative in this model market. 
Figure \ref{fig:ComparingRisksWithoutAndWithDerivatives} displays the risks of the representative agent and of agent $a$ with respect to $\lambda^a$ and $\lambda^b$ when no derivative is available and when a market for it is available.
\begin{figure}[htb]
	\centering
	\subfigure
	{
		\includegraphics[scale=0.40]{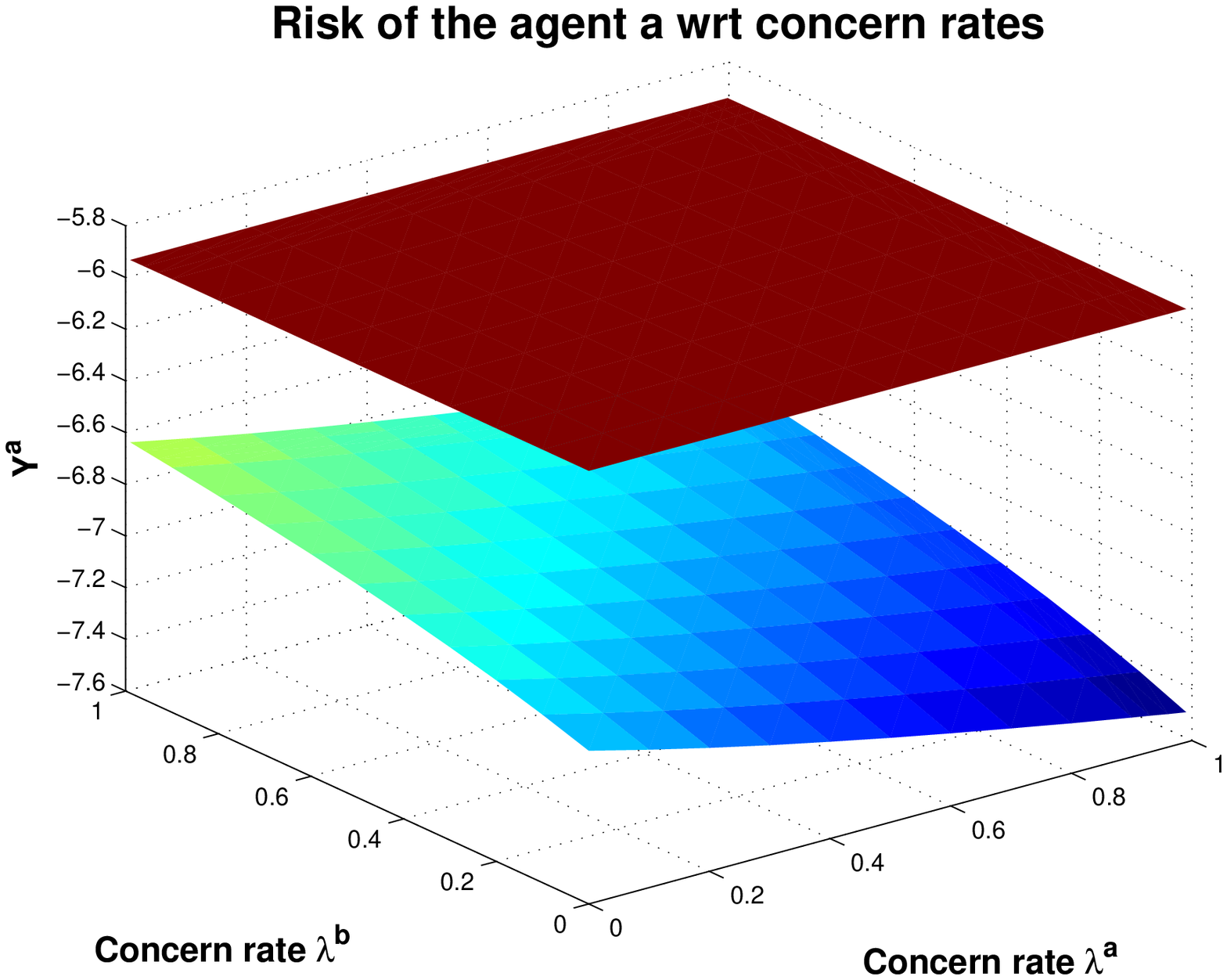}
	}
	\subfigure
		{
		\includegraphics[scale=0.40]{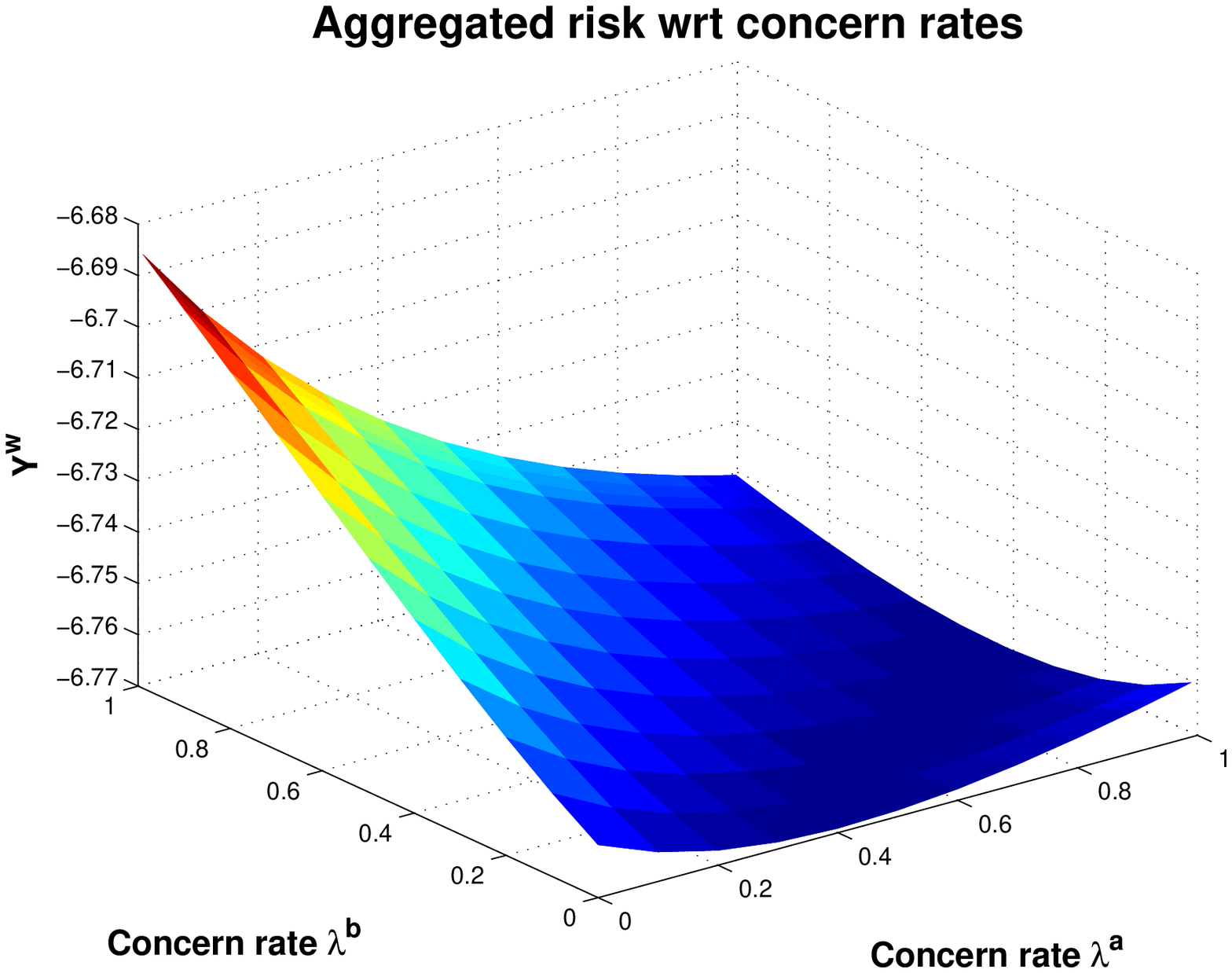}
	}
	\caption{Left : risk $Y^a_0$ when the derivative is not available (flat surface) and when it is (tilted surface), as a function of the concern rates $(\lambda^a,\lambda^b)$. 
				Right: same plot for the aggregated risk $Y^w_0$ (the two surfaces are equal).}
\label{fig:ComparingRisksWithoutAndWithDerivatives}
\end{figure}

We observe in the plot on the right that adding the derivative does not change the aggregated risk. This is clear if one views it as the risk of the representative agent: being alone by construction, the zero net supply condition means that she must keep a zero position in the derivative, and hence does not benefit from its presence (compare the agent of Section \ref{subsubsection-example-single.agent.no.derivative} with Example \ref{example_EntropicOptimizedDriver}).  

For an individual agent however (left plot), the availability of the derivative always leads to a reduction of risk. We observe that in the absence of the derivative, the risk of agent $a$ does not depend on the concern rates. We can apply the methodology of Sections \ref{section-SingleAgentSection} and \ref{sec4-repagent} to find that the optimal porfolios of the agents in this situation are given by 
	\begin{align*}
		\pi^{*,i,1}_t = 		\frac{1}{1-\lambda^a\lambda^b} \ \frac{\gamma_i \theta^S + \wt{Z}^{i,1}_t}{\sigma^S S_t} 
								+	\frac{\lambda^i}{1-\lambda^a\lambda^b} \ \frac{\gamma_j \theta^S + \wt{Z}^{j,1}_t}{\sigma^S S_t}
		\quad\text{for}\ j \neq i \in \{a,b\},
	\end{align*}
while the minimized risk equation is given by the BSDE
	\begin{align*}
		d\wt{Y}^i_t = - \Big[ -\frac12 \gamma_i \big(\theta^S\big)^2 - \wt{Z}^{i,1}_t \theta^S + \frac{1}{2\gamma_i}\big(\wt{Z}^{i,2}\big)^2 \Big] \udt + \langle \wt{Z}^{i}_t ,{\ud W_t} \rangle 
		\quad\text{with}\quad \wt{Y}^i_T = - H^i(S_T,R_T) .
	\end{align*}
This shows analytically that the value of the problem, $Y^a_0$, depends on neither $\lambda^a$ nor $\lambda^b$ while the optimal strategy does, as was already observed in Proposition 4.1 in \cite{FreiDosReis2011}.

\subsubsection*{Playing the game repeatedly leads to disaster.}

The above study considers a one-period model with (continuous-time) trading until the horizon $T=1$ month. Imagine now the repetition of this trading period over time and assume no significant changes to the agents' endowments or the dynamics of the financial and external risks. 

At the level of the agents' preferences, with the sole exception of the concern rates, they do not change with time. Specifically, we assume that their risk tolerances, and consequently the entropic risk measures $\rho^\cdot_0$ used to assess their risk in \eqref{intro-perffunc}, are fixed throughout; however, their concern rates $\lambda^\cdot$ over their relative performance may vary. This can account for some herding or other behavioral mechanism: after each period, each agent can review the results of everyone's performance, carry this information into the next period and update their concern rate accordingly.

Figure \ref{fig:SingleRiskwrtLambda} sheds some light on the outcome of playing this game repeatedly. Indeed, each agent benefits from a unilateral increase of their concern rate $\lambda$ while they are worse off with an increase in the other's concern rate. So they have an incentive to increase $\lambda$, as the trading periods are repeated, culminating in Assumption \ref{assump:ConcernRates} being violated as $(\lambda^a,\lambda^b) \rightarrow (1,1)$. 

It is interesting to note that this drifting toward the singularity of the model, $(\lambda^a,\lambda^b) = (1,1)$, is not captured by the risk assessments. Figures \ref{fig:weighted-risk-vs-lambdas} and \ref{fig:SingleRiskwrtLambda} show that $Y^w_0$, $Y^a_0$ and $Y^b_0$ remain bounded. At the level of the investment strategies, the trading activity in the derivative slows down but persists. The sharing of the external risk becomes less efficient, because the agents are increasingly concerned about losing out to the other, but does not disappear. However, the investment in the stock explodes (see Figure \ref{fig:TradingActivityInStock}). We stress that this behavior arises only \emph{after} the derivative is introduced in the market. Indeed, as shown by Figure \ref{fig:ComparingRisksWithoutAndWithDerivatives}, when the derivative is not available and the agents in $\bA$ are only concerned with the relative performance of their strategy \emph{over the market}, they have no incentive to having increasingly high concern rates. The particular shape of the surface $(\lambda^a,\lambda^b) \mapsto Y^i_0$, risk decreasing with $\lambda^i$ but increasing with $\lambda^j$, appears only when the derivative is made available. In this situation, the agents are placed in direct interaction (by trading) in addition to the indirect one (social): each agent makes now gains directly \emph{over the other}. The final result is a potential destabilization of the stock market.

%
%
%
%
%
\section{Conclusion}

In this work, we analyzed the effect of a form of social interaction between agents on an equilibrium pricing mechanism. Specifically, we considered the pricing of a (market-completing) derivative introduced to allow market participants to share the risk associated with an external and non-tradable risk factor. The social interaction here takes the form of concerns over relative performance.

From a theoretical point of view, we have shown how to solve the problem for general risk measures and a finite number of agents, when assuming that the derivative completes the market. This involves solving a coupled system of quadratic BSDEs. Due to the heterogeneous rates of concerns of the agents, the risks of the agents cannot be aggregated by the usual infimal convolution technique, so we developed it further and introduced the \emph{weighted-dilated infimal convolution} variant. 

We then focused on the particular case of the entropic risk measure and were able to determine sufficient conditions to design a derivative that completes the market. In a market model with two agents representing opposite profiles of exposure to the external risk, we explored the impact of the social interactions on the benefit brought by financial innovation.

We found that the introduction of the derivative always reduces the risk, at the level of individual agents. However, the particular distribution of this risk reduction means that both agents have an incentive to become more concerned with their relative performance. At the global level, while this merely decreases the volume of derivatives exchanged, this leads to an explosion of the volumes traded in the previously-existing financial asset. In practice, the assumption that the agents are small and that the price dynamics of the stock is independent of their actions fails to hold. Thus, although the stock price is fundamentally independent of the external risk, introducing the derivative can lead to unintended consequences on what was a stable stock market. We stress that this phenomenon is not captured by the risk measures. Therefore one should not only use the performance of the risk measure when evaluating the possible benefits of a new policy (the introduction of the derivative, here). This also stresses the importance of having a  systemic view: studying the problem from the point of view of an individual investor shows that the availability of the derivative is always beneficial, but at the global level the picture has strong nuances. Strongly undesirable endogenous phenomena can emerge in the dynamics, arising essentially from the interaction between the various agents and their possibility to adapt to the new policy.
%
%
%
%
%
%
%
%
\appendix

\section{Stochastic analysis: notation, spaces and base results}
\label{appendix-NotionStochAna}

\subsection*{Spaces \& Notations}
We define the following spaces for $p> 1$, $q\geq 1$, $n,m,d,k\in \bN$: 
 $C^{0,n}([0,T]\times \bR^d,\bR^k)$ is the space of continuous functions
endowed with the $\lVert \cdot \rVert_\infty$-norm  
that are $n$-times continuously differentiable in the spatial variable; 
$C^{0,n}_b$ contains all bounded functions of $C^{0,n}$; 
the first superscript $0$ is dropped for
functions independent of time; $L^{p}(\cF_t,\bR^d)$,  $t\in[0,T]$, is the
space of $d$-dimensional $\cF_t$-measurable random variables $X$ with norm  
$\|X\|_{L^p} = \bE[\, |X|^p]^{1/p} < \infty$;
$L^\infty$ refers to the subset of essentially bounded random variables;
$\cS^{p}([0,T]\times \bR^d)$ is the space of $d$-dimensional measurable
$\cF$-adapted processes $Y$ satisfying $\|Y \|_{\cS^p} =
\bE[\sup_{t\in[0,T]}|Y_t|^p]^{1/p}
<\infty$; $\cS^\infty$ refers to the subset of $\cS^{p}(\bR^d)$ of essentially bounded processes; $\cH^{p}([0,T]\times \bR^{d})$ is
the space of $d$-dimensional
measurable $\cF$-adapted processes $Z$ satisfying $\|Z
\|_{\cH^p}=\bE[\big(\int_0^T|Z_s|^2 \uds\big)^{p/2}]^{1/p}<\infty$; 
For a probability measure $\bQ$, we denote $\cHBMO(\bQ)$ as the space of processes $Z\in \cH^p(\bQ)$ for any $p\geq 2$ such that for some constant $K_{BMO}>0$
\begin{align*}
\sup_{\tau\in\cT_{[0,T]}} 
\big\lVert 
     \bE^\bQ\big[ \int_\tau^T |Z_s|^2\uds \big|\cF_\tau \big]
           \big\rVert_\infty \leq K_{BMO}<\infty,
\end{align*}
where $\cT_{[0,T]}$ is the set of all stopping times $\tau\in[0,T]$. As an easy consequence, if $Z \in \cHBMO(\bQ)$, then $\int H \ud Z\in \cHBMO(\bQ)$ for any bounded adapted process $H$. Processes in $\cHBMO$ have very convenient properties. For the reference measure $\bP$ we write directly $\cHBMO$ instead of $\cHBMO(\bP)$ .

For more information on BMO spaces and their relation with BSDEs see Section 2.3  in \cite{ImkellerDosReis2010} or Section 10.1 in \cite{Touzi2013}; we state, for reference's sake, some of them in the next result. 
\begin{lemma}
\label{lemma-BMOproperties}
Let $Z\in\cHBMO$ and define $\Phi_\cdot:=\int_0^\cdot Z_s\udws$. Then we have:
\begin{itemize}
  \item[1)] The stochastic exponential $\cE(\Phi_T)$ is uniformly
  integrable.
  \item[2)] There exists a number $r>1$ such that $\cE(\Phi_T)\in L^r$. This property follows from the \emph{Reverse H\"older inequality}.
  The maximal $r$ with this property can be expressed explicitly in terms of the BMO norm of $\Phi_\cdot$. There exists as well an upper bound for $\| \cE(\Phi_T) \|_{L^r}^r$ depending only on $T$, $r$ and the BMO norm of $\Phi$.
\end{itemize}
\end{lemma}

\subsection{Basics of Malliavin's calculus}
\label{sec-MalliavinCalculus}
We briefly introduce the main notation of the stochastic calculus of
variations also known as Malliavin's calculus. For more details, we
refer the reader to \cite{nualart2006}, for its application to BSDEs we refer
to \cite{Imkeller2008}. Let ${\bf \cS}$ be the space
of random variables of the form
\[
\xi = F\Big((\int_0^T h^{1,i}_s \ud W^1_s)_{1\le i\le
n},\cdots,(\int_0^T h^{d,i}_s \ud W^d_s)_{1\le i\le n}\Big), 
\] where $F\in C_b^\infty(\bR^{n\times d})$, $h^1,\cdots,h^n\in L^2([0,T];
\bR^d)$, $n\in\bN.$
To simplify notation, assume that all $h^j$ are written as row
vectors.
For $\xi\in {\bf \cS}$, we define $D = (D^1,\cdots, D^d):{\bf \cS}\to
L^2(\Omega\times[0,T])^d$ by
\[
D^i_\theta \xi = \sum_{j=1}^n \frac{\partial F}{\partial x_{i,j}}
\Big( \int_0^T h^1_t \ud W_t,\ldots,\int_0^T
h^n_t\ud W_t\Big)h^{i,j}_\theta,\quad 0\leq \theta\leq T,\quad 1\le
i\le d,\]
and for $k\in\bN$ its $k$-fold iteration by
$D^{(k)} = (D^{i_1}\cdots D^{i_k})_{1\le i_1,\cdots, i_k\le d}$. For
$k\in\bN,\, p\ge 1$ let $\bD^{k,p}$ be the closure of $\cS$ with respect to
the norm
\[
\lVert \xi \lVert_{k,p}^p 
= \bE\Big[\|\xi\|^p_{L^p}
			+ \sum_{i=1}^{k}\|| D^{(k)]} \xi| \|_{(\cH^p)^i}^p\Big].
\]
$D^{(k)}$ is a closed linear operator on the space $\mathbb{D}^{k,p}$.
Observe that if $\xi\in \bD^{1,2}$ is $\cF_t$-measurable, then $D_\theta
\xi=0$ for $\theta \in (t,T]$. Further denote 
$\mathbb{D}^{k, \infty}=\cap_{p>1}\mathbb{D}^{k,p}$. We also need Malliavin calculus for $\bR^m$-valued smooth stochastic processes. For $k\in\bN, p\ge 1,$ denote by $\mathbb{L}^{k,p}(\bR^m)$ the set
of $\bR^m$-valued progressively measurable processes $u = (u^1,\cdots, u^m)$
on $[0,T]\times \Omega$ such that
\begin{itemize}
\item[i)]
for Lebesgue-a.a. $t\in[0,T]$, $u(t,\cdot)\in(\mathbb{D}^{k,p})^m$;
\item[ii)] $[0,T]\times \Omega \ni (t,\omega)\mapsto D^{(k)} u(t,\omega)\in
(L^2([0,T]^{1+k}))^{d\times n}$ admits a progressively measurable
version;
\item[iii)] $\lVert u \lVert_{k,p}^p
=  \|u\|_{\cH^p}^p+\sum_{i=1}^{k} \| \,D^i
u\,\|_{(\cH^p)^{1+i}}^p\,<\infty$.
\end{itemize}

Note that Jensen's inequality gives\footnote{The reason behind this last
inequality is that within the BSDE framework the
usual tools to obtain a priori estimates yield with much difficulty the LHS
while with relative ease the RHS.} for all $p\geq 2$
\begin{align*}
\bE\Big[\Big( 
\int_0^T \int_0^T|D_u X_t|^2\ud u\,\ud t
             \Big)^{\frac{p}{2}} \Big]
\leq
T^{p/2-1}\int_0^T  \| D_u X\|_{\cH^p}^p
\ud u. 
\end{align*}
We recall a result from \cite{Imkeller2008} concerning the rule for the
Malliavin differentiation of It\^o integrals which is of use in
applications of Malliavin's calculus to stochastic analysis.
\begin{theorem}[Theorem 2.3.4 in \cite{Imkeller2008}]
Let $(X_t)_{t\in[0,T]}\in \cH^2$ be an adapted process 
and define $M_t:=\int_0^t X_r\udwr$ for $t\in[0,T]$. Then, 
$X\in\bL^{1,2}$ if and only if $M_t\in\bD^{1,2}$ for any $t\in[0,T]$.

Moreover, for any $0\leq s, t\leq T$ we have
$D_s M_t 
= X_s\1_{\{s\leq t\}}(s) 
              + \1_{\{s\leq t\}}(s)\int_s^t D_s X_r \udwr$.
\end{theorem}


\subsection{Basic Malliavin calculus results for SDEs}

With relation to the Brownian motions $W^R$ and $W^S$, we denote the Malliavin differential operators $D^{W^R}$ and $D^{W^S}$, see Appendix \ref{sec-MalliavinCalculus}. 
\begin{proposition}
\label{prop-SDEresults}
Let Assumption \ref{assump-DiffData} hold. Then SDEs 
\eqref{diffusion-R} and \eqref{diffusion-S} have a unique solution 
$R,S\in \cS^p$ for any $p\geq 2$ and 
\begin{enumerate}[i)]
	\item 
$R,S\in \bD^ {1,2}$. 
	We have $D^{W^S}_u R_t = D^{W^R}_u S_t =0$ for any 
$t,u\in [0,T]$ as well as 
\begin{align}
\label{eq:MallderividentityforR}
D^{W^R}_u R_t = \1_{\{u\leq t\}} b
\qquad \text{and}\qquad
D^{W^S}_u S_t = \1_{\{u\leq t\}} \sigma^S S_t,\qquad t,u\in [0,T];
\end{align}
\item For any jointly measurable function $\psi:[0,T]\times\bR\times\bR\to\bR$ that is Lipschitz (in the second space variable), it holds that 
\begin{align}
 \label{eq:mallderivmatch}
	D^{W^R}_u \big(\psi(t,S_t,R_t)\big) = D^{W^R}_r \big(\psi(t,S_t,R_t)\big)\qquad \forall u,r\in[0,t],\quad t\in[0,T].
\end{align}
Furthermore, $\Big(D^{W^R}_0 \big(\psi(\cdot,S_\cdot,R_\cdot)\big)\Big)\in \cS^\infty$.
	\item $H^D,H^a\in \bL^{1,2} \cap \cS^\infty$ for any $a\in \bA$ 
	(recall 
	\eqref{agents-payoff-Ha-bond-payoff-HD}) and there exists $M>0$ for any $0\leq r,u\leq T$ and any $\zeta\in \bA\cup\{D\}$ such that $D^{W^R}_u H^\zeta = D^{W^R}_r H^\zeta$ and 
$0<|D^{W^R}_\cdot H^\zeta| \leq M$. 

\item Let $\zeta\in \bA\cup\{D\}$ and let $r_0\in \bR$. The mapping $r_0 \mapsto (D^{W^R}_u H^\zeta)$ is Lipschitz continuous uniformly in $u\in [0,T]$ for any $s_0\in (0,+\infty)$.
\end{enumerate}
\end{proposition}
\begin{proof}
Throughout let $\zeta\in \bA\cup\{D\}$. 
General results on SDEs follow from e.g. Section 2 in \cite{ImkellerDosReis2010}, standard Malliavin calculus, the fact that $S$ is a Geometric Brownian motion and $\mu^R\in C([0,T],\bR)$.

\emph{Proof of i)} 
The identity $D^{W^S}_u R_t = D^{W^R}_u S_t =0$ is trivial. 

\emph{Proof of ii)}
We prove \eqref{eq:mallderivmatch}: assume $\psi$ to be differentiable, then for $u,r\in[0,t]$
\begin{align*} 
D^{W^R}_u \big(\psi(t,S_t,R_t)\big) 
& 
= (\partial_{x_2} \psi)(t,S_t,R_t) b 
= D^{W^R}_r \big(\psi(t,S_t,R_t)\big) , 
\end{align*}
where we used \eqref{eq:MallderividentityforR}. Now a standard approximation by mollification delivers the two results.

\emph{Proof of iii)} The form of the $\cF_T$-measurable payoffs $H^D,H^a$ is quite specific and it is clear that for $0\leq u \leq T$ and $\zeta\in \bA\cup\{D\}$
\begin{align}
D_u^{W^R} H^\zeta 
&
= D_u^{W^R} \big(h^\zeta(S_T,R_T)\big)
= \big\langle (\nabla h^\zeta)(S_T,R_T),(0,\1_{\{u\leq T\}} b)\big\rangle
\label{eq:DWRH*}
= b(\partial_{x_2} h^\zeta)(S_T,R_T)
\end{align}
The boundedness of $D_\cdot^{W^R} H^\zeta$ follows from the uniform boundedness of the  derivatives of $h^\zeta \in C^2_b$. We can then conclude that if $\partial_{x_2} h^\zeta\neq 0$ then it follows that $D_\cdot^{W^R} H^\zeta \neq 0$ and, moreover, the identity $D^{W^R}_u H^\zeta = D^{W^R}_r H^\zeta$ follows from \eqref{eq:mallderivmatch}.

\emph{Proof of iv)}
We now close with the proof of the last statement. Take $s_0\in (0,+\infty)$ and let $r_0,\wt r_0\in\bR$ be two initial conditions for $R$ (see \eqref{diffusion-R}) and we denote the corresponding SDE solutions $R$ and $\wt R$ respectively. We also denote $H^\zeta$ and $\wt H^\zeta$ the random variables depending on $R$ and $\wt R$ respectively. Due to the linear form of \eqref{diffusion-R} it is immediate that $R_t-\wt R_t=r_0-\wt r_0$ for any $t\in[0,T]$.

The properties of $|D_u^{W^R} H^\zeta -D_u^{W^R} \wt H^\zeta|$ follow from those of $\partial_{x_2} h^\zeta$ and \eqref{eq:DWRH*}. By assumption $h^\zeta$ is twice continuously differentiable (in space)  with bounded derivatives, hence, for some $K\geq 0$ 
\begin{align*}  
\big|(\partial_{x_2} h^\zeta)(S_T,R_T) - (\partial_{x_2} h^\zeta)(S_T,\wt R_T)\big|
\leq
K |R_T-\wt R_T|= K |r_0-\wt r_0|.
\end{align*}
It follows that for some constant $C\geq 0$ independent of the data $u,s_0,r_0$ and $\wt r_0$ one has, as required, $|D_u^{W^R} H^\zeta -D_u^{W^R} \wt H^\zeta|\leq C|r_0-\wt r_0|$.
\end{proof}
%
%
%
%
%
%
%


\bibliography{PerfConcerns}  

\end{document}